\documentclass[aps,prb,reprint,amsmath,amssymb,amsfonts,superscriptaddress,longbibliography,nofootinbib]{revtex4-2}
\makeatletter
\newcommand*{\rom}[1]{\expandafter\@slowromancap\romannumeral #1@}
\makeatother

\usepackage[T1]{fontenc}
\usepackage[colorlinks=true, allcolors=blue]{hyperref}
\usepackage{lipsum}
\usepackage{verbatim}
\usepackage{enumitem}
\usepackage{stmaryrd}
\usepackage{xspace}

\usepackage{bm}
\usepackage{amsthm}
\usepackage{array}
\usepackage{mathtools}
\DeclarePairedDelimiter\ceil{\lceil}{\rceil}

\DeclareMathOperator{\diag}{diag}

\usepackage{physics}
\usepackage{siunitx}

\usepackage[table,svgnames,dvipsnames]{xcolor}
\usepackage{graphicx}
\usepackage{tikz}
\usepackage{quantikz}
\usetikzlibrary{quantikz2}
\usepackage{color}

\usepackage{dcolumn}
\usepackage{tcolorbox}
\usepackage{booktabs}
\usepackage{multirow}
\usepackage{makecell}

\usepackage{subfigure}
\usepackage{epsfig}

\usepackage{orcidlink}

\usepackage{cleveref}
\crefname{equation}{Eq.}{Eqs.}
\crefname{section}{Sec.}{Secs.}
\crefname{appendix}{Appendix}{Appendices}
\crefname{figure}{Fig.}{Figs.}
\crefname{table}{Table}{Tables}
\crefname{theorem}{Theorem}{Theorems}

\newtheorem{theorem}{Theorem}
\newtheorem{lemma}[theorem]{Lemma}
\newtheorem{definition}{Definition}

\newtheorem{corollary}[theorem]{Corollary}

\newcommand{\nocontentsline}[3]{}
\newcommand{\tocless}[2]{\vspace{2ex}\bgroup\let\addcontentsline=\nocontentsline#1{#2}\egroup}
\newcommand{\toclesslabel}[3]{\vspace{2ex}\bgroup\let\addcontentsline=\nocontentsline#1{#2\label{#3}}\egroup}

\newlength\replength
\newcommand\repfrac{.33}

\newcommand\rulewidth{.6pt}
\newcommand\tdashfill[1][\repfrac]{\cleaders\hbox to \replength{%
  \smash{\rule[\arraystretch\ht\strutbox]{\repfrac\replength}{\rulewidth}}}\hfill}

\newcommand\tdotfill[1][\repfrac]{\cleaders\hbox to \replength{%
  \smash{\raisebox{\arraystretch\dimexpr\ht\strutbox-.1ex\relax}{.}}}\hfill}

\begin{document}

\title{Digital Simulation of Non-Hermitian Knotted Bands on Quantum Hardware}

\author{Truman Yu \surname{Ng}\,\orcidlink{0009-0006-5036-2298}}
\thanks{These two authors contributed equally.}
\affiliation{Department of Physics, National University of Singapore, Singapore 117551, Republic of Singapore}
\affiliation{A*STAR Quantum Innovation Centre (Q.InC), Agency for Science, Technology and Research (A*STAR), 2 Fusionopolis Way, Innovis \#08-03, Singapore 138634, Republic of Singapore\looseness=-1}

\author{Yuzhu Wang \orcidlink{0000-0003-4151-5097}}
\thanks{These two authors contributed equally.}
\affiliation{School of Physical and Mathematical Sciences, Nanyang Technological University, Singapore 639798, Republic of Singapore}
\email{yuzhu.wang@ntu.edu.sg}

\author{Wei Jie Chan}
\affiliation{Department of Physics, National University of Singapore, Singapore 117551, Republic of Singapore}

\author{Ruizhe Shen}
\affiliation{Department of Physics, National University of Singapore, Singapore 117551, Republic of Singapore}

\author{Tianqi Chen}
\affiliation{Bioinformatics Institute, Agency for Science, Technology and Research (A*STAR), 30 Biopolis Street, No.~07-01 Matrix, Singapore 138671, Republic of Singapore}
\affiliation{A*STAR Quantum Innovation Centre (Q.InC), Agency for Science, Technology and Research (A*STAR), 2 Fusionopolis Way, Innovis \#08-03, Singapore 138634, Republic of Singapore\looseness=-1}

\affiliation{Centre for Quantum Technologies, National University of Singapore, Singapore 117543, Republic of Singapore}

\author{Ching Hua Lee}
\email{phylch@nus.edu.sg}
\affiliation{Department of Physics, National University of Singapore, Singapore 117551, Republic of Singapore}

\date{\today}

\begin{abstract}
Knots and links represent a fundamental motif of non-local connectivity that permeates the physical sciences from string theory to protein folds. While spectral braiding has been explored in two-band non-Hermitian models across various platforms, its direct simulation and characterization on programmable quantum hardware, particularly beyond two strands,  remains a formidable challenge due to the limitations of variational optimization in these systems. Here, we introduce a family of non-Hermitian multi-band twister models and implement a non-variational protocol to characterize their complex braided band structures on a programmable superconducting quantum processor. By mapping the winding of eigenstates to the spectral topology, we devise an efficient measurement strategy that extracts braid information—including braid words and knot invariants like the Alexander and Jones polynomials—without requiring full spectral tomography or repeated optimization. We experimentally demonstrate the reconstruction of complicated knots and links such as the Hopf chain and Solomon’s knot. Our approach provides a general framework for investigating exotic non-Hermitian topology on near-term quantum devices, opening a route to simulate more sophisticated topological structures in knot theory.
\end{abstract}

\maketitle

\tocless\section{Introduction}
\label{sec:intro}

Knots -- loops in 3D space which cannot be continuously deformed into trivial rings~\cite{rolfsen2003knots,murasugi2007knot,crowell2012introduction,lickorish1997introduction} -- have garnered tremendous interest in the physical sciences. They constitute an exotic and unifying topological motif that permeates a remarkable range of natural and engineered systems. In biology and chemistry, knots emerge in contexts from protein folding~\cite{Mallam2009does,Faisca2015knotted} to molecular polymers~\cite{Forgan2011chemical,Fielden2017molecular,Zhang2024mechanical,Sun2025,Yang2025c}. In the realm of physics, knot topology plays an equally central role, underpinning phenomena across string theory and quantum gravity~\cite{witten1989quantum}, topological semimetals~\cite{PhysRevB.96.041103,PhysRevB.96.201305,RevModPhys.93.025002,li2018realistic}, non-Hermitian band structures~\cite{wang2021topological,PhysRevB.106.195425}, as well as artificial platforms such as liquid crystals~\cite{hall2025fusion,bupathy2025knots}, magnetic systems~\cite{kasai2025controlling,li2026electrically}, and photonic systems~\cite{PhysRevLett.133.230603,sun2025reconfigurable,song2025shortcuts,kim2026programmable,Wang2026observation}. The braids which form knots upon closure underpin topological quantum computation \cite{simon2023topological,Kitaev2002,Nayak2008,Andersen2023,Minev2025} where anyons encode quantum information in their worldlines, and their non-abelian exchange statistics implement quantum gates that are topologically protected against local noise and decoherence. This was experimentally demonstrated in recent seminal works on non-abelian anyons on a superconducting processor \cite{Iqbal2024,Xu2024a}. 

In condensed matter, knotted band structures are particularly interesting, capturing band permutations unique to non-local band connectivity. Non-Hermitian multi-band Hamiltonians are a natural setting for such behavior because their two-dimensional complex spectra admit branch-point singularities, notably exceptional points (EPs) where energies become degenerate and eigenstates coalesce, enforcing band permutations under parameter loops~\cite{Heiss2012,Dembowski2001,El-Ganainy2017,Xu2016a,Hodaei2017,Ruter2010, yang2026beyond} and enabling non-trivial braiding in contrast with Hermitian band structures. This motivation is reinforced by rapid progress in non-Hermitian topology, including unconventional non-local skin response described by generalized Brillouin zones~\cite{yokomizo2019non,ashida2020non,lin2023topological,meng2025generalized} and the enriched classification of topological phases~\cite{PhysRevX.9.041015,Delplace2021,Xiao2022,Jiang2018,Zhang2021c,Borgnia2020,Ding2022,Kunst2018}. Experimentally, non-Hermitian braids and related phenomena have been realized on various platforms such as solid-state simulators~\cite{PhysRevLett.127.090501,Wang2025e}, optical cavities~\cite{patil2022measuring}, optical resonators~\cite{wang2021topological}, trapped ion simulators~\cite{PhysRevLett.130.163001}, electric circuits~\cite{PhysRevB.109.165128,Cao2025}, superconducting circuits~\cite{PRXQuantum.6.020328,Han2024}, static mechanical lattices~\cite{wang2025observing}, magnonics ~\cite{Rao2024}, photonic simulators ~\cite{Wang2026observation} and acoustic lattices~\cite{zhang2022observation,Zhang2023f,Tong2025a,Tang2022}. In parallel, machine-learning-based algorithms have been developed to characterize non-Hermitian braids computationally \cite{chen2024machine,long2024unsupervised}, motivating complementary protocols that extract braid information from feasible measurements in multi-band experiments.

Quantum processors provide a timely and programmable route to access eigenstate information relevant to non-Hermitian band topology, enabled by recent progress in quantum simulations of non-Hermitian Hamiltonians~\cite{naghiloo2019quantum,wu2019observation,li2019observation,partanen2019exceptional,ding2021experimental,Wang2021observation,Liu2021dynamically,ren2022chiral,Chen2023AKLT,PhysRevLett.127.090501,Jebraeilli2025nonHermitianqubit,Zheng2018DualityPT,Huang2019WeakMeasurementPT,koh2022stabilizing,Liu2023PracticalNH,An2023LCHSNonunitary,Zheng2013FastEvolutionPT,Wen2019DigitalPTSimulation,chen2026robust, Dogra2021PTBreakingSuperconducting,koh2024realization, Lin2022NHDynamicalTopologySSE,Bian2023AntiPTTrappedIon,Kivela2024PseudoHermitianLZSM,Shen2024FockSkinScars,Koh2025EdgeClusterBursts,Shen2025CircuitMitigation,PhysRevLett.127.090501,yu2021experimental,shen2025observation,Wu2023KnotTopologySingleSpin,shen2026observation}. However, extracting topological information in multi-band non-Hermitian models remains experimentally demanding on present-day devices~\cite{Preskill2018,Wu2023KnotTopologySingleSpin,yu2021experimental,Zhang2025WeylExceptionalRings,Chen2025ThirdOrderEPTomography,Zhang2025EP3BraidingStateTransfer,Han2024HigherOrderEPInvariants,Wang2025NHNonAbelianNV,shen2025observation,Shen2024FockSkinScars,Koh2025EdgeClusterBursts}. Existing quantum-processor demonstrations largely focus on two-band models via non-unitary time evolution~\cite{PhysRevLett.127.090501,PhysRevLett.130.163001,yu2021experimental} where non-local signatures require deep circuits and many shots~\cite{laakkonen2025less,Preskill2018,Takagi2022}. For multi-band Hamiltonians, time-evolution-based schemes are further constrained because they preferentially project onto eigenstates associated with extremal imaginary energies~\cite{PhysRevLett.127.090501,PhysRevLett.130.163001,yu2021experimental}, while approaches that aim to access multiple bands typically rely on repeated, band-by-band variational optimization~\cite{zhao2023universal,xie2024variational,Cerezo2021} or non-unitary primitives such as quantum singular value decomposition~\cite{Rebentrost2018qSVD,Wang2021variationalquantum,Zhang2025exponential}. The resulting circuit depth and sampling overhead grow rapidly with the number of bands, precisely where more intricate multi-strand knots and links become accessible, thus motivating measurement protocols that extract spectral topology on quantum hardware without full tomography or repeated optimization.

In this work, we develop and implement a non-variational protocol to characterize braided band structures in generic quantum circuits, and concretely demonstrate how it can map out knotted band structures of up to four strands on IBM quantum processors. To this end, we introduce a family of multi-band Hamiltonians specifically designed to realize arbitrary torus knots and links. By preparing and reconstructing their eigenstates, we extract the spectral braiding information encoded in these models. Our approach hinges on measuring an experimentally accessible winding matrix that encodes the projected eigenstate trajectories. From the winding number matrix, we can recover braid words that characterize the multi-strand braid, ostensibly avoiding full spectral tomography. Our proposed measurement-and-reconstruction strategy can, in principle, be extended to higher-spin models and more sophisticated braiding topologies, with overheads that increase linearly with the number of reconstructed eigenstates and winding number computations. 

\tocless\section{Results}
\label{sec:results}

\tocless\subsection{Knots and Braids}

Braids formalize how worldline trajectories of objects (``strands'') intertwine as they evolve. 
In our context, the evolution parameter is labeled $k$, and the evolving objects are the $N$ complex energy eigenvalues $\{E_i(k)\}$, $i=1,...,N$ of an $N$-band Hamiltonian $\hat{H}(k)$. Although the eigenvalues (``bands'') are periodic as an unordered set, individual eigenvalue branches may be permuted after one cycle of $k$. This band permutation reflects the monodromy of the multi-sheeted energy surface and can occur when the physical loop winds around branch-point singularities in the parameter space. Thus, the bands may be ordered differently after one period. One essential feature is that the order of successive interchanges matters non-trivially: different orderings of the braid crossings constitute topologically different braids, i.e., do not commute, even if they ultimately correspond to the same permutation of the strands.

Mathematically, we describe braids as isotopy classes described by the braid group $B_N$, which is generated by $N-1$ elementary exchanges $\sigma_1, \ldots, \sigma_{N-1}$ of the $N$ strands\footnote{See Fig.~\ref{fig:braid_closure} for illustrative examples with $N=2$ and $4$}. Its presentation is~\cite{Artin1947}
\begin{equation} 
B_N = \expval{\sigma_1, \ldots, \sigma_{N-1}\,|\, \sigma_{i} \sigma_{i+1} \sigma_{i} = \sigma_{i+1} \sigma_{i} \sigma_{i+1}, \sigma_i\sigma_j = \sigma_j\sigma_i}
\label{eq:results/knots_and_braids/braid_group_presentation}
\end{equation}
where $i \leq i \leq N-2$ and $\abs{i-j} \geq 2$. Here, $\sigma_i$ represents the crossing of strand $i$ over strand $i+1$, and its inverse $\sigma^{-1}_i$ deonotes the corresponding under-crossing. 
The first and second Artin relations $ \sigma_{i} \sigma_{i+1} \sigma_{i} = \sigma_{i+1} \sigma_{i} \sigma_{i+1}$ and $\sigma_i\sigma_j = \sigma_j\sigma_i$ respectively encode the associativity of exchanges among three adjacent strands, and the independence of exchanges acting on non-adjacent strand pairs. A knot or link is obtained from a braid by connecting the endpoints of each strand to form one or more closed components, as illustrated in \cref{fig:braid_closure}a and \cref{fig:braid_closure}b for braid words $\sigma_1^2$ and $\sigma_1\sigma_3\sigma_2\sigma_1\sigma_3\sigma_2$ corresponding to a two-strand Hopf link and a four-strand Solomon's knot. Alexander’s theorem guarantees that every knot or link admits such a braid representation~\cite{Alexander1928}. In physical terms, this precisely maps band permutations and their braid words to knotted or linked spectral structures. As different braids can yield the same closed knot or link under Markov equivalence \cite{kassel2008braid}, experimentally resolving the braid up to these moves is sufficient to identify the underlying knot or link type.

\
The topological richness of braids and their corresponding knots is reflected in the complexity of their topological invariants: there are infinitely many knot types -- and the number of prime knots increases exponentially with minimal crossing number, with tabulations already exceeding $10^9$ classes at twenty crossings ~\cite{Ernst1987GrowthPrimeKnots,Thistlethwaite2025Prime20Crossing}. 
No single conventional invariant is complete since distinct knots may share the same Alexander, Jones, and even HOMFLY-PT polynomials, as exemplified by illustrative mutant knots and Kanenobu families ~\cite{MortonCromwell1996Mutants,Kanenobu1986SamePolynomial}. Fortunately, for knots presented as braid closures, polynomial invariants remain algorithmically accessible; notably, the Jones polynomial follows directly from the reduced Burau representation of the braid word (see \cref{app:knot_polynomials/computation_of_the_Jones_polynomials,appxeqn:Jonespoly_final} in the Supplementary Material) \cite{Conway2018BurauAlexander}. 

\begin{figure}
\centering
\includegraphics[width=\linewidth]{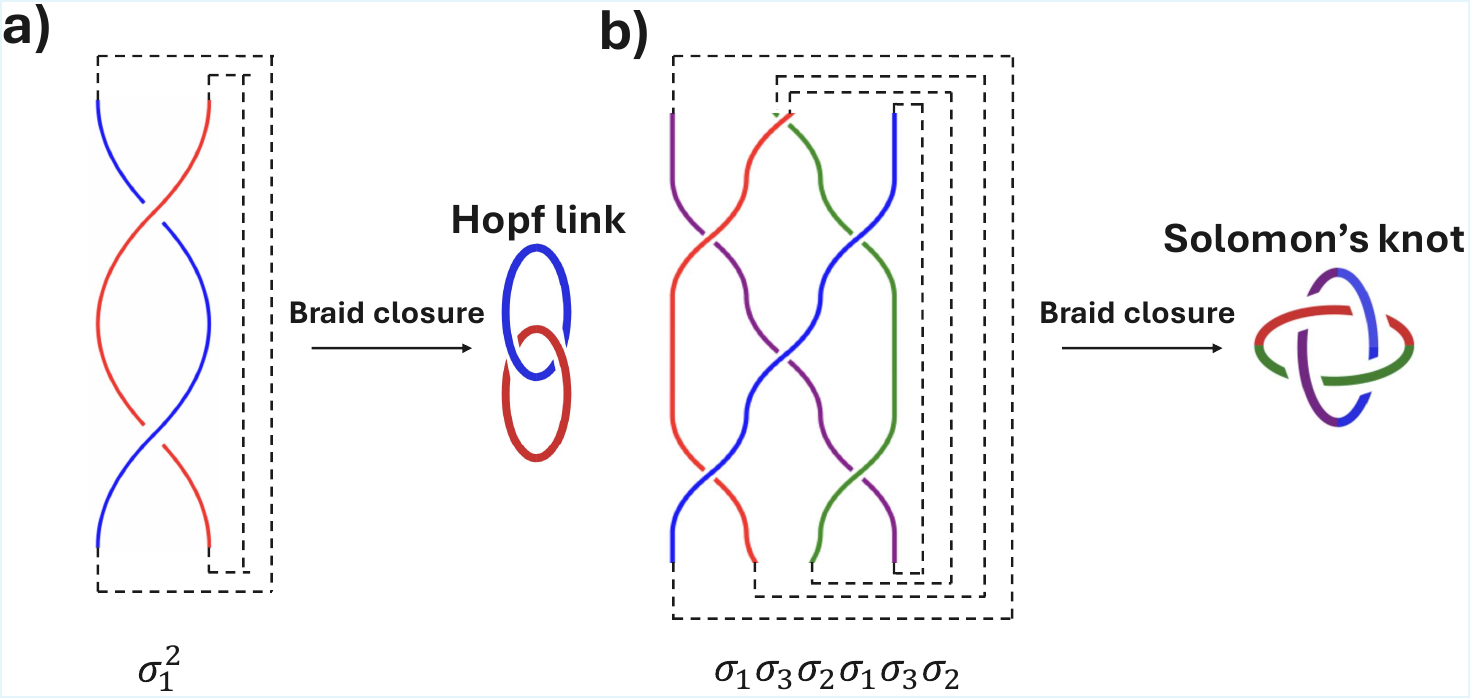}
\caption{\textbf{Examples of braids and their corresponding knots/links from braid closure.} \textbf{a)} The Hopf link is formed from two strands. \textbf{b)} Solomon's knot is formed from four strands. $\sigma_i$ denotes strand $i$ crossing over strand $i+1$. The braid word is formed by recording the crossings from the bottom to the top. The twister models that host these braids are described in their respective sections in the Measured Knotted Bands on Quantum Hardware section.
} 
\label{fig:braid_closure}
\end{figure}

\tocless\subsubsection{Braiding of Complex Energies}

\begin{figure*}[htbp!]
\centering
\includegraphics[width=\linewidth]{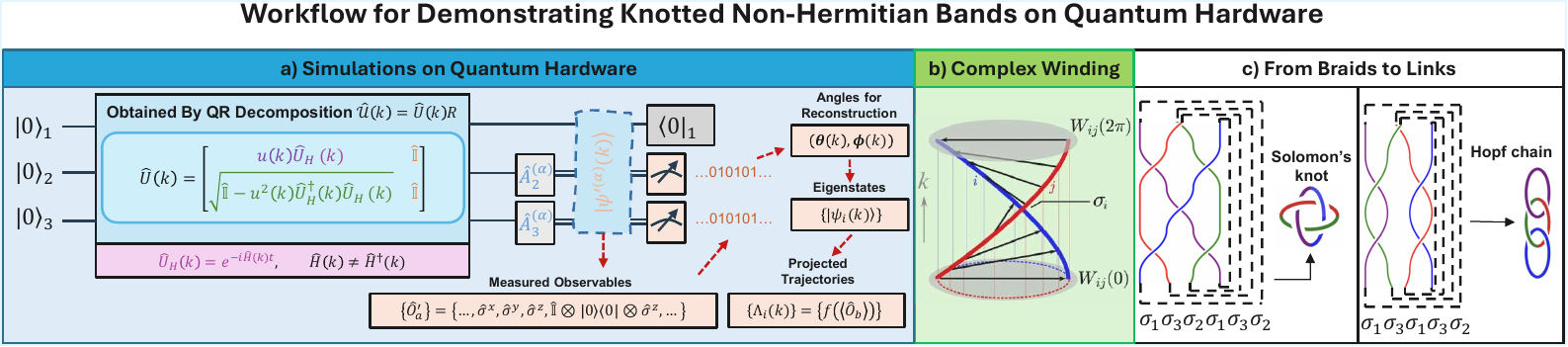}
\caption{\textbf{Our workflow from quantum circuit implementation to the measurement and topological characterization of non-Hermitian braiding.} \textbf{a)} Given any non-Hermitian $N$-band model, a $M$-qubit non-unitary operator $e^{-i\hat{H}^{(N)}(k)t}$ is embedded in a $(M+1)$-qubit operator $\hat{\mathcal{U}}(k)$ where $M = \ceil{\log_2 N}$. The first qubit is the ancilla and the remaining qubits encode the system with the evolved state $\ket*{\psi^{'(\alpha)}(k)}$ given in \cref{eq:results/quantum_circuit_hardware_implementation/reconstruction/psi_prime}. We perform a QR decomposition of $\hat{\mathcal{U}}(k)$ \cref{eq:results/quantum_circuit_hardware_implementation/obtaining/block-embedded_U} to get the desired unitary operator $\hat{U}(k)$ corresponding to the target eigenstate that is transpiled on hardware. Local rotation operators \cref{eq:results/quantum_circuit_hardware_implementation/reconstruction/A} are applied to measure $\hat{\sigma}^x$ and $\hat{\sigma}^y$ in the computational basis. The measured observables $\{O'_a\}$ are chosen to reproduce an eigenstate's components which allows one to compute the Bloch angles $(\vb*{\theta}(k), \vb*{\phi}(k))$ that parameterize the eigenstates $\{\ket{\psi_i(k)}\}$ and the measurements are postselected on the ancilla with $\ket{0}$. Since the eigenstates encode energies, we extract the energies via $\{\hat{O}_b\}$, thus yielding the projected eigenstate trajectories \cref{eq:results/quantum_circuit_hardware_implementation/characterizing_braids_and_knotted_structures/Lambda}. \textbf{b)} Classically, we compute the winding numbers $W_{ij}(k)$ defined in \cref{eq:results/quantum_circuit_hardware_implementation/characterizing_braids_and_knotted_structures/partial_winding} which quantifies the phase difference between bands $i$ and $j$. Here, $W_{ij}(2\pi) = 1/2$ since the band indices are permuted at $k = 2\pi$. $\sigma_i$ denotes band $i$ crossing over band $j$. \textbf{c)} Upon braid closure, we get a knot or link. The parameters used for Solomon's knot and the Hopf chain are $(m_0, m_1) = (-0.5, -0.4)$ and $(m_0, m_1) = (2, 1.1)$ respectively in \cref{eq:quantum_hardware_simulation/four-band_twister_model/H4_twister_model}. In Solomon's knot, the band indices (blue, red, green, purple) $(1, 2, 3, 4)$ at $k = 0$ are permuted to become (purple, green, red, blue) $(4, 3, 2, 1)$ at $k = 2\pi$ after one period. A second period permutes the band indices as $(4, 3, 2, 1) \to (1, 2, 3, 4)$ thus $n = 2$ in \cref{eq:results/quantum_circuit_hardware_implementation/characterizing_braids_and_knotted_structures/partial_winding}, hence $W_{ij}(k)$ is well-defined since the band indices are permuted back to $(1, 2, 3, 4)$. 
}
\label{fig:methodology_illustration}
\end{figure*}

For concreteness, we base our spectral band braiding on a family of $N$-band non-Hermitian twister Hamiltonians which encompass the $N=2$-band twister model \cite{PhysRevLett.126.010401} as a special case. The family is defined by 

\begin{equation}
\hat{H}^{(N)}(k) 
= m_0 \hat{\Sigma}^{(N)} + \sum_{v=1}^V m_v \hat{T}_v^{(N)}(k),
\label{eq:results/knots_and_braids/braiding/twister_model}
\end{equation}
where $\hat{\Sigma}^{(N)}_{pq} = \left(1 - \frac{2(p-1)}{N - 1}\right) \delta_{pq}$, with $1\leq p, q \leq N$, are the elements of the diagonal matrix that ascribes a linear energy shift to the bands, and
\begin{equation}
\hat{T}_v^{(N)}(k) = 
\mqty[
0 & 0  & \cdots & 0 & e^{i v k} \\
1 & 0  & \cdots & 0 & 0 \\
0 & 1  & \cdots & 0 & 0 \\
\vdots & \vdots & \ddots & 0 & 0 \\
0 & 0  & \cdots & 1 & 0
], 
\label{eq:results/knots_and_braids/braiding/Tn_twister_model}
\end{equation}
with $m_0, m_1, \ldots, m_V \in \mathbb{C}$ that contains the $v$-th harmonic $e^{i\frac{v k}{N}}$. Higher harmonics increase the complexity of the $k$-winding trajectories and result in more sophisticated braiding possibilities -- see \cref{app:generalized_twister_models} in the Supplementary Material for details. Our family of twister models hosts braids whose closures correspond to torus knots or links, whose exact shapes can be tuned by the coefficients $m_v$ \cite{rolfsen2003knots,zhang2021tidal,bode2017knotted}. In this work, we consider only cases with positive $v$ such that the braids formed can be described by $\{\sigma_i\}$ generators only. 

\tocless\subsection{Quantum Circuit Implementation}

Our overall task is to perform digital quantum simulations of the complex energy bands of \cref{eq:results/knots_and_braids/braiding/twister_model} (specifically, \cref{eq:quantum_hardware_simulation/two-band_twister_model/H2_twister_model,eq:quantum_hardware_simulation/four-band_twister_model/H4_twister_model} for $N = 2$ and $4$) and measure how they braid using  
$ M=\ceil{\log_2 N}$ qubits with $\ceil{\cdot}$ the ceiling function. 
A programmable quantum processor can simulate any desired model of interest with well-defined state preparation and measurement protocols, as shown in \cref{fig:methodology_illustration}. 

While it is also possible to encode a braid with $\ceil{N/2}$ qubits with each qubit representing two adjacent strands, encoding it optimally in $\ceil{\log_2 N}$ qubits incurs less noise. Besides, although $\ceil{N/2}$ qubits give more control over each qubit since we can treat them individually, this representation is unable to encode braids where crossings occur between strands encoded by different qubits, thus reducing the ability to realize certain braids. 

\tocless\subsubsection{Obtaining Desired Braiding States Through Non-Unitary Evolution} 

The first step is to obtain the  Hamiltonian eigenstates that undergo the braiding.  Consider a generic $N$-band non-Hermitian Hamiltonian $\hat{H}(k)$ parameterized by a periodic parameter $k\in[0, 2\pi]$. To converge to the target eigenstate $\ket{\varphi_{\max}(k)}$ with $\max[\Im(E(k))]$, we realize a non-unitary evolution generated by 
\begin{equation}
\hat{U}_H(k) = e^{-i\hat{H}(k)t}, 
\label{eq:results/quantum_circuit_hardware_implementation/obtaining/U_H}
\end{equation}
where $t$ is the evolved time, to be implemented as Trotterized evolutions. We remark that it may be necessary to rotate $\hat{H}(k)$ by a suitable additional phase factor to converge to $\ket{\varphi_{\max}(k)}$ (c.f. \cref{eq:methods/obtaining_eigenstates/U_H}). Through the evolution, $\hat{U}_H(k)$ amplifies the components in the evolved state such that $\ket{\varphi_{\max}(k)} = \hat{U}_H(k)\ket{0}^{\otimes M} = \sum_j e^{-iE_j(k)t}c_j\ket{\varphi_j(k)} \approx \ket{\psi_{\max}(k)}$ where $\{\ket{\varphi(k)}\}$ are the right eigenstates of $\hat{H}(k)$ and $\ket{\psi_{\max}(k)}$ is the reconstructed eigenstate (c.f. \cref{eq:results/quantum_circuit_hardware_implementation/reconstruction/psi}).

As $\hat H(k)$ is non-Hermitian, $\hat U_H(k)$ as defined above is non-unitary and cannot be physically realized with native quantum gates. A tried-and-tested way forward is to introduce an additional ancilla qubit such that $\hat{U}_H(k)$ is subsumed by a larger unitary evolution $\hat U$ that is unitary: to recover the action of $\hat{U}_H(k)$, we perform post-selection on the ancilla qubit after the evolution as depicted in \cref{fig:methodology_illustration}a. 

To effect the non-unitary evolution $\hat{U}_H(k)$ by a unitary $\hat{U}(k)$, we construct $\hat{U}(k)$ by defining an intermediary matrix $\hat{\mathcal{U}}(k)$ \cite{PRXQuantum.2.010342} 
\begin{equation}
\hat{\mathcal{U}}(k) = \mqty[
u(k)\hat{U}_H(k) & \hat{\mathbb{I}} \\
\hat{C}(k) & \hat{\mathbb{I}}
] = \hat{U}(k)R, 
\label{eq:results/quantum_circuit_hardware_implementation/obtaining/block-embedded_U}
\end{equation}
where $u(k) \geq 0$, $u(k)^{-2}$ is the maximum eigenvalue of $\hat{U}_H^\dagger(k) \hat{U}_H(k)$, and $\hat{C}(k) = \sqrt{\hat{\mathbb{I}} - u(k)^2\hat{U}_H^\dagger(k)\hat{U}_H(k)}$. 
To recover the desired unitary $\hat{U}(k)$, one performs the QR decomposition to decompose $\hat{\mathcal{U}}(k)$ into $\hat{U}(k)$ and $R$, an unimportant upper triangular matrix. This can be shown to work because $u(k)\hat{\mathcal{U}}(k)$ and $\hat{C}(k)$ form the first $2^M$ orthonormal columns. Finally, the desired unitary $\hat{U}(k)$ is transpiled to the hardware's native gate set. 

\tocless\subsubsection{Reconstruction of Eigenstates}

As illustrated in \cref{fig:methodology_illustration}a and detailed below, we reconstruct the eigenstates and extract their relative windings by measuring the expectations of appropriately chosen operators $\{O'_a\}$ -- a generic $N$-component state contains $2(N - 1)$ independent real parameters $(\vb*{\theta}(k),\vb*{\phi}(k))$ up to normalization and a global phase, and determining it thus requires at least $2(N - 1)$ independent observables. We remark that most aspects of our methodology remain applicable for generic $N$-band models even though we illustrate it specifically for our twister models. 

Specifically, we measure the observables $\{\hat{O}'_a\}$ (see \cref{eq:quantum_hardware_simulation/two-band/measured_observables_expectation,eq:quantum_hardware_simulation/four-band/measured_observables_expectation}) with respect to the evolved state 

\begin{equation}
\ket*{\psi^{'(\alpha)}(k)} = (\hat{\mathbb{I}}_1\otimes\hat{A}^{(\alpha)}_2 \otimes \cdots \otimes \hat{A}^{(\alpha)}_{M+1}) \hat{U}(k)\ket{0}^{\otimes(M+1)},  
\label{eq:results/quantum_circuit_hardware_implementation/reconstruction/psi_prime}
\end{equation}
where the first qubit is the ancilla, the remaining $M$ qubits encode the system, and the local basis rotations are given by 

\begin{equation}
\hat{A}_i^{(\alpha)} \in \{\hat{R}^y_i(-\pi/2), \hat{R}^x_i(\pi/2), \hat{\mathbb{I}}_i\}, 
\label{eq:results/quantum_circuit_hardware_implementation/reconstruction/A}
\end{equation}
where $\alpha\in\{x, y, z\}$, $\hat{R}^x_i(\theta) = \exp(-i\hat{\sigma}^x_i\theta/2)$, and $\hat{R}^y_i(\theta) = \exp(-i\hat{\sigma}^y_i\theta/2)$ \footnote{We seek to measure $\expval{\hat{\sigma}_i^x}$ and $\expval{\hat{\sigma}_i^y}$. $\hat{A}^{(x)}_i$ and $\hat{A}^{(y)}_i$ satisfy $\hat{\sigma_i^x} = \hat{A}^{(x)\dagger}_i\hat{\sigma}^z\hat{A}^{(x)}_i$ and $\hat{\sigma_i^y} = \hat{A}^{(x)\dagger}_i\hat{\sigma}^z\hat{A}^{(x)}_i$ respectively.}; these rotations are necessary to measure $\hat{\sigma}^x$ and $\hat{\sigma}^y$. 

Having obtained the evolved state $\ket*{\psi^{'(\alpha)}(k)}$, we measure it, and the shots are postselected on the ancilla outcome $\ket{0}$ to obtain the non-unitary dynamics. To see this, observe that the action of $\hat{U}(k)$ yields 

\begin{equation}
\begin{aligned}
\hat{U}(k)(\ket{0}\otimes\ket{0}^{\otimes M}) &= \ket{0}\otimes u(k)\hat{U}_H(k)\ket{0}^{\otimes M} \\
&\qquad + \ket{1} \otimes\hat{C}(k)\ket{0}^{\otimes M}. 
\label{eq:results/quantum_circuit_hardware_implementation/reconstruction/U_action}
\end{aligned}
\end{equation}
The bitstrings that correspond to $\ket{1}$ for the ancilla qubit are discarded, and the distribution over the retained bitstrings, with proper normalization, gives the measured expectation values which are necessary for the next step: reconstruction of the eigenstates $\{\ket{\psi_i(k)}\}$. 

For each eigenstate at some $k$, up to normalization and a global phase, we parameterize it in terms of the Bloch angles $(\vb*{\theta}(k), \vb*{\phi}(k))$

\begin{equation}
\ket{\psi(k)} = \frac{1}{\mathcal{N}(k)}\mqty[
g_1(\vb*{\theta}(k), \vb*{\phi}(k)) \\
g_2(\vb*{\theta}(k), \vb*{\phi}(k)) \\
\vdots \\
g_N(\vb*{\theta}(k), \vb*{\phi}(k))
], 
\label{eq:results/quantum_circuit_hardware_implementation/reconstruction/psi}
\end{equation}
where the functions $g_i(\vb*{\theta}(k), \vb*{\phi}(k))\in\mathbb{C}$ except $g_N(\vb*{\theta}(k), \vb*{\phi}(k))\in\mathbb{R}$. Choosing the measured observables $\{\hat{O}'_a\}$ depend on the parameterization given in \cref{eq:results/quantum_circuit_hardware_implementation/reconstruction/psi}. Since a single qubit encodes a two-band model, only two angles are required. Conversely, the additional degrees of freedom in a four-band twister model require six angles. The measured observables and the Bloch angles will be explicitly defined when we specialize to the two- and four-band twister models in their respective sections. 

\tocless\subsubsection{Characterizing the Braids and Knotted Structures}

From the measured state, the corresponding eigenvalue and its braiding can be extracted, at least for the twister models considered and their generalizations.
This is done by performing a stereographic projection from the Bloch hypersphere to a hyperplane. 
In practice, one chooses the component that has the simplest dependence on $E(k)$ among others, and the projection is computed by evaluating the expectation values of a set of observables $\{\hat{O}_b\}$ to extract $E(k)$ from $\ket{\psi(k)}$. Each trajectory $\Lambda(k)$ is defined as

\begin{equation}
\Lambda(k) := f(\{\!\mel{\psi(k)}{\hat{O}_b}{\psi(k)}\}), 
\label{eq:results/quantum_circuit_hardware_implementation/characterizing_braids_and_knotted_structures/Lambda}
\end{equation}
with $\ket{\psi(k)}$ given in \cref{eq:results/quantum_circuit_hardware_implementation/reconstruction/psi}. The precise form of $f(\cdot)$ depends on the choice of $\{\hat{O}_b\}$, which should be homeomorphic to the energy eigenvalues, e.g., $\{\Lambda_i(k)\}\propto\{E_i(k)\}$, to ensure that the trajectories $\{\Lambda_i(k)\}$ reproduce the spectral braiding information of $\{E_i(k)\}$ as desired. For the two- and four-band twister models, $\{\Lambda_i(k)\}$ are explicitly defined in their respective sections (see \cref{eq:quantum_hardware_simulation/two-band/Lambda,eq:quantum_hardware_simulation/four-band/Lambda}). Note that full eigenstate reconstruction is not required experimentally; the eigenstates primarily act as an intermediate construct for determining the mapping $f(\cdot)$ to give $\{\Lambda_i(k)\}$ that connects to $\{E_i(k)\}$.

With the trajectories $\{\Lambda_i(k)\}$, a straightforward measurable quantity is the pairwise $k$-dependent winding number $W_{ij}(k)$ which represents the relative winding between the $i$-th and $j$-th eigenstates shown in \cref{fig:methodology_illustration}b. $W(k)$ is symmetric with elements defined as 

\begin{equation}
W_{ij}(k) := \frac{1}{2\pi i}\int_{0}^k dk' \, \partial_{k'}\ln[\Lambda_i(k')-\Lambda_j(k')]. 
\label{eq:results/quantum_circuit_hardware_implementation/characterizing_braids_and_knotted_structures/partial_winding}
\end{equation}
Geometrically, $W_{ij}(k)$ quantifies the net phase difference between $\Lambda_i(k)$ and $\Lambda_j(k)$, and hence how band $i$ winds around band $j$. Importantly, the partial winding information up to $k$ allows us to identify the braid crossings. 

In the experiment, an $N$-band Hamiltonian is periodic after one cycle but the individual band labels may be permuted. Consequently, the spectrum returns to itself as an unordered set after one period although a given strand may return to its initial label only after $n$ periods. To obtain a label-independent and thus topological winding matrix for a closed braid, we first compute the one-period winding matrix $\bar W := W(2\pi)$ and then combine the contributions from the remaining periods by applying the band permutation matrix $P$ extracted from the reconstructed eigenstates

\begin{equation}
\mathcal{W} := \frac{1}{n} \sum_{a=0}^{n-1} (P^{-1})^a \bar{W} P^a,
\label{eq:results/quantum_circuit_hardware_implementation/characterizing_braids_and_knotted_structures/winding_number_matrix}
\end{equation}
where $P^n = \mathbb{I}$. This procedure converts the measured one-period trajectories into the winding number matrix of the closed multi-strand braid. The formal distinction between the $k$-dependent winding $W_{ij}(k)$, the one-period winding $\bar{W}_{ij}$, and the topological winding $\mathcal{W}_{ij}$ is given in Methods.

As the complex energies/eigenstates vary in $k$, they trace out ``worldline'' trajectories in three dimensions as depicted in \cref{fig:methodology_illustration}b. Projecting them onto a chosen plane (typically the plane of the paper) defines a braid diagram where the trajectories become the braid strands that either cross over or under each other, as in \cref{fig:methodology_illustration}c, where each crossing point is defined as 

\begin{equation}
W_{ij}(k_{\mathrm{crossing}}) := \frac{1}{2\pi}\bigg(\frac{\pi}{2} + r\pi\bigg) = \frac{1}{4} + \frac{r}{2}, 
\label{eq:results/quantum_circuit_implementation/characterizing_braids_and_knotted_structures/W_crossing_condition}
\end{equation}
with $r\in\mathbb{Z}$. At each crossing in \cref{fig:methodology_illustration}c, there is a corresponding $W_{ij}(k_{\mathrm{crossing}})$ where $r$ is inferred. The intuition behind \cref{eq:results/quantum_circuit_implementation/characterizing_braids_and_knotted_structures/W_crossing_condition} is as follows. Whenever there is a crossing point in \cref{fig:methodology_illustration}c, the black vector that captures the winding in \cref{fig:methodology_illustration}b either points into or out of the paper which is quantified by the $\pi/2$ term in \cref{eq:results/quantum_circuit_implementation/characterizing_braids_and_knotted_structures/W_crossing_condition}; the two possible orientations of the normal vector are captured by $r\pi$. 

Since there is a one-to-one correspondence between $W_{ij}(k_{\mathrm{crossing}})$ and the crossing points, we can extract the braid word by recording how the band indices are permuted after each crossing. This is notated as $\tau_{\pi(i)\pi(i+1)}^{(-1)^r}$, where $\pi(\cdot)$ is the permutation operator that maps a band index to another index after a crossing. Depending on $r$ and the permuted band indices, a braid generator corresponding to a crossing point is written as follows

\begin{equation}
\tau_{\pi(i)\pi(i+1)}^{\pm 1} = \tau_{\pi(i+1)\pi(i)}^{\mp 1} \Longrightarrow \sigma_{\pi(i)}^{\pm 1\cdot (-1)^{\delta_{4N}}}.
\label{eq:results/quantum_circuit_implementation/characterizing_braids_and_knotted_structures/crossing_to_braid_generator}
\end{equation}
The equality is due to the two possible orientations of the normal vector. Importantly, the same normal vector must be referenced when forming the braid word to ensure consistency. Repeating this analysis for all crossing points yields the braid word. See \cref{eq:extraction_of_braid_words/phase_shifted_W_crossing_condition,eq:extraction_of_braid_words/braid_word_extraction} in Methods for the geometric intuition behind \cref{eq:results/quantum_circuit_implementation/characterizing_braids_and_knotted_structures/W_crossing_condition} and the procedure to form the braid word from $W_{ij}(k)$, respectively. 

\tocless\subsection{Measured Knotted Bands on Quantum Hardware}

Having discussed the workflow for the quantum circuit implementation of knotted bands, we move on to discuss specific implementation results on the IBM quantum computer. Each unitary $\hat{U}(k)$ and local rotation operators $\hat{A}_i^{(\alpha)} \in \{\hat{R}^y_i(-\pi/2), \hat{R}^x_i(\pi/2), \hat{\mathbb{I}}_i\}$ thus define a circuit which is transpiled to the native gate set of \texttt{ibm\textunderscore marrakesh}. All simulations were executed with 40000 shots per circuit on a uniform grid of 100 $k$-points. 

\tocless\subsubsection{Measuring Braids with $N = 2$ Strands}

\begin{figure*}[htbp!]
\centering
\includegraphics[width=\linewidth]{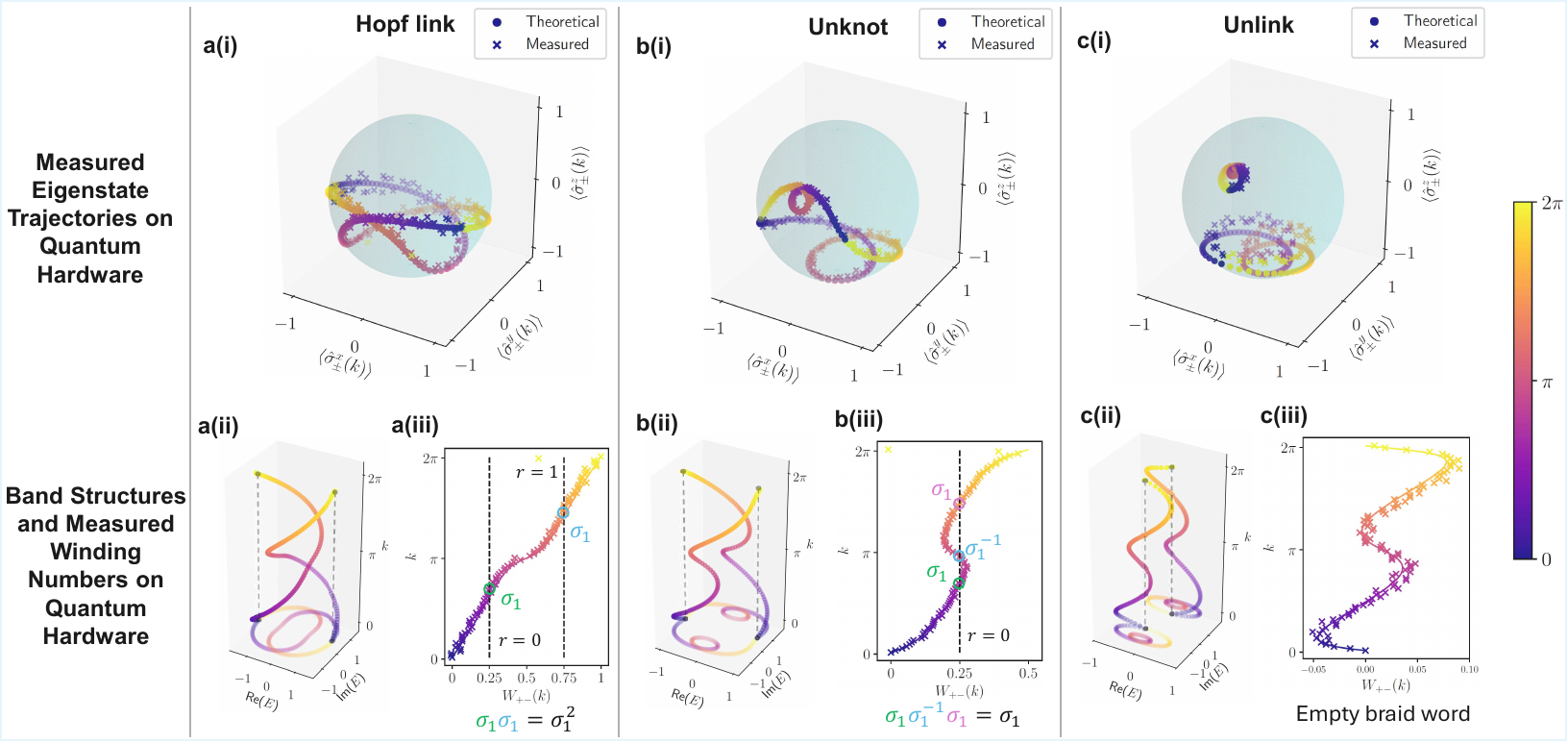}
\caption{\textbf{Characterization of braids from the measured winding number $W_{+-}(k)$ for the two-band twister model \cref{eq:quantum_hardware_simulation/two-band_twister_model/H2_twister_model}.} Its eigenvalues $E_\pm(k)$ are defined below \cref{eq:quantum_hardware_simulation/two-band_twister_model/H2_twister_model}. In the upper row, the measured and theoretical values of $\expval{\hat{\sigma}_\pm^\alpha(k)}$ are given in \cref{eq:quantum_hardware_simulation/two-band/measured_observables_expectation,eq:quantum_hardware_simulation/two-band/theoretical_observables_expectation} respectively. In all plots, the color gradient indicates the evolution of $k$ for $k\in[0, 2\pi]$ and the cross markers are measured data on \texttt{ibm\textunderscore marrakesh}. The winding number computation is given in \cref{eq:quantum_hardware_simulation/two-band_twister_model/winding_number_discretized}, and the black dashed lines in the bottom row correspond to \cref{eq:results/quantum_circuit_implementation/characterizing_braids_and_knotted_structures/W_crossing_condition}, which denote the crossings on the braid diagram and the projection plane for each knot/link. The chosen parameters agree with Ref.~\cite{yu2021experimental}. \textbf{a(i)} The reconstructed eigenstate trajectories $\ket{\psi_\pm(k)}$ \cref{eq:quantum_hardware_simulation/two-band/reconstructed_eigenstate} \cite{PhysRevLett.126.010401} for the Hopf link with parameters $(m_0, m_1) = (0.5338, 0.6)$. \textbf{a(ii)} The band structure, which forms a braid when projected onto a plane; a Hopf link is formed when flattened onto the complex energy plane. \textbf{a(iii)} The measured winding number $W_{+-}(k)$ \cref{eq:quantum_hardware_simulation/two-band_twister_model/winding_number_discretized} on \texttt{ibm\textunderscore marrakesh} with two crossings (c.f. \cref{fig:braid_closure}a) at $W_{+-}(k_{\mathrm{crossing}}) = 1/4$ and $W_{+-}(k_{\mathrm{crossing}}) = 3/4$ corresponding to $r = 0$ and $r = 1$ respectively in \cref{eq:results/quantum_circuit_implementation/characterizing_braids_and_knotted_structures/W_crossing_condition}. Hence, the braid word is $\sigma_1^2$. \textbf{b(i)} The unknot with parameters $(m_0, m_1) = (1.273, 0.6)$. \textbf{b(ii)} An unknot is formed. \textbf{b(iii)} The winding number $W_{+-}(k)$ with three crossings, but two of them cancel out hence the braid word is $\sigma_1$. \textbf{c(i)} The unlink with parameters $(m_0, m_1) = (1.8889, 0.6)$. \textbf{c(ii)} An unlink is formed. \textbf{c(iii)} There is no winding thus the braid word is empty.
}
\label{fig:H2_twister_band_structure_bloch_sphere_W}
\end{figure*}

As a warm-up, we consider the case with $N=2$ strands which requires two qubits: one physical qubit that encodes the model and an ancilla qubit to postselect on non-unitary dynamics. By specializing \cref{eq:results/knots_and_braids/braiding/twister_model,eq:results/knots_and_braids/braiding/Tn_twister_model} to $N = 2$, the two-band twister model is given by

\begin{equation}
\begin{aligned}
\hat{H}^{(2)}(k) &= im_0\hat{\Sigma}^{(2)} + m_1 \hat{T}_1^{(2)}(k) + \hat{T}_2^{(2)}(k) \\
&= \mqty[
im_0 & e^{2ik} + e^{ik}m_1 \\
1 + m_1 & -im_0
],
\label{eq:quantum_hardware_simulation/two-band_twister_model/H2_twister_model}
\end{aligned}
\end{equation}
where $m_0, m_1\in\mathbb{R}$, $\hat{\Sigma}^{(2)} = \hat{\sigma}^z$, and $E_\pm(k) = \pm\sqrt{e^{2ik}(m_1+1) + m_1e^{ik}(m_1+1)-m_0^2}$. The band structures of \cref{eq:quantum_hardware_simulation/two-band_twister_model/H2_twister_model} can braid in three topologically distinct ways corresponding to the Hopf link (see \cref{fig:braid_closure}a), unknot, and unlink upon closure. For concreteness, we specialize to the parameters $(m_0, m_1) = (0.5338, 0.6)$, $(1.273, 0.6)$, $(1.8889, 0.6)$ for realizing the Hopf link, unknot, and unlink respectively (see \cref{fig:H2_phase_diagram} for the regions that host these knots/links in Methods). 

Using the eigenstate-preparation protocol in \cref{fig:methodology_illustration}a, both eigenstates $\ket{\psi_\pm(k)}$ of $\hat{H}^{(2)}(k)$ are
prepared on the quantum processor via block embedding and postselection. As such, \cref{eq:results/quantum_circuit_hardware_implementation/reconstruction/psi_prime} reduces to $\ket*{\psi^{'(\alpha)}_\pm(k)} = (\hat{\mathbb{I}}_1 \otimes \hat{A}_2^{(\alpha)})\hat{U}_\pm(k)\ket{00}$ where $\hat{U}_\pm(k)$ is obtained in the discussion regarding \cref{eq:results/quantum_circuit_hardware_implementation/obtaining/block-embedded_U}, and $\hat{A}_2^{(\alpha)} \in \{\hat{R}^y_2(-\pi/2), \hat{R}^x_2(\pi/2), \hat{\mathbb{I}}_2\}$ with $\alpha \in \{x, y, z\}$. The measured expectation values are defined as 

\begin{equation}
\begin{aligned}
\expval{\hat{\sigma}_\pm^\alpha(k)} &:= \mel*{\psi^{'(\alpha)}_\pm(k)}{(\hat{\mathbb{I}}\otimes\hat{\sigma}^z)}{\psi^{'(\alpha)}_\pm(k)} \\
&\xrightarrow[]{\text{PS}} \abs*{\braket*{00}{\psi^{'(\alpha)}_\pm(k)}}^2 - \abs*{\braket*{01}{\psi^{'(\alpha)}_\pm(k)}}^2.  
\label{eq:quantum_hardware_simulation/two-band/measured_observables_expectation}
\end{aligned}
\end{equation}
Since there is only one qubit used for each $\hat{H}^{(2)}(k)$, the measured observables are $\{\hat{O}'_a\} = \{\hat{\mathbb{I}} \otimes \hat{\sigma}^x, \hat{\mathbb{I}} \otimes \hat{\sigma}^y, \hat{\mathbb{I}} \otimes \hat{\sigma}^z\}$ with the measurements of $\hat{\mathbb{I}} \otimes \hat{\sigma}^x$ and $\hat{\mathbb{I}} \otimes \hat{\sigma}^y$ facilitated by the rotation operator $\hat{A}_2^{(x)}$ and $\hat{A}_2^{(y)}$ respectively. Conversely, the theoretical expectation values are defined as 

\begin{equation}
\expval{\hat{\sigma}_\pm^\alpha(k)} := \mel{\varphi_\pm(k)}{\hat{\sigma}^\alpha}{\varphi_\pm(k)}, 
\label{eq:quantum_hardware_simulation/two-band/theoretical_observables_expectation}
\end{equation}
where $\ket{\varphi_\pm(k)}$ are obtained by diagonalizing \cref{eq:quantum_hardware_simulation/two-band_twister_model/H2_twister_model}. 

In total, 600 circuits with $t = 20$ were simulated for a knot/link as there are three observables, 100 $k$ points, and two bands. Each circuit yields a distribution of bitstrings upon measurements in  \cref{eq:quantum_hardware_simulation/two-band/measured_observables_expectation}. Postselection is performed and the distribution is rescaled to account for unphysical bitstrings that are discarded. The Bloch angles are then computed using \cref{eq:quantum_hardware_simulation/two-band/Bloch_angles} to reconstruct the eigenstates \cref{eq:quantum_hardware_simulation/two-band/reconstructed_eigenstate}. 

With the measured observables  $\expval{\hat{\sigma}^x_\pm(k)}$, $\expval{\hat{\sigma}^y_\pm(k)}$, and $\expval{\hat{\sigma}^z_\pm(k)}$, we obtain the normalized reconstructed eigenstates 

\begin{equation}
\ket{\psi_\pm(k)} = \mqty[
\cos(\theta_\pm(k)/2) \\
\sin(\theta_\pm(k)/2)e^{i\phi_\pm(k)}
], 
\label{eq:quantum_hardware_simulation/two-band/reconstructed_eigenstate}
\end{equation}
which is a two-component specialization of \cref{eq:results/quantum_circuit_hardware_implementation/reconstruction/psi} with 

\begin{equation}
\theta_\pm(k) = \cos^{-1}\expval{\hat{\sigma}^z_\pm(k)}, \quad \text{and} \quad \phi_\pm(k) = \tan^{-1}\frac{\expval{\hat{\sigma}^y_\pm(k)}}{\expval{\hat{\sigma}^x_\pm(k)}}.
\label{eq:quantum_hardware_simulation/two-band/Bloch_angles}
\end{equation} 
For the two-band model, the two energies are inversion-symmetric about the origin, so $\Lambda_+(k) = -\Lambda_-(k)$. Thus, we have 

\begin{equation}
\Lambda(k) = \frac{1}{4}p_+(k)p_-(k) = \frac{-im_0}{(1+m_1)^2}E(k) 
\label{eq:quantum_hardware_simulation/two-band/Lambda}
\end{equation}
as a concrete expression of \cref{eq:results/quantum_circuit_hardware_implementation/characterizing_braids_and_knotted_structures/Lambda} with $p_+(k)p_-(k)$ the explicit form of $f(\cdot)$, and 
\begin{equation}
p_\pm(k) = \frac{\expval{\hat{\sigma}^x_+(k)} \pm i\expval{\hat{\sigma}^y_+(k)}}{1 - \expval{\hat{\sigma}^z_+(k)}} \pm \frac{\expval{\hat{\sigma}^x_-(k)} \pm i\expval{\hat{\sigma}^y_-(k)}}{1 - \expval{\hat{\sigma}^z_-(k)}}. 
\label{eq:quantum_hardware_simulation/two-band/p}
\end{equation}
In \cref{eq:quantum_hardware_simulation/two-band/p}, the observables $\{\hat{O}_b\} = \{\hat{\sigma}^x, \hat{\sigma}^y, \hat{\sigma}^z\}$ are chosen to extract $E(k)$ from $\ket{\psi(k)}$. As the reconstructed eigenstates $\ket{\psi_\pm(k)}$ evolve in $k$, they trace out trajectories in the space of spin expectation values. The upper row of \cref{fig:H2_twister_band_structure_bloch_sphere_W} shows the good agreement between $\ket{\psi_\pm(k)}$ \cref{eq:quantum_hardware_simulation/two-band/reconstructed_eigenstate} and the theoretical eigenstates $\ket{\varphi_\pm(k)}$ obtained by diagonalizing \cref{eq:quantum_hardware_simulation/two-band_twister_model/H2_twister_model}, thus validating our eigenstate-preparation protocol. 

Taking $N = 2$ braid strands in \cref{eq:results/quantum_circuit_hardware_implementation/characterizing_braids_and_knotted_structures/partial_winding}, the winding of $\Lambda_+(k)$ around $\Lambda_-(k)$ takes the form $W_{+-}(k) = \frac{1}{2\pi i}\int_{0}^k \, \partial_{k'}\ln\Lambda(k')dk'$ which for discretized computations can be written as   
\begin{equation}
W_{+-}(k) \approx \frac{1}{2\pi} \sum_{k'=0}^k \ln\bigg[\frac{p_+(k'+\Delta k')p_-(k'+\Delta k')}{p_+(k')p_-(k')}\bigg], 
\label{eq:quantum_hardware_simulation/two-band_twister_model/winding_number_discretized}
\end{equation}
where $\Delta k'$ is the difference between each point. See the Methods section for a detailed derivation of \cref{eq:quantum_hardware_simulation/two-band_twister_model/winding_number_discretized} and we show it in the bottom row of \cref{fig:H2_twister_band_structure_bloch_sphere_W}. Observe that the measured winding numbers agree well with theory and have hence captured the braiding information as desired. 

The winding number records the braid crossings and we now elaborate on how the braid word can be extracted from it for this simplest two-strand context. Since there are only two strands, the braid group $B_2$ ($N = 2$ in \cref{eq:results/knots_and_braids/braid_group_presentation}) is generated by a single braid generator $\sigma_1$. We begin with the Hopf link, which has two crossings as shown in the eigenstate trajectories and the band structure in \cref{fig:H2_twister_band_structure_bloch_sphere_W}a(i) and \cref{fig:H2_twister_band_structure_bloch_sphere_W}a(ii) respectively. Picking a projection plane parallel to $\Lambda_+(k) - \Lambda_-(k)$ (3-vector), the crossings (black dashed lines in \cref{fig:H2_twister_band_structure_bloch_sphere_W}a(iii)) occur at $W_{+-}(k_{\mathrm{crossing}}) = 1/4$ or $W_{+-}(k_{\mathrm{crossing}}) = 3/4$ with phase difference exactly $\pi/2$. At $W_{+-}(k_{\mathrm{crossing}}) = 1/4$, we have $\tau^{+1}_{12}\Longrightarrow\sigma_1$ since $r = 0$ in \cref{eq:results/quantum_circuit_implementation/characterizing_braids_and_knotted_structures/W_crossing_condition}. At $W_{+-}(k_{\mathrm{crossing}}) = 3/4$, we have $r = 1$. Since the bands have swapped at $W_{+-}(k_{\mathrm{crossing}}) = 1/4$, we have $\tau_{21}^{-1} = \tau_{12}^{+1} \Longrightarrow \sigma_1$ which yields the braid word $\sigma_1^2$. 

In contrast, the unknot has three crossings. Two crossings manifest in the eigenstate trajectories as loops and the band structure in \cref{fig:H2_twister_band_structure_bloch_sphere_W}b(i) and \cref{fig:H2_twister_band_structure_bloch_sphere_W}b(ii) respectively since bands $+$ and $-$ at $k = 0$ are mapped to bands $-$ and $+$ at $k = 2\pi$. As the eigenstate trajectory at $k = 0$ coincides with the other trajectory at $k = 2\pi$, these two trajectories join to form a single loop upon braid closure. A similar analysis follows and we have the braid word $\sigma_1\sigma_1^{-1}\sigma_1 = \sigma_1$ since the two crossings cancel out. 

Finally, the unlink is trivial in that there are no crossings as the bands do not wind around each other, and the eigenstate trajectories do not intersect as shown in \cref{fig:H2_twister_band_structure_bloch_sphere_W}c(i) and \cref{fig:H2_twister_band_structure_bloch_sphere_W}c(ii). Thus, $W_{+-}(k)$ in \cref{fig:H2_twister_band_structure_bloch_sphere_W}c(iii) is trivial and the braid word is empty. Here, the two strands do not join, and each forms an independent loop upon braid closure. 

For $N > 2$, it becomes more challenging to reconstruct the complete set of eigenstates and the $\Lambda(k)$ expression is also less obvious to obtain as the energies are not inversion-symmetric about the origin. Nevertheless, we show below that this can still be reliably performed.

\begin{figure*}[htbp!]
\centering
\includegraphics[width=\linewidth]{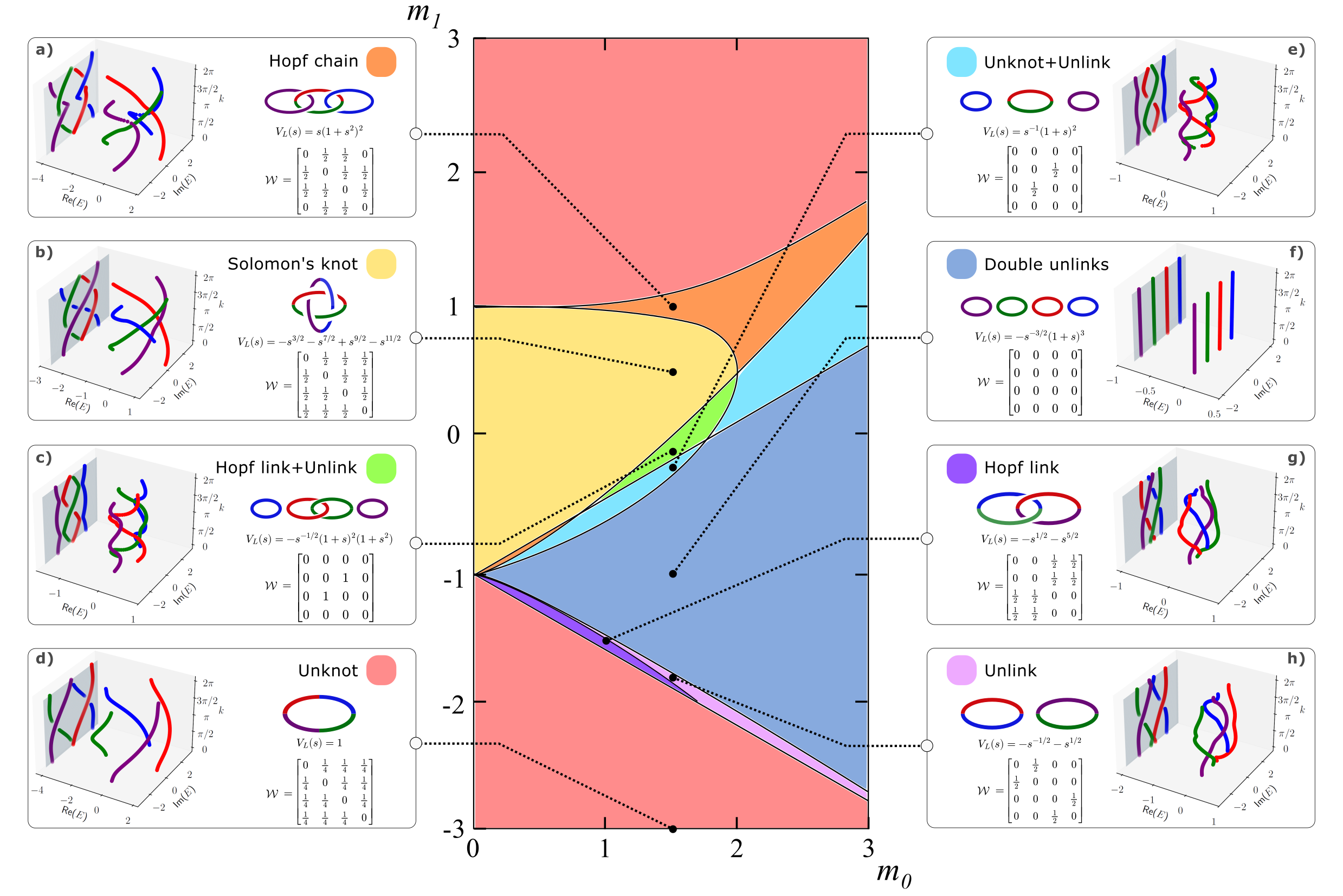}
\caption{\textbf{Variety of topologically inequivalent braids of the four-band twister model \cref{eq:quantum_hardware_simulation/four-band_twister_model/H4_twister_model}.} There are eight regions, each corresponding to a braid, and the boundaries are given in \cref{eq:methods/phase_diagrams/four-band/phase_boundaries}. The Jones polynomials $V_L(s)$ for a knot/link $L$, and their winding number matrices \cref{eq:results/quantum_circuit_hardware_implementation/characterizing_braids_and_knotted_structures/winding_number_matrix} are shown for each braid. Both serve as topological indices that uniquely distinguish the knots and links, although this property does not hold in general. The parameters used are $(m_0, m_1) = (1.5, 1)$, $(1.5, 0.5)$, $(1.5, -0.08)$, $(1.5, -3)$, $(1.5, -0.18)$, $(1.5, -1)$, $(1, -1.5)$, $(1.5, -1.8)$ for the \textbf{a)} Hopf chain, \textbf{b)} Solomon's knot, \textbf{c)} Hopf link + unlink, \textbf{d)} unknot, \textbf{e)} unknot + unlink, \textbf{f)} double unlinks, \textbf{g)} Hopf link, and \textbf{h)} unlink respectively.}
\label{fig:H4_phase_diagram}
\end{figure*}

\tocless\subsubsection{Measuring Braids with $N = 4$ Strands}

\begin{figure*}[htbp!]
\centering
\includegraphics[width=\linewidth]{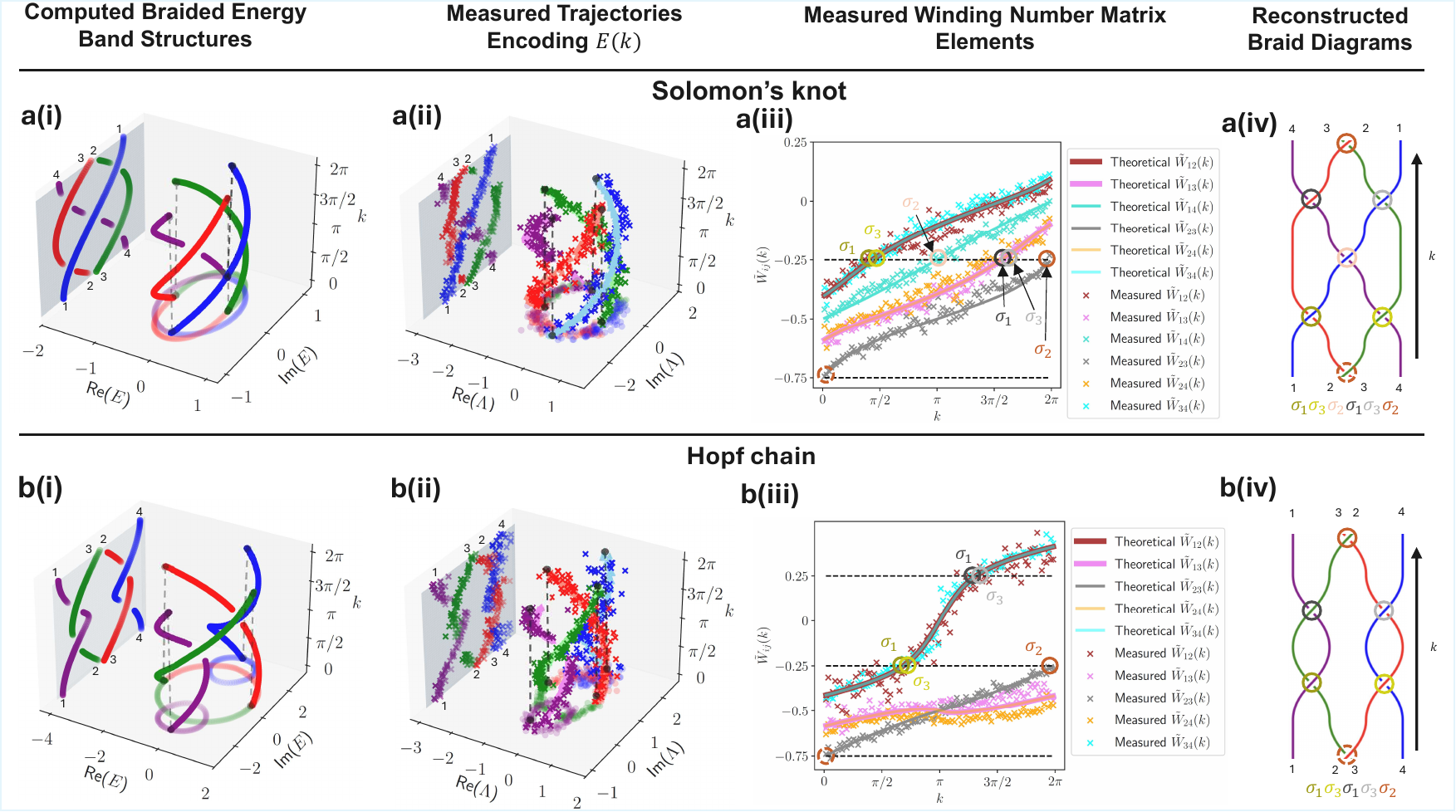}
\caption{\textbf{Characterization of braids via the elements of the phase-shifted measured winding numbers $\tilde{W}_{ij}(k)$ \cref{eq:quantum_hardware_simulation/four-band/phase_shifted_W} for the four-band twister model \cref{eq:quantum_hardware_simulation/four-band_twister_model/H4_twister_model}.} \textbf{a(i)} The band structure is computed by diagonalizing \cref{eq:quantum_hardware_simulation/four-band_twister_model/H4_twister_model} with parameters $(m_0, m_1) = (-0.5, -0.4)$. \textbf{a(ii)} The projected eigenstate trajectories \cref{eq:quantum_hardware_simulation/four-band/Lambda} are in good agreement with \textbf{a(i)}. \textbf{a(iii)} Measured winding number matrix elements $\tilde{W}_{ij}(k)$ after phase-shifting \cref{eq:quantum_hardware_simulation/four-band/winding_number_element_discretized} on \texttt{ibm\textunderscore marrakesh}. The six crossings (solid colored circles) in the evolution of $\tilde{W}_{ij}(k)$ are detected whenever $\tilde{W}_{ij}(k) = -1/4$ (black dashed lines); the black dashed lines correspond to \cref{eq:results/quantum_circuit_implementation/characterizing_braids_and_knotted_structures/W_crossing_condition}. \textbf{a(iv)} The braid diagram is reconstructed from the identified crossings. \textbf{b(i)} Same as \textbf{a(i)} with parameters $(m_0, m_1) = (2, 1.1)$. \textbf{b(ii)} The projected eigenstate trajectories are in good agreement with \textbf{b(i)}. \textbf{b(iii)} Same as \textbf{a(iii)} with five crossings instead. \textbf{b(iv)} Same as \textbf{a(iv)}. Note that the grey curve in both cases intersects the black dashed lines at $\tilde{W}_{ij}(k) = -1/4$ and $\tilde{W}_{ij}(k) = -3/4$. Due to periodicity, only one crossing should be chosen; in this case, we choose the second crossing near $k = 2\pi$. Importantly, note that the labels $i$ and $j$ in $\tilde{W}_{ij}(k)$ correspond to the initial labels at $k = 0$ before any permutation of the band indices.}
\label{fig:H4_twister_solomon_knot_hopf_chain_band_structure_braid_diagram_W}
\end{figure*}

We next demonstrate how four-strand braided band structures can be reconstructed on quantum hardware. Beyond requiring an additional qubit, having four eigenstates to measure presents a host of additional challenges. The most obvious challenge is that some target eigenstates are not extremal in $\Im(E)$, and thus cannot be obtained through naive $e^{\pm i\hat{H}(k)t}$ evolution. Additional measured observables are also necessary to compute the Bloch angles to parameterize an eigenstate. Furthermore, an analytic expression of a four-component eigenstate is complicated in general, thus extracting $E(k)$ via the expectation values of $\{\hat{O}_b\}$ to obtain $\Lambda(k)$ is non-trivial (see \cref{app:four-band/measurement_protocol} in the Supplementary Material). Below, we describe how these challenges can be overcome.

Importantly, having more strands confers far more sophisticated mathematical possibilities than  \cref{eq:quantum_hardware_simulation/two-band_twister_model/H2_twister_model}. With $N = 4$ strands (eigenstates), \cref{fig:H4_phase_diagram} reveals a rich variety of topologically distinct knotted band structures of the four-band twister model (c.f. \cref{eq:results/knots_and_braids/braiding/twister_model,eq:results/knots_and_braids/braiding/Tn_twister_model})

\begin{equation}
\begin{aligned}
\hat{H}^{(4)}(k) &= 
im_0\hat{\Sigma}^{(4)} + m_1 \hat{T}_1^{(4)}(k) + \hat{T}_2^{(4)}(k) \\
&= \mqty[
im_0 & 0 & 0 & e^{2ik} + e^{ik}m_1 \\
1 + m_1 & \frac{im_0}{3} & 0 & 0 \\
0 & 1 + m_1 & -\frac{im_0}{3} & 0 \\
0 & 0 & 1 + m_1 & -im_0
], 
\label{eq:quantum_hardware_simulation/four-band_twister_model/H4_twister_model}
\end{aligned}
\end{equation}
where $m_0, m_1\in\mathbb{R}$ and $\hat{\Sigma}^{(4)} = \diag(1, 1/3, -1/3, -1)$. The branch points in its energy dispersion $E_{a,b}(k) = \frac{a}{3}\sqrt{-5m_0^2 +b \sqrt{16m_0^4 + 81e^{ik}(m_1+1)^3(m_1+e^{ik})}},\ a, b \in\{+1,-1\}$ give rise to intricate multi-valuedness whose holonomy properties are exactly captured by the braid structure.

To highlight a few examples, \cref{eq:quantum_hardware_simulation/four-band_twister_model/H4_twister_model} hosts a braid that corresponds to the Hopf chain with five crossings in \cref{fig:H4_phase_diagram}a, with parameters $(m_0, m_1) = (1.5, 1)$, which generalizes the Hopf link to three interlocking rings. The most complicated example is shown in \cref{fig:H4_phase_diagram}b that corresponds to Solomon's knot with parameters $(m_0, m_1) = (1.5, 0.5)$, and it has six crossings. In both examples, $\mathcal{W}_{ij} = 1/2$ as a band does not return to its initial position after one period (see \cref{fig:methodology_illustration}b). There exist other variations that comprise different combinations of the Hopf link with the unlink and unknot, as shown in \cref{fig:H4_phase_diagram}. 

The boundaries (see \cref{eq:methods/phase_diagrams/four-band/phase_boundaries} in Methods) that demarcate topologically inequivalent braids are given by the loci of gap closure $E_{i, j}(k) = 0$. Each link can be characterized by the Alexander polynomial knot invariant $\Delta_L(s)$, which can be computed directly from the braid word via the Burau representation \cite{Alexander1928} \cref{eq:knot_polynomials/computation_of_knot_polynomials/Alexander_polynomial} (see \cref{app:knot_polynomials/computation_of_the_Alexander_polynomials} in the Supplementary Material for its definition and computational details). Note, however, that the Alexander polynomial (like any other single knot invariant) cannot uniquely identify all possible knots: here, in particular, it vanishes for disjoint links. But together with the winding number matrix \cref{eq:results/quantum_circuit_hardware_implementation/characterizing_braids_and_knotted_structures/winding_number_matrix} (also shown in \cref{fig:H4_phase_diagram}), all eight inequivalent knots/links shown can be distinctly identified. 

To reconstruct the eigenstates of \cref{eq:quantum_hardware_simulation/four-band_twister_model/H4_twister_model} on a quantum processor and compute the winding number matrix, note that the four strands are encoded in two qubits, together with an ancilla qubit to postselect for non-unitary dynamics. As aforementioned, reconstructing all eigenstates in a four-band model is challenging. Unlike the two-band twister model, we will need to reconstruct the eigenstate trajectories on the Bloch hypersphere (or construct Majorana stars \cite{majorana1932atomi}) with six independent angles. Instead of \cref{eq:results/quantum_circuit_hardware_implementation/obtaining/U_H}, we define $\hat{U}_H = e^{-i \hat{H}_\lambda^{(4)}(k)t}$ where $\hat{H}_\lambda^{(4)} = e^{i\lambda}\hat{H}^{(4)}(k)$ with $\lambda\in[0, 2\pi]$ (see the discussion around \cref{eq:methods/obtaining_eigenstates/rotated_Hamiltonian} in Methods) to access the complete set of eigenstates. In total, 4800 quantum circuits with $t = 20$ were simulated for each knot/link in \cref{fig:H4_phase_diagram} except the unknot + unlink with $t = 25$.

Among several equivalent possible choices, the expectation values of the most convenient set of measured observables $\{\hat{O}'_a\}$ for reconstructing an eigenstate is (see Methods for the derivation of \cref{eq:quantum_hardware_simulation/four-band/measured_observables_expectation})

\begingroup
\allowdisplaybreaks 
\begin{align}
\expval{\hat{\sigma}^\alpha_i(k)}_{+-,A} &= \mel*{\psi_i^{'(\alpha)}(k)}{(\hat{\mathbb{I}}\otimes\ket{0}\bra{0}\otimes\hat{\sigma}^z)}{\psi_i^{'(\alpha)}(k)} \notag \\
&\xrightarrow[]{\text{PS}} \abs*{\braket*{000}{\psi_i^{'(\alpha)}(k)}}^2 - \abs*{\braket*{001}{\psi_i^{'(\alpha)}(k)}}^2, \notag \\
\expval{\hat{\sigma}^\alpha_i(k)}_{+-,B} &= \mel*{\psi_i^{'(\alpha)}(k)}{(\hat{\mathbb{I}}\otimes\ket{1}\bra{1}\otimes\hat{\sigma}^z)}{\psi_i^{'(\alpha)}(k)} \notag \\
&\xrightarrow[]{\text{PS}} \abs*{\braket*{010}{\psi_i^{'(\alpha)}(k)}}^2 - \abs*{\braket*{011}{\psi_i^{'(\alpha)}(k)}}^2, \notag \\
\expval{\hat{\sigma}^\alpha_i(k)}_{-,AB} &= \mel*{\psi_i^{'(\alpha)}(k)}{(\hat{\mathbb{I}}\otimes\hat{\sigma}^z\otimes\ket{1}\bra{1})}{\psi_i^{'(\alpha)}(k)} \notag \\
&\xrightarrow[]{\text{PS}} \abs*{\braket*{001}{\psi_i^{'(\alpha)}(k)}}^2 - \abs*{\braket*{011}{\psi_i^{'(\alpha)}(k)}}^2, \notag \\
\expval{\hat{\sigma}^\alpha_i(k)}_{+-,AB} &= \mel*{\psi_i^{'(\alpha)}(k)}{(\hat{\mathbb{I}}\otimes\hat{\sigma}^z\otimes\hat{\mathbb{I}})}{\psi_i^{'(\alpha)}(k)} \notag \\
&\xrightarrow[]{\text{PS}} \abs*{\braket*{000}{\psi_i^{'(\alpha)}(k)}}^2 + \abs*{\braket*{001}{\psi_i^{'(\alpha)}(k)}}^2 \notag \\
&\quad -\abs*{\braket*{010}{\psi_i^{'(\alpha)}(k)}}^2 - \abs*{\braket*{011}{\psi_i^{'(\alpha)}(k)}}^2, \label{eq:quantum_hardware_simulation/four-band/measured_observables_expectation}
\end{align}
\endgroup
where $i$ labels each strand. $\ket*{\psi_i^{'(\alpha)}(k)} = (\hat{\mathbb{I}}_1\otimes\hat{A}_2^{(\alpha)} \otimes \hat{A}_3^{(\alpha)})\hat{U}_i(k)\ket{000}$ is a specialization of \cref{eq:results/quantum_circuit_hardware_implementation/reconstruction/psi_prime} to two qubits and shown in \cref{fig:methodology_illustration}a, and $\hat{A}_i^{(\alpha)} \in \{\hat{R}^y_i(-\pi/2), \hat{R}^x_i(\pi/2), \hat{\mathbb{I}}_i\}$ with $\alpha \in \{x, y, z\}$. After measuring these expectation values, we can obtain the six Bloch angles

\begin{equation}
\begin{aligned}
\theta_{\mu, i}(k) &= \cos^{-1}\expval{\hat{\sigma}_i^z(k)}_{+-,\mu}, \\
\phi_{\nu, i}(k) &= -\tan^{-1}\frac{\expval{\hat{\sigma}_i^y(k)}_{+-,\nu}}{\expval{\hat{\sigma}_i^x(k)}_{+-,\nu}}, \\
\phi_{AB, i}(k) &= -\tan^{-1}\frac{\expval{\hat{\sigma}_i^y(k)}_{-,AB}}{\expval{\hat{\sigma}_i^x(k)}_{-,AB}}, 
\label{eq:quantum_hardware_simulation/four-band/Bloch_angles}
\end{aligned}
\end{equation}
where $\mu\in\{A, B, AB\}$ and $\nu\in\{A, B\}$. Together, these quantities allow us to reconstruct a normalized eigenstate 

\begin{equation}
\ket{\psi_i(k)} = 
\mqty[
\cos\frac{\theta_{A, i}(k)}{2}\cos\frac{\theta_{AB, i}(k)}{2}e^{i\phi_{A, i}(k)}e^{i\phi_{AB, i}(k)} \\
\sin\frac{\theta_{A, i}(k)}{2}\cos\frac{\theta_{AB, i}(k)}{2}e^{i\phi_{AB, i}(k)} \\
\cos\frac{\theta_{B, i}(k)}{2}\sin\frac{\theta_{AB, i}(k)}{2}e^{i\phi_{B, i}(k)} \\
\sin\frac{\theta_{B, i}(k)}{2}\sin\frac{\theta_{AB, i}(k)}{2}
]. 
\label{eq:quantum_hardware_simulation/four-band/reconstructed_eigenstate}
\end{equation}
We remind the reader that the observables in \cref{eq:quantum_hardware_simulation/four-band/measured_observables_expectation} are precisely $\{\hat{O}'_a\}$ in \cref{fig:methodology_illustration}a, chosen to compute \cref{eq:quantum_hardware_simulation/four-band/Bloch_angles}. 

Now, we also need to examine the observables $\{\hat{O}_b\} = \{(\hat{\mathbb{I}} - \hat{\sigma}^z)\otimes\hat{\sigma}^x, (\hat{\sigma}^z - \hat{\mathbb{I}})\otimes\hat{\sigma}^y, (\hat{\mathbb{I}} - \hat{\sigma}^z)\otimes(\hat{\mathbb{I}} - \hat{\sigma}^z)\}$ for the purpose of extracting the braiding spectrum $E_i(k)$ from $\ket{\psi_i(k)}$. By computing $\mel{\psi_i(k)}{\hat{O}_b}{\psi_i(k)}$, we get the trajectories 

\begin{equation}
\Lambda_i(k) = \frac{X_i(k) + iY_i(k)}{Z_i(k)} = \frac{3c_3(k) + E_i(k)}{2c_0(k)}, 
\label{eq:quantum_hardware_simulation/four-band/Lambda}
\end{equation}
where $X_i(k) = \mel*{\psi_i(k)}{(\hat{\mathbb{I}} - \hat{\sigma}^z)\otimes\hat{\sigma}^x}{\psi_i(k)}$, $Y_i(k) = \mel*{\psi_i(k)}{(\hat{\sigma}^z - \hat{\mathbb{I}})\otimes\hat{\sigma}^y}{\psi_i(k)}$, $Z_i(k) = \mel*{\psi_i(k)}{(\hat{\mathbb{I}} - \hat{\sigma}^z)\otimes(\hat{\mathbb{I}} - \hat{\sigma}^z)}{\psi_i(k)}$, $c_0(k)\in i\mathbb{R}$, and $c_3(k)\in\mathbb{C}$. Importantly, as shown in \cref{app:four-band/measurement_protocol} of the Supplementary Material, our setup possesses the useful property

\begin{equation}
\Lambda_i(k) -\Lambda_j(k) \propto E_i(k)-E_j(k)
\end{equation}
for any pair of bands $i$ and $j$. Hence, the winding between $E_i(k)$ and $E_j(k)$ can be measured and quantified by measuring the eigenstate windings through the difference $\Lambda_i(k) - \Lambda_j(k)$. Since $\{\Lambda_i(k)\}$ are able to reproduce the braiding information in $\{E_i(k)\}$,  \cref{eq:results/quantum_circuit_hardware_implementation/characterizing_braids_and_knotted_structures/partial_winding} can be adapted to yield 

\begin{equation}
\begin{aligned}
W_{ij}(k) &\approx \frac{1}{n} \sum_{a = 1}^{n} \frac{1}{2\pi} \\
&\quad \times \sum_{k'=0}^k \ln\bigg[\frac{\Lambda_i^{(a)}(k'+\Delta k') - \Lambda_j^{(a)}(k'+\Delta k')}{\Lambda_i^{(a)}(k') - \Lambda_j^{(a)}(k')}\bigg], 
\label{eq:quantum_hardware_simulation/four-band/winding_number_element_discretized}
\end{aligned}
\end{equation}
which is a generalization of \cref{eq:quantum_hardware_simulation/two-band_twister_model/winding_number_discretized} as \cref{eq:quantum_hardware_simulation/four-band/winding_number_element_discretized} is the general form for $N$-band models. For $N = 2$ when the energies are inversion-symmetric about the origin, the difference $\Lambda_i(k) - \Lambda_j(k)$ is equivalent to the difference with respect to the origin. 

\cref{fig:H4_twister_solomon_knot_hopf_chain_band_structure_braid_diagram_W} illustrates how we characterize the braids via the elements of the phase-shifted measured winding numbers $\tilde{W}_{ij}(k)$ (c.f. \cref{eq:results/quantum_circuit_hardware_implementation/characterizing_braids_and_knotted_structures/partial_winding}) which we define as 

\begin{equation}
\tilde{W}_{ij}(k) := W_{ij}(k) - \frac{\pi/2 - \chi_{ij}(0)}{2\pi},  
\label{eq:quantum_hardware_simulation/four-band/phase_shifted_W}
\end{equation}
where $\chi_{ij}(0) = \arg(\Lambda_i(0) - \Lambda_j(0))$, and choosing $\pi/2$ corresponds to a plane that is parallel to the $\Im(\Lambda)$ plane. In general, the phase difference between each pair of bands is different. Importantly, to have well-defined winding numbers, the phase differences must be with respect to the same projection plane (braid diagram) so that the braid crossings are accurately detected by \cref{eq:results/quantum_circuit_implementation/characterizing_braids_and_knotted_structures/W_crossing_condition} (see the discussion regarding \cref{eq:extraction_of_braid_words/phase_shifted_W} in Methods for the slightly more general form of \cref{eq:quantum_hardware_simulation/four-band/phase_shifted_W}). 

We now consider Solomon's knot as an example. The band structure in \cref{fig:H4_twister_solomon_knot_hopf_chain_band_structure_braid_diagram_W}a(i) shows the braid diagram and the braid that forms Solomon's knot upon braid closure. Our measurement protocol reproduces the braiding information of $\{E_i(k)\}$ in the projected eigenstate trajectories $\{\Lambda_i(k)\}$ as shown in \cref{fig:H4_twister_solomon_knot_hopf_chain_band_structure_braid_diagram_W}a(ii). Observe that there is good agreement with \cref{fig:H4_twister_solomon_knot_hopf_chain_band_structure_braid_diagram_W}a(i) in that the braids are correctly encoded in $\{\Lambda_i(k)\}$ which results in an accurate reproduction of the crossings. This is vital as we can form the braid word either by reading off directly from the braid diagram in \cref{fig:H4_twister_solomon_knot_hopf_chain_band_structure_braid_diagram_W}a(ii) or from the evolution of $\tilde{W}_{ij}(k)$ in \cref{fig:H4_twister_solomon_knot_hopf_chain_band_structure_braid_diagram_W}a(iii). 

The braid group $B_4$ ($N = 4$ in \cref{eq:results/knots_and_braids/braid_group_presentation}) for Solomon's knot is generated by the three braid generators $\sigma_1$, $\sigma_2$, and $\sigma_3$. A crossing in \cref{fig:H4_twister_solomon_knot_hopf_chain_band_structure_braid_diagram_W}a(ii) is recorded in \cref{fig:H4_twister_solomon_knot_hopf_chain_band_structure_braid_diagram_W}a(iii) whenever $\tilde{W}_{ij}(k_{\mathrm{crossing}}) = -1/4$. As such, the six crossings in \cref{fig:H4_twister_solomon_knot_hopf_chain_band_structure_braid_diagram_W}a(ii) correspond to the six colored circles in \cref{fig:H4_twister_solomon_knot_hopf_chain_band_structure_braid_diagram_W}a(iii). Since a one-to-one correspondence is made, we can appropriately identify each crossing with a braid generator $\sigma_i$ to form a braid word, thereby facilitating the reconstruction of the braid diagram in \cref{fig:H4_twister_solomon_knot_hopf_chain_band_structure_braid_diagram_W}a(iv). 

The first crossing is between bands 1 and 2 (blue and red in \cref{fig:H4_twister_solomon_knot_hopf_chain_band_structure_braid_diagram_W}a(ii)) such that the band indices are permuted as $(1, 2) \to (2, 1)$. Since $\tilde{W}_{12}(k_{\mathrm{crossing}}) = -1/4$, we have $r = -1$ in \cref{eq:results/quantum_circuit_implementation/characterizing_braids_and_knotted_structures/W_crossing_condition}, and thus \cref{eq:results/quantum_circuit_implementation/characterizing_braids_and_knotted_structures/crossing_to_braid_generator} becomes $\tau_{12}^{-1} \Longrightarrow \sigma_1^{-1\cdot(-1)^{\delta_{44}}} = \sigma_1$. The second crossing swaps bands 3 and 4 (green and purple) such that $(3, 4) \to (4, 3)$, and we have $\sigma_3$. Relabeling after two crossings, the band indices for blue and purple become $(1, 4) \to (2, 3)$ at the third crossing, hence we have $\sigma_2$. At the fourth crossing, the red and purple band indices are $(2, 4) \to (1, 2)$ so we have $\sigma_1$. At the fifth crossing, the blue and green band indices are $(1, 3) \to (3, 4)$ so we have $\sigma_3$. At the final crossing, the red and green band indices become $(2, 3) \to (2, 3)$ so we have $\sigma_2$. Although it seems that there are two crossings at $\tilde{W}_{23}(k_{\mathrm{crossing}}) = -3/4$ and $\tilde{W}_{23}(k_{\mathrm{crossing}}) = -1/4$, these occur near to $k = 0$ and $k = 2\pi$. As such, they are identified together as $k$ has a period of $2\pi$ and we choose the crossing at $\tilde{W}_{ij}(2\pi) = -1/4$. Hence, the braid word is $\sigma_1\sigma_3\sigma_2\sigma_1\sigma_3\sigma_2$. 

A similar analysis follows for the Hopf chain, where we also observe good agreement between \cref{fig:H4_twister_solomon_knot_hopf_chain_band_structure_braid_diagram_W}b(i) and \cref{fig:H4_twister_solomon_knot_hopf_chain_band_structure_braid_diagram_W}b(ii). The braid word is extracted from \cref{fig:H4_twister_solomon_knot_hopf_chain_band_structure_braid_diagram_W}b(iii) and the correspondence for each crossing is shown in \cref{fig:H4_twister_solomon_knot_hopf_chain_band_structure_braid_diagram_W}b(iv). Repeating the same analysis gives $\sigma_1\sigma_3\sigma_1\sigma_3\sigma_2$ for the Hopf chain. Through the braid word, the exact Jones polynomial can be uniquely computed for its braid closure (see \cref{app:knot_polynomials/computation_of_the_Jones_polynomials} in the Supplementary Material for computational details.)

Several remarks are in order. For the twister models we consider, the winding numbers are sufficient to characterize the braids although this is \emph{not} true in general as it is not a knot invariant like the Jones polynomial. Furthermore, visualizing the band structure and the projected eigenstate trajectories are unnecessary in obtaining the braid word; rather, \emph{only} the evolution of $\tilde{W}_{ij}(k)$ in \cref{fig:H4_twister_solomon_knot_hopf_chain_band_structure_braid_diagram_W}a(iii) and \cref{fig:H4_twister_solomon_knot_hopf_chain_band_structure_braid_diagram_W}b(iii) is necessary. The crucial step is that one must carefully keep track of the band indices after each crossing so that the correct braid generator is ascribed to it. See the section ``From the $k$-dependent Measured Winding Numbers to the Braid Word'' in Methods for the detailed procedure. 

\tocless\section{Discussion}
\label{sec:discussion}

In this work, we proposed an eigenstate-based measurement protocol and demonstrated it on a digital quantum computer for a family of non-Hermitian twister models. This is the first implementation on a quantum processor that demonstrates physical access to braided spectral structures comprising up to four strands, whose braid closures yield non-trivial links such as Solomon's knot and the Hopf chain. In particular, we developed a non-variational protocol that reconstructs the complete eigenstate manifold of multi-band non-Hermitian systems deterministically. Our approach deterministically resolves the global eigenstate structure of the Hamiltonian, enabling direct access to the complete spectral manifold across multiple bands.  This further allows us to obtain the winding number matrix by directly measuring a set of Pauli observables from the reconstructed eigenstates, which can already encode sufficient information of the spectral windings without complete spectral characterization or iterative circuit optimization procedures, thereby substantially reducing the associated computational overhead compared to alternative approaches such as quantum phase estimation.

Having obtained the $k$-dependent pairwise winding numbers, they are systematically mapped to braid words from the exchange structure of eigenstate trajectories. As such, knot invariants, including the Alexander and Jones polynomials, are computed via the extracted braid words. In principle, these ingredients establish an experimentally accessible pipeline that bridges quantum state characterization and knot theory, allowing braiding topology to be systematically measured even when direct energy spectroscopy is impractical or inaccessible. 

Looking forward, our framework opens a new route to physically exploring braid topology in regimes that remain experimentally elusive on other established platforms~\cite{PhysRevLett.127.090501,PhysRevLett.130.163001}. Several directions may further extend the scope and impact of this framework. On the algorithmic side, reducing classical post-processing overhead through techniques such as classical shadow tomography~\cite{huang2020predicting} or machine learning~\cite{kauffman2020rectangularknotdiagramsclassification,yu2021experimental,chen2024machine} could enable more scalable extractions of topological data directly from measurements on a quantum processor. Integrating these protocols with emerging quantum error correction architectures~\cite{google2023suppressing,google2025quantum} will be essential for accessing models beyond the four-band twister and/or deeper circuits with controlled fidelity. 

Beyond methodological improvements, our approach opens avenues toward qualitatively new physics, e.g., experimental access to categorified invariants such as Khovanov homology~\cite{schmidhuber2025quantumalgorithmkhovanovhomology}, and the exploration of engineering and probing non-Abelian braiding and topological structures in programmable digital quantum processors~\cite{Xu2023Digital,xu2024non,wang2025simulatingnonabelianstatisticsparafermions}. In general, these prospects suggest that quantum processors are not only capable of simulating topology, but also of discovering and manipulating new forms of topological structures beyond conventional spectral characterization.  

\vspace{0.5cm}
{\begin{center}
    \textbf{AUTHOR CONTRIBUTIONS}
\end{center}
\vspace{0.2cm}
T. Y. N and Y. W contributed equally to this work. C. H. L proposed the initial idea, advised on the implementations and is the overall supervisor. T. Y. N implemented the formalism on quantum hardware and performed the numerical simulations. Y. W solved the generalized twister model, developed the formalism to measure winding numbers and extract braid words on quantum hardware. W. J. C contributed to the code used in the numerical simulations and the computation of knot polynomials. Y. W and T. C developed the rotation method for multi-band models. All authors contributed to the writing of the manuscript. 

\vspace{0.5cm}
{\begin{center}
    \textbf{DATA AVAILABILITY}
\end{center}
\vspace{0.2cm}
The data used in this work is available from the corresponding authors upon reasonable request. 

\vspace{0.5cm}
{\begin{center}
    \textbf{CODE AVAILABILITY}
\end{center}
\vspace{0.2cm}
The code used in this work is available from the corresponding authors upon reasonable request. 

\vspace{0.5cm}
{\begin{center}
    \textbf{ACKNOWLEDGEMENTS}
\end{center}
\vspace{0.2cm}
This research is supported by the Ministry of Education, Singapore (MOE award number: MOE-T2EP50224-0021 (WBS no. A-80035050100)). T.Y.N is supported by the A*STAR Graduate Academy. T.~C. acknowledges partial support from the A*STAR Quantum
Innovation Centre (Q.~InC) Strategic Research and Translational Thrust. We acknowledge the use of IBM Quantum services in
this work. The views expressed are those of the authors
and do not reflect the official policy or position of IBM
or the IBM Quantum team.

\tocless\section{Methods}

\tocless\subsection{Obtaining the Eigenstates of Multi-Band Non-Hermitian Hamiltonians on Quantum Hardware}

To reconstruct the braid structure, we need the complete set of eigenstates of $\hat{H}^{(N)}(k)$ on a periodic grid. However, for a fixed initial state, the non-unitary time-evolution operator in \cref{eq:results/quantum_circuit_hardware_implementation/obtaining/U_H} generally converges to only one target eigenstate at each sampled $k$. Repeating the same preparation protocol is therefore insufficient to recover the complete eigenstate manifold required for the multi-strand braid.

To selectively access different eigenstates while preserving the braid structure, we introduce a global phase rotation of the Hamiltonian
\begin{equation}
\hat{H}_{\lambda}^{(N)}(k) = e^{i\lambda}\hat{H}^{(N)}(k),
\label{eq:methods/obtaining_eigenstates/rotated_Hamiltonian}
\end{equation}
where $\lambda\in[0,2\pi]$ is a rotation angle in the complex-energy plane. This transformation rotates the energy spectrum but leaves the eigenstates unchanged. It therefore changes which eigenstate is preferentially prepared under non-unitary time evolution, while preserving the underlying braid structure. 

For each sampled point $k_j$, we then consider the parameterized state
\begin{equation}
\ket{\psi(k_j,\lambda)} = e^{-i\hat{H}_{\lambda}^{(N)}(k_j)t}\ket{0}^{\otimes M},
\end{equation}
where $M = \ceil{\log_2 N}$. In principle, one could sweep $\lambda$ directly on hardware and identify distinct prepared states from the measurement outcomes. In practice, such a brute-force sweep would be experimentally costly. We therefore determine the optimal rotation angles classically and then use the resulting values in the quantum simulation. Let $\{\ket{\varphi_i(k_j)}\}_{i=1}^N$ denote the right eigenstates obtained by diagonalizing $\hat{H}^{(N)}(k_j)$. For the $i$-th eigenstate at $k_j$, we choose the optimal rotation angle $\lambda_{i,j}^{\mathrm{opt}}$ by maximizing the overlap 

\begin{equation}
\lambda_{i,j}^{\mathrm{opt}}  = \arg\max_{\lambda_i\in[0,2\pi]} \abs{\braket{\psi(k_j,\lambda_i)}{\varphi_i(k_j)}}.
\label{eq:methods/obtaining_eigenstates/maximizing_lambda}
\end{equation}
Repeating this procedure over the grid gives the set
\begin{equation}
\vb*{\lambda}_i^{\mathrm{opt}} = \{\lambda_{i,1}^{\mathrm{opt}},\lambda_{i,2}^{\mathrm{opt}},\ldots,\lambda_{i,N_k}^{\mathrm{opt}}\},
\end{equation}
for the $i$-th eigenstate, where $N_k$ is the number of sampled points on the grid. \footnote{An explicit example for Solomon's knot is shown in \cref{fig:H4_twister_solomon_knot_overlap}.}

Using these optimized angles, we define the eigenstate-resolved non-unitary evolution operator
\begin{equation}
\hat{U}_H^{(i)}(k_j)=e^{-i\hat{H}_{\lambda_{i,j}^{\mathrm{opt}}}^{(N)}(k_j)t},
\label{eq:methods/obtaining_eigenstates/U_H}
\end{equation}
which prepares the $i$-th eigenstate at $k_j$. Applying this construction to all $i = 1, \ldots, N$ at all sampled points $k_j$ yields the complete set of reconstructed eigenstates needed for the braid analysis. We remark that as $N$ increases, it becomes harder (though still possible) to find $\vb*{\lambda}_i^{\mathrm{opt}}$ as the range of valid angles is increasingly smaller.

\begin{figure}[htbp!]
\centering
\includegraphics[width=\linewidth]{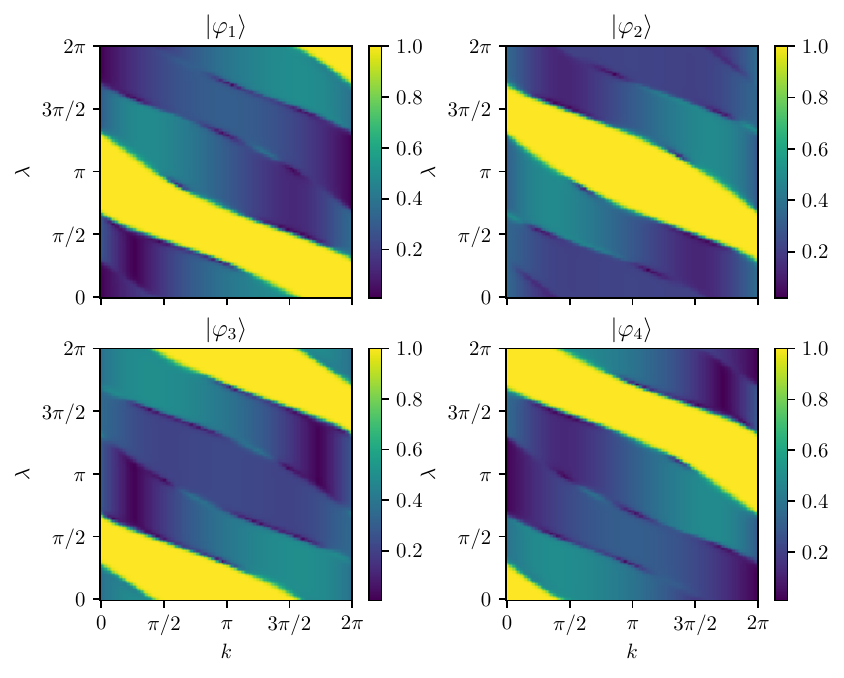}
\caption{\textbf{Determining the rotation angles for Solomon's knot.} In each plot, the yellow regions denote the angles $\lambda(k)$ that yield the largest overlap $\abs{\braket{\psi(k_j, \lambda_i^{\mathrm{opt}})}{\varphi_i(k_j)}}$ (see \cref{eq:methods/obtaining_eigenstates/maximizing_lambda}) where $\ket{\psi(k_j, \lambda_i^{\mathrm{opt}})}$ is the classically evolved state with $t = 20$ and $\ket{\varphi_i(k_j)}$ is the true eigenstate obtained by diagonalizing \cref{eq:quantum_hardware_simulation/four-band_twister_model/H4_twister_model}. The parameters used are $(m_0, m_1) = (-0.5, -0.4)$.}
\label{fig:H4_twister_solomon_knot_overlap}
\end{figure}

\tocless\subsection{Two-band Twister Model}
Here, we detail the measurement protocol that allows us to extract $E_\pm(k)$ from the set of measured observables $\{\hat{\sigma}^x, \hat{\sigma}^y, \hat{\sigma}^z\}$, and we state the boundaries of the topologically inequivalent regions. The two-band twister model is 

\begin{equation}
\hat{H}^{(2)}(k) = \mqty[
im_0 & e^{2ik} + e^{ik}m_1 \\
1 + m_1 & -im_0
],
\end{equation}
and its eigenvalues are

\begin{equation}
E_\pm(k) = \pm\sqrt{e^{2ik}(m_1+1) + m_1e^{ik}(m_1+1)-m_0^2}, 
\label{eq:methods/two-band/eigenvalues}
\end{equation}
where $\abs{e^{ik}} = 1$.

\tocless\subsubsection{Measurement Protocol}

Any normalized eigenstate of $\hat{H}^{(2)}(k)$ at $k$ can be written as a two-component spinor. Up to an overall global phase, such a spinor can be parameterized by a point on the Bloch sphere with polar angle $\theta(k)$ and azimuthal angle $\phi(k)$. This representation is particularly useful because the Bloch sphere is naturally embedded in the space of expectation values of the Pauli operators. In fact, for a pure two-level state, the expectation values of the Pauli operators form a unit vector that specifies the corresponding point on the Bloch sphere. The Bloch angles $\theta(k)$ and $\phi(k)$ are given by substituting the spinor representation \cref{eq:quantum_hardware_simulation/two-band/reconstructed_eigenstate} into $\expval{\hat{\sigma}^{x,y,z}(k)}$ and solving as shown in \cref{eq:quantum_hardware_simulation/two-band/Bloch_angles}.

For the two-band twister model, the winding number can be extracted directly from the eigenstate expectation values without reconstructing the full spectrum. Since $E_-(k)=-E_+(k)$, there is only one independent band pair 

\begin{equation}
W_{+-}(k) = \frac{1}{2\pi i}\int_0^k dk\,\partial_{k'} \ln \Lambda(k'),
\label{eq:methods/measurement_protocol/two-band/winding_number}
\end{equation}
where $\Lambda(k')$ is a complex quantity constructed from measurable observables. 

A generic non-Hermitian two-band model can be written as
\begin{equation}
\hat{H}^{(2)}_\text{generic}(k)=\vb{d}(k)\cdot\vb*{\hat{\sigma}},
\end{equation}
where $\vb{d}(k)=[d_x(k),d_y(k),d_z(k)]^T\in\mathbb{C}^3$ and
$d_+(k)=d_-^*(k)=d_x(k)+id_y(k)\in\mathbb{C}$. For each eigenstate, we reconstruct the Bloch vector from the measured Pauli expectation values

\begin{equation}
\vb{L}_\pm(k)=
\bigl[
\expval{\hat{\sigma}^x_\pm(k)},
\expval{\hat{\sigma}^y_\pm(k)},
\expval{\hat{\sigma}^z_\pm(k)}
\bigr]^T.
\end{equation}
Using the stereographic projection described in the Supplementary Material, these observables are mapped to complex trajectories for the two eigenstates. From them, we construct the symmetric and antisymmetric combinations $\mathcal{L}_\pm(k)$, which for the twister model reduce to~\footnote{The detailed derivation and more explanations are given in \cref{app:two-band/measurement_protocol}.}
\begin{equation}
\mathcal{L}_+(k)=\frac{d_z^*(k)}{d_-(k)},
\qquad
\mathcal{L}_-(k)=\frac{E^*(k)}{d_-(k)}.
\end{equation}
We then define
\begin{equation}
\Lambda(k)=\mathcal{L}_+(k)\mathcal{L}_-^*(k)
=\frac{d_z^*(k)E(k)}{d_-(k)d_+(k)}.
\end{equation}
Specializing to the parameters in $\hat{H}^{(2)}(k)$ yields $\Lambda(k) = \frac{1}{4}p_+(k)p_-(k)$ as given in \cref{eq:quantum_hardware_simulation/two-band/Lambda}. For the twister model in \cref{eq:results/knots_and_braids/braiding/twister_model}, $d_z(k)$ is a constant and $d_-(k)d_+(k)\in\mathbb{R}$, so only the phase of $\Lambda(k)$ contributes to $W_{+-}(k)$. Thus, the winding of $\Lambda(k)$ faithfully reproduces the band winding. In the simulations, \cref{eq:methods/measurement_protocol/two-band/winding_number} is discretized as described in \cref{eq:quantum_hardware_simulation/two-band_twister_model/winding_number_discretized}. 

\tocless\subsubsection{Topologically Inequivalent Regions in Parameter Space}

The degeneracy in \cref{eq:methods/two-band/eigenvalues} occurs at $E_+(k) = E_-(k)$ or equivalently, $E_+(k) = 0$ since $E_-(k) = -E_+(k)$. Solving $E_+(k) = 0$ gives the boundaries that demarcate each topologically inequivalent braid. The boundaries are defined by a special point at $(m_0, m_1) = (0, -1)$, and 

\begin{equation}
\begin{aligned}
F_1(m_0, m_1) &= (m_1+1)^2 - m_0^2 = 0, \\
F_2(m_0, m_1) &= m_1^2-1+m_0^2 = 0, \\
F_3(m_0, m_1) &= 1 + m_0^2 + m_1 = 0.
\label{eq:methods/phase_diagrams/two-band/phase_boundaries}
\end{aligned}
\end{equation}
See \cref{app:two-band/phase_diagram} in the Supplementary Material for the derivation of \cref{eq:methods/phase_diagrams/two-band/phase_boundaries}. 

\begin{figure}[htbp!]
\centering
\includegraphics[width=\linewidth]{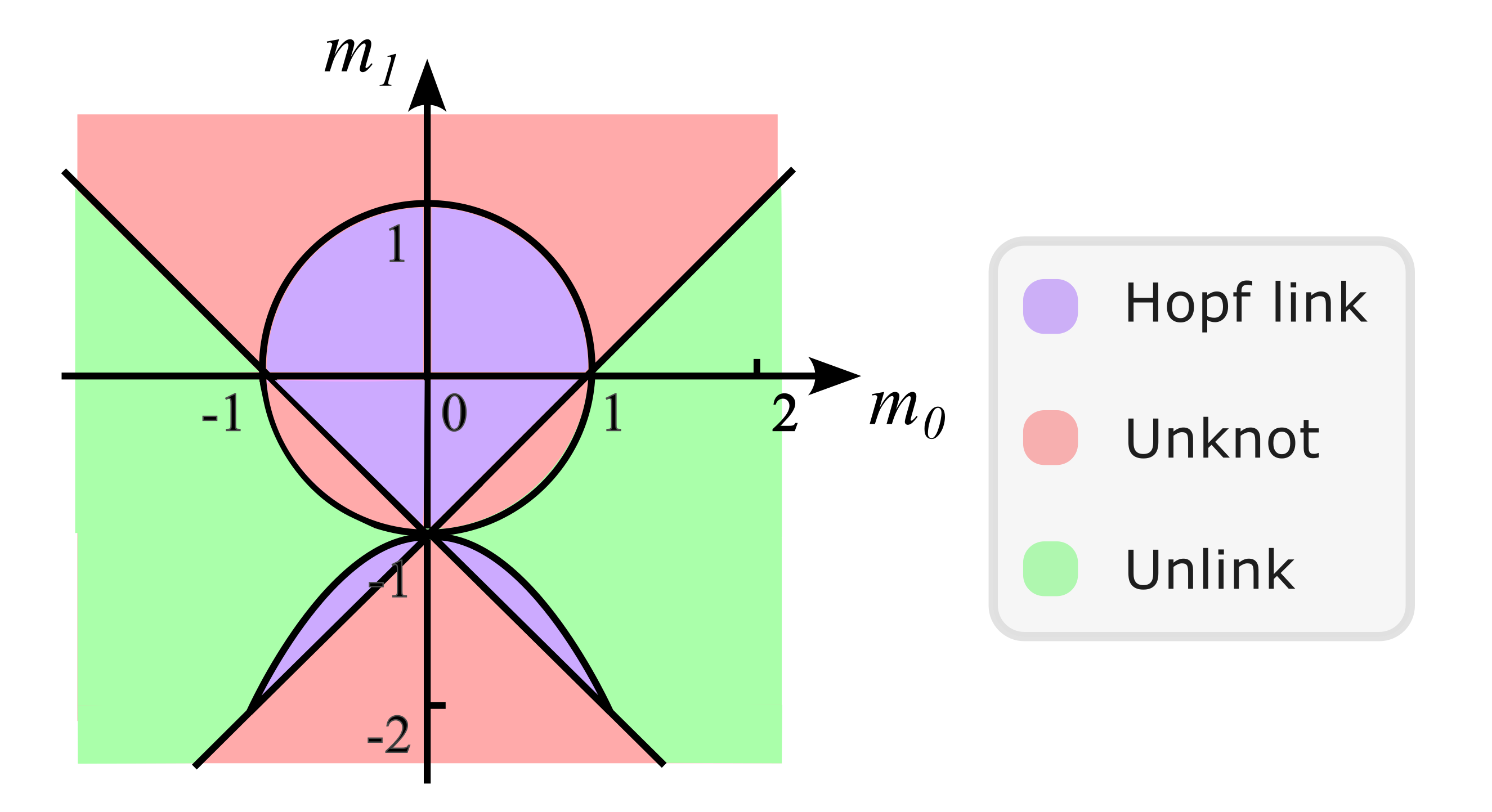}
\caption{\textbf{Variety of topologically inequivalent braids of the two-band twister model \cref{eq:quantum_hardware_simulation/two-band_twister_model/H2_twister_model}.} The boundaries are given in \cref{eq:methods/phase_diagrams/two-band/phase_boundaries}, and \cref{eq:methods/measurement_protocol/two-band/winding_number} distinguishes each braid. The analogous diagram for $N = 4$ braid strands is far more complicated and a particular parameter cross-section is presented in \cref{fig:H4_phase_diagram} of the main text.}
\label{fig:H2_phase_diagram}
\end{figure}

\tocless\subsection{Four-band Twister Model}
Unlike the two-band twister model, the set of observables to be measured (\cref{eq:quantum_hardware_simulation/four-band/measured_observables_expectation}) are not restricted to the Pauli operators only. Thereafter, we explain how to extract $E_i(k)$ from $\ket{\psi_i(k)}$ by taking expectation values of a set of observables $\{\hat{O}_b\}$ to get $\Lambda_i(k)$. Finally, we state the boundaries of the topologically inequivalent regions. Here we examine the four-band twister model 

\begin{equation}
\hat{H}^{(4)}(k) = \mqty[
im_0 & 0 & 0 & e^{2ik} + e^{ik}m_1 \\
1 + m_1 & \frac{im_0}{3} & 0 & 0 \\
0 & 1 + m_1 & -\frac{im_0}{3} & 0 \\
0 & 0 & 1 + m_1 & -im_0
],
\end{equation}
whose eigenvalues are 
\begin{equation}
E_{\pm,\pm}(k) = \pm\frac{1}{3}\sqrt{-5m_0^2 \pm S(k)},
\label{eq:methods/four-band/eigenvalues}
\end{equation}
where $S(k) = \sqrt{16m_0^4 + 81e^{ik}(m_1+1)^3(m_1+e^{ik})}$. 

\tocless\subsubsection{Measurement Protocol}

For the four-band model, the eigenstates reside in a two-qubit Hilbert space and can be written as a four-component spinor
$\ket{\psi(k)}=[\psi_A^+,\psi_A^-,\psi_B^+,\psi_B^-]^T$, where each component admits the polar form
$\psi_\alpha^{\beta} = c_\alpha^{\beta} e^{i\phi_\alpha^{\beta}}$ with $c_\alpha^{\beta}, \phi_\alpha^{\beta} \in \mathbb{R}$, $\alpha\in\{A,B\}$ and $\beta \in\{+,-\}$.
After imposing normalization and discarding an overall global phase, a generic pure state is specified by six real parameters. 

A convenient parametrization is obtained by organizing these parameters hierarchically, in direct analogy with the Bloch sphere representation for a two-level system. Concretely, we first separate the amplitudes into two sectors associated with the first qubit, and parameterize their total weights by an angle $\theta_{AB}(k)$

\begin{equation}
\begin{aligned}
\sqrt{|\psi_A^+|^2+|\psi_A^-|^2} &= \cos\frac{\theta_{AB}(k)}{2}, \\ 
\sqrt{|\psi_B^+|^2+|\psi_B^-|^2} &= \sin\frac{\theta_{AB}(k)}{2}.
\end{aligned}
\end{equation}
Within each sector, the remaining internal structure is that of an effective two-level system, so the normalized two-component spinors can be parameterized by Bloch angles $(\theta_A(k),\phi_A(k))$ and $(\theta_B(k),\phi_B(k))$

\begin{equation}
\begin{aligned}
\mqty[\psi_A^+\\ \psi_A^-]
&=
\cos\frac{\theta_{AB}(k)}{2}
\mqty[
\cos(\theta_A(k)/2)e^{i\phi_A(k)}\\[2pt]
\sin(\theta_A(k)/2)
],\\
\mqty[\psi_B^+\\ \psi_B^-]
&=
\sin\frac{\theta_{AB}(k)}{2}
\mqty[
\cos(\theta_B(k)/2)e^{i\phi_B(k)}\\[2pt]
\sin(\theta_B(k)/2)
].
\end{aligned}
\end{equation}
Finally, an additional relative phase between the two sectors is encoded by an inter-sector phase $\phi_{AB}(k)$, which we attach to the $A$-sector components. Adopting a gauge in which the final component is real-valued, this construction yields the parametric form in \cref{eq:quantum_hardware_simulation/four-band/reconstructed_eigenstate}.

The logic behind \cref{eq:quantum_hardware_simulation/four-band/Bloch_angles} closely parallels the reconstruction procedure for the two-band model, and the key idea is to apply a similar Bloch sphere reconstruction to effective two-level subsystems obtained by conditioning on the state of the other qubit. Specifically, by projecting onto the $A$ or $B$ sector of the first qubit, the remaining amplitudes $(\psi_\alpha^+,\psi_\alpha^-)$ with $\alpha\in\{A,B\}$ form an effective two-component spinor. Conditional measurements of the operators $\hat{\sigma}^{x,y,z}_{+-, \mu}$ where $\mu\in\{A, B, AB\}$, therefore reconstruct the Bloch vector of this spinor allowing one to extract the angles $(\theta_A,\phi_A)$ and $(\theta_B,\phi_B)$ exactly as in the two-band case. 

To determine the remaining inter-sector parameters, we instead condition on the spin state of the second qubit. Projecting onto $\beta\in\{+,-\}$ isolates the amplitudes $(\psi_A^\beta,\psi_B^\beta)$, which again behave as a two-component spinor. Measuring $\hat{\sigma}^{x,y,z}_{-,AB}$ therefore determines the relative phase between the $A$ and $B$ sectors, yielding the angle $\phi_{AB}$. The operator $\hat{\sigma}^{z}_{+-, AB}$ measures the population imbalance between the two sectors and fixes the polar angle $\theta_{AB}$.

Now that we can extract the eigenstates, our goal is to associate each eigenstate $\ket{\psi_i(k)}$ with a complex quantity whose winding in the Brillouin zone reproduces the winding between band energies. The key observation is that, for a broader class of Hamiltonians including the four-band twister model, one component of the eigenstate can be made affine in the energy by a convenient gauge choice. Fixing the last component of $\ket{\psi_i(k)}$ to be real, we write
\begin{equation}
\ket{\psi_i(k)}=
\frac{1}{\mathcal{N}_i(k)}
\mqty[
\cdot \\
\cdot \\
\psi_i^{(3)}(k) \\
1
],
\end{equation}
where $\mathcal{N}_i(k)$ is a normalization factor. Under the twister model constraints, the third component takes the form
\begin{equation}
\psi_i^{(3)}(k) = \frac{1}{\mathcal{N}_i(k)}\frac{3c_3(k)+E_i(k)}{2},
\end{equation}
so differences between the bands eliminate the constant factors and become directly proportional to the energy differences

\begin{equation}
\mathcal{N}_i(k)\psi_i^{(3)}(k)-
\mathcal{N}_j(k)\psi_j^{(3)}(k)
\propto
E_i(k)-E_j(k).
\end{equation}
To construct an experimentally accessible quantity free of normalization factors, we measure a set of observables $\{O_i\}$

\begin{equation}
\begin{aligned}
X_i(k) &= \mel{\psi_i(k)}{(\hat{\mathbb{I}}-\hat{\sigma}^z)\otimes\hat{\sigma}^x}{\psi_i(k)}, \\
Y_i(k) &= \mel{\psi_i(k)}{(\hat{\sigma}^z-\hat{\mathbb{I}})\otimes\hat{\sigma}^y}{\psi_i(k)}, \\
Z_i(k) &= \mel{\psi_i(k)}{(\hat{\mathbb{I}}-\hat{\sigma}^z)\otimes(\hat{\mathbb{I}}-\hat{\sigma}^z)}{\psi_i(k)}.
\end{aligned}
\end{equation}
This defines the generalized stereographic projection
\begin{equation}
\Lambda_i(k)=\frac{X_i(k)+iY_i(k)}{Z_i(k)} = \frac{3c_3(k) + E_i(k)}{2c_0(k)},
\label{eq:methods/measurement_protocol/four-band/Lambda}
\end{equation}
whose phase tracks the phase of $E_i(k)$ up to a factor that does not contribute to the winding number. Consequently, the winding of $\Lambda_i(k) - \Lambda_j(k)$ is tantamount to the winding of $E_i(k) - E_j(k)$, thus enabling the band winding numbers to be extracted directly from measurable observables. 

\tocless\subsubsection{Topologically Inequivalent Regions in Parameter Space}

There are two degeneracies in \cref{eq:methods/four-band/eigenvalues} which occur at $S(k) = 0$ and $-5m_0^2 \pm S(k) = 0$. Solving the former and latter equations yields the first and second sets of equations in \cref{eq:methods/phase_diagrams/four-band/phase_boundaries}}. Six equations define the boundaries shown in \cref{fig:H4_phase_diagram}, and a degenerate point at $(m_0, m_1) = (0, -1)$ where all six equations intersect. To summarize, these six equations are 

\begin{equation}
\begin{aligned}
F_1(m_0, m_1) &= 16m_0^4 + 81(m_1+1)^4 = 0, \\
F_2(m_0, m_1) &= 16m_0^4 + 81(m_1+1)^3(1-m_1) = 0, \\
F_3(m_0, m_1) &= 16m_0^4 - 81(m_1+1)^3 = 0, \\
G_1(m_0, m_1) &= m_0^4 - 9(m_1+1)^4 = 0, \\
G_2(m_0, m_1) &= m_0^4 - 9(m_0+1)^3(1-m_1) = 0, \\
G_3(m_0, m_1) &= m_0^4 + 9(m_1+1)^3 = 0. 
\label{eq:methods/phase_diagrams/four-band/phase_boundaries}
\end{aligned}
\end{equation}

\tocless\subsection{Extraction of the Winding Number Matrix}

Once the eigenstates have been reconstructed, the remaining task is to extract the two pieces of information needed to characterize the braid structure: the permutation of bands after one period, and the pairwise winding numbers $\bar{W}_{ij}$ (see below \cref{eq:results/quantum_circuit_hardware_implementation/characterizing_braids_and_knotted_structures/partial_winding}) between bands over the first period. Together, these determine the winding number matrix $\mathcal{W}$ \cref{eq:results/quantum_circuit_hardware_implementation/characterizing_braids_and_knotted_structures/winding_number_matrix} introduced in the main text. On quantum hardware, the permutation matrix is obtained directly from overlaps between reconstructed eigenstates, while the winding number matrix can be computed efficiently from the data within a single period.

\tocless\subsubsection{Obtaining the Braid Permutation Matrix on Quantum Hardware}

Given the reconstructed eigenstates $\ket{\psi_j(k)}$, we first form the overlap matrix between the initial and final states over one period

\begin{equation}
P'_{sj} = \abs{\braket{\psi_s(0)}{\psi_j(2\pi)}}.
\end{equation}
The permutation matrix $P$ is then obtained by retaining, for each final band index $j$, only the dominant overlap

\begin{equation}
P_{ij} = \vb{1}\!\left(P'_{ij}=\max_s P'_{sj}\right),
\label{permut_matrix}
\end{equation}
where $\vb{1}(\cdot)$ is the indicator function. This procedure suppresses spurious switching caused by small but finite overlaps between distinct eigenstates and retains only the most likely band pairing after one period. With $P$ determined, one can then organize the pairwise band windings over subsequent periods.

\tocless\subsubsection{Efficient Computation of the Winding Number Matrix}

We distinguish between three related winding quantities in this work. The first is the $k$-dependent winding matrix, $W_{ij}(k)$, which describes the accumulated winding between bands $i$ and $j$ up to the parameter value $k$; the second is the total winding over one period, $\bar{W}_{ij}$, which is a $k$-independent but perturbation-sensitive quantity partially characterizing the braid structure; the third is the topological winding number matrix $\mathcal{W}$ shown in \cref{fig:H4_phase_diagram}, and we will introduce their connections and differences below.

Considering the winding information in each period is equivalent, to reduce computational cost, we compute the pairwise winding evolution only over the first period \footnote{Replace \cref{eq:methods/extraction_of_the_winding_number_matrix/efficient_computation/winding_number_one_period} with the discretized versions \cref{eq:quantum_hardware_simulation/two-band_twister_model/winding_number_discretized,eq:quantum_hardware_simulation/four-band/winding_number_element_discretized} accordingly and set $k = 2\pi$.} (set $k = 2\pi$ in $W_{ij}(k)$ defined in \cref{eq:results/quantum_circuit_hardware_implementation/characterizing_braids_and_knotted_structures/partial_winding})

\begin{equation}
\bar{W}_{ij} = \frac{1}{2\pi i}\int_{0}^{2\pi} dk'\,\partial_{k'}\ln\!\left[\Lambda_i(k')-\Lambda_j(k')\right],
\label{eq:methods/extraction_of_the_winding_number_matrix/efficient_computation/winding_number_one_period}
\end{equation}
and generate the $k$-independent winding number matrices for the remaining $n-1$ periods by permuting the band labels according to $P$. Since all periods contain the same information up to this relabeling, the winding number matrix can be written as

\begin{equation}
\mathcal{W} = \frac{1}{n}\sum_{a=0}^{n-1}(P^{-1})^a \bar{W} P^a,
\label{eq:methods/extraction_of_the_winding_number_matrix/efficient_computation/winding_number_matrix}
\end{equation}
where $P$ is the permutation matrix defined in \cref{permut_matrix}, which satisfies $P^n = \mathbb{I}$. In this way, the winding information over all periods is obtained from the first-period data together with the permutation structure extracted from the reconstructed eigenstates. Furthermore, our approach can be adapted to more general models as long as one can extract linear functionals of $E_i(k)$ from the eigenstate components. 

\tocless\subsection{From the $k$-dependent Measured Winding Numbers to the Braid Word}
\label{sec:methods/extraction_of_braid_words}

Projecting the three-dimensional worldlines (strands) of complex braiding eigenvalues onto a particular plane gives a braid diagram. To visualize the concepts explained herein, see \cref{fig:H4_twister_solomon_knot_hopf_chain_band_structure_braid_diagram_W} for the braid diagram of Solomon's knot and the Hopf chain of the four-band twister model \cref{eq:quantum_hardware_simulation/four-band_twister_model/H4_twister_model}.

In the following, we will treat $\Lambda_i(k)$ and $\Lambda_j(k)$ either as complex numbers or real-valued 3-vectors with $k$ as the third component. At $k = 0$, we define $\omega_{ij}(0) := \Lambda_i(0) - \Lambda_j(0)$ which is a 3-vector that is parallel to the band structure. The argument of $\omega_{ij}(0)$ is defined as $\chi_{ij}(0) := \arg(\omega_{ij}(0))$ and this defines the projection plane as $\omega_{ij}(0)$ is parallel to it. In this work, we consider the $\Re(\Lambda)$-$k$ plane to be the reference for rotations. Thus, by definition, $W_{ij}(k)$ is defined according to the projection plane which has an angle $\chi_{ij}(0)$ from the $\Re(\Lambda)$ axis.

We define the phase-shifted winding number as 

\begin{equation}
\tilde{W}_{ij}(k) := W_{ij}(k) - \frac{\chi_{pq}(0) - \chi_{ij}(0)}{2\pi}
\label{eq:extraction_of_braid_words/phase_shifted_W}
\end{equation}
for a pair of eigenstates $i$ and $j$. There are two possible choices of $\chi_{pq}(0)$ in general. If $\chi_{pq}(0) = \chi$ for all pairs, a fixed projection plane is chosen. For example, we chose a fixed projection plane such that $\chi = \pi/2$ in \cref{fig:H4_twister_solomon_knot_hopf_chain_band_structure_braid_diagram_W} (c.f. \cref{eq:quantum_hardware_simulation/four-band/phase_shifted_W}). Otherwise, a pair of bands $i$ and $j$ may be chosen such that $\chi_{ij}(0)$ implicitly defines the projection plane thus reducing  \cref{eq:extraction_of_braid_words/phase_shifted_W} to $\tilde{W}_{ij}(k) = W_{ij}(k)$. However, \cref{eq:extraction_of_braid_words/phase_shifted_W} still holds for other bands $p$ and $q$ for $i \neq p$ and $j \neq q$ as the $\chi_{ij}(0)/2\pi$ factor compensates for the phase differences from different pairs of bands $i$ and $j$. We emphasize that  \emph{any} projection plane may be chosen to reconstruct a braid diagram. 

Conceptually, $\{\omega_{ij}(k)\}$ contains the vectors along $\Lambda_i(k)$ and $\Lambda_j(k)$. Since $\omega_{ij}(0)$ is parallel to the projection plane, we seek some $\omega_{ij}(k)$ whose projection on the plane is a point. This is equivalent to solving for the set of $k$ points that satisfy $\omega_{ij}(0)\cdot\omega_{ij}^\perp(k) = 0$. In other words, $\omega_{ij}^\perp(k)$ is the normal vector of the projection plane. Intuitively, the angle between $\omega_{ij}(0)$ and $\omega_{ij}^\perp(k)$ is precisely $\tilde{W}_{ij}(k)$ as it encodes the winding information between two bands. However, in practice, we only need to check the phase-shifted winding number $\tilde{W}_{ij}(k)$

\begin{equation}
\tilde{W}_{ij}(k_{\mathrm{crossing}}) = \frac{1}{2\pi}\bigg(\frac{\pi}{2} + r\pi\bigg) = \frac{1}{4} + \frac{r}{2},
\label{eq:extraction_of_braid_words/phase_shifted_W_crossing_condition}
\end{equation}
where $r\in\mathbb{Z}$, to identify the crossings. Geometrically, we have $\pi/2$ by the above discussion, $r\pi$ denotes the multiple possible normal vectors of the projection plane, and we divide by $2\pi$ due to the periodicity of the Hamiltonian. 

There is a relationship between $r$ in \cref{eq:extraction_of_braid_words/phase_shifted_W_crossing_condition} and the braid generators in a braid word. Given two bands $i$ and $i+1$, the crossing information is encoded in  $\tau_{\pi(i)\pi(i+1)}^{(-1)^r}$ and thus the braid generator $\sigma_{\pi(i)}^{\pm 1}$ is recorded as

\begin{equation}
\tau_{\pi(i)\pi(i+1)}^{\pm 1} = \tau_{\pi(i+1)\pi(i)}^{\mp 1} \Longrightarrow \sigma_{\pi(i)}^{\pm 1\cdot  (-1)^{\delta_{4N}}}, 
\label{eq:extraction_of_braid_words/braid_word_extraction}
\end{equation}
where $\delta_{4N}$ is the Kronecker delta, $N$ is the number of bands, and $\pi(\cdot)$ denotes the permutation of an index. 

We now have the necessary ingredients to extract the braid words from the $\tilde{W}(k)$ plots obtained from the simulations. First, we identify the $k$ points in $\tilde{W}_{ij}(k)$ that satisfy \cref{eq:extraction_of_braid_words/phase_shifted_W_crossing_condition} and note the bands that braid or equivalently cross each other on the braid diagram. Braiding permutes the band indices and we need to relabel the band indices after each crossing to track the braiding evolution. At each crossing, we record the braid generator associated with the crossing, and we repeat this for all crossings for $k\in[0, 2\pi]$, thus resulting in the braid word. A summary of braid words associated with the braids in the two- and four-band twister models is shown in \cref{tab:knot_link_braid_words_polynomials}. 

We now explicitly detail the formation of the braid word for Solomon's knot from \cref{fig:H4_twister_solomon_knot_hopf_chain_band_structure_braid_diagram_W}a(iii) and bold the labels that are permuted. At each crossing, we have 

\begin{enumerate}
    \item The brown and cyan curves both intersect with the black dashed line at a similar $k$ point but the order of the crossings do not matter. Let's consider the braiding between bands 1 and 2 (brown curve) first. The band labels become 
    \begin{equation*}
    \begin{aligned}
    \mqty{\vb{(1, 2)} \\ (1, 3) \\ (1, 4) \\ (2, 3) \\ (2, 4) \\ (3, 4)} \to \mqty{\vb{(2, 1)} \\ (2, 3) \\ (2, 4) \\ (1, 3) \\ (1, 4) \\ (3, 4)} \qquad &\Rightarrow \tilde{W}_{12}(k_{\mathrm{crossing}}) = -1/4 \\
    &\Rightarrow \sigma_1^{-1\cdot(-1)^{\delta_{44}}} = \sigma_1.
    \end{aligned}
    \end{equation*} \\
    All labels are permuted by swapping 1 and 2, and we used \cref{eq:results/quantum_circuit_implementation/characterizing_braids_and_knotted_structures/W_crossing_condition,eq:results/quantum_circuit_implementation/characterizing_braids_and_knotted_structures/crossing_to_braid_generator}. 
    \item Next is the cyan curve. We have 
    \begin{equation*}
    \mqty{(2, 1) \\ (2, 3) \\ (2, 4) \\ (1, 3) \\ (1, 4) \\ \vb{(3, 4)}} \to \mqty{(2, 1) \\ (2, 4) \\ (2, 3) \\ (1, 4) \\ (1, 3) \\ \vb{(4, 3)}} \qquad \Rightarrow \tilde{W}_{34}(k_{\mathrm{crossing}}) = -1/4 \Rightarrow \sigma_3.
    \end{equation*}
    \item Next is the red curve. We have 
    
    \begin{equation*}
    \mqty{(2, 1) \\ (2, 4) \\ \vb{(2, 3)} \\ (1, 4) \\ (1, 3) \\ (4, 3)} \to \mqty{(3, 1) \\ (3, 4) \\ \vb{(3, 2)} \\ (1, 4) \\ (1, 2) \\ (4, 2)} \qquad \Rightarrow \tilde{W}_{14}(k_{\mathrm{crossing}}) = -1/4 \Rightarrow \sigma_2.
    \end{equation*}
    Note that the labels $i$ and $j$ in $\tilde{W}_{ij}(k)$ are based on the original labels at $k = 0$ as shown in \cref{fig:H4_twister_solomon_knot_hopf_chain_band_structure_braid_diagram_W}a(iii). In particular, the third crossing is between the original band labels $1$ and $4$, which have been permuted such that $(1, 4) \to (2, 4) \to (2, 3)$. Thus, we know that the braid generator must be $\sigma_2$. 
    \item Next is the yellow curve. We have
    \begin{equation*}
    \mqty{(3, 1) \\ (3, 4) \\ (3, 2) \\ (1, 4) \\ \vb{(1, 2)} \\ (4, 2)} \to \mqty{(3, 2) \\ (3, 4) \\ (3, 1) \\ (2, 4) \\ \vb{(2, 1)} \\ (4, 1)} \qquad \Rightarrow \tilde{W}_{24}(k_{\mathrm{crossing}}) = -1/4 \Rightarrow \sigma_1.
    \end{equation*}
    \item Next is the blue curve. We have
    \begin{equation*}
    \mqty{(3, 2) \\ \vb{(3, 4)} \\ (3, 1) \\ (2, 4) \\ (2, 1) \\ (4, 1)} \to \mqty{(4, 2) \\ \vb{(4, 3)} \\ (4, 1) \\ (2, 3) \\ (2, 1) \\ (3, 1)} \qquad \Rightarrow \tilde{W}_{13}(k_{\mathrm{crossing}}) = -1/4 \Rightarrow \sigma_3.
    \end{equation*}
    \item Finally, it is the green curve. Note that a black dashed line intersects the green curve at $k = 0$ and $k = 2\pi$. Since the system is periodic, only one of these is counted as a crossing. Here, we choose the crossing at $k = 2\pi$. We have
    \begin{equation*}
    \mqty{(4, 2) \\ (4, 3) \\ (4, 1) \\ \vb{(2, 3)} \\ (2, 1) \\ (3, 1)} \to \mqty{(4, 3) \\ (4, 2) \\ (4, 1) \\ \vb{(3, 2)} \\ (3, 1) \\ (2, 1)} \qquad \Rightarrow \tilde{W}_{23}(k) = -1/4 \Rightarrow \sigma_2.
    \end{equation*}
\end{enumerate}
Collecting the braid generators, we form the braid word $\sigma_1\sigma_3\sigma_2\sigma_1\sigma_3\sigma_2$. 

\begin{figure}[htbp!]
\centering
\includegraphics[width=\linewidth]{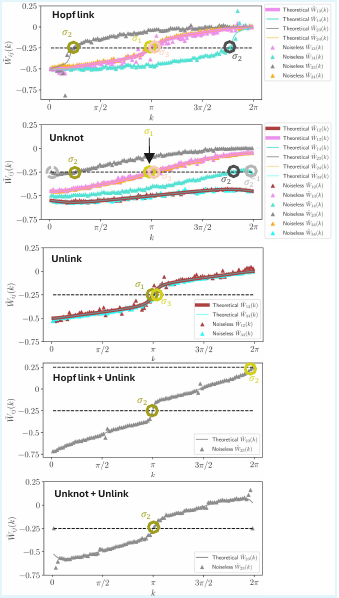}
\caption{\textbf{The winding number evolution for the Hopf link + unlink, unknot + unlink, Hopf link, unknot, and unlink of the four-band twister model \cref{eq:quantum_hardware_simulation/four-band_twister_model/H4_twister_model}.} The black dashed lines correspond to \cref{eq:extraction_of_braid_words/phase_shifted_W_crossing_condition}, which denote the crossings (solid colored circles) on the braid diagrams. It is important to note that the labels $i$ and $j$ in $\tilde{W}_{ij}(k)$ correspond to the initial labels at $k = 0$ before any permutations. The parameters used are $(m_1, m_2) = (-1, -1.5)$, $(-0.7, -1.5)$, $(-2.5, -2.4)$, $(-1.85, 0.25)$, $(-1.4, -0.25)$ for the Hopf link, unknot, unlink, Hopf link + unlink, and unknot + unlink respectively.}
\label{fig:H4_twister_five_phases_W}
\end{figure}

\clearpage
\pagebreak


\let\oldaddcontentsline\addcontentsline
\renewcommand{\addcontentsline}[3]{}
\bibliography{v1/Referencesv2}
\let\addcontentsline\oldaddcontentsline

\clearpage 
\pagebreak


\appendix
\onecolumngrid

\begin{center}
{\large\textbf{Supplementary Material for ``Digital Simulation of Non-Hermitian Knotted Bands on Quantum Hardware''} \par}
\end{center}

In this Supplementary Material, we provide additional technical details and derivations that support and extend the results presented in the main text. The contents of each section are as follows:

\begin{itemize}
    \item \cref{app:generalized_twister_models}: Definition of the generalized twister models \cref{eq:results/knots_and_braids/braiding/twister_model} and a rigorous characterization of their band structures in terms of torus knots and links. 
    \item \cref{app:characterizing_non-Hermitian_braids}: Discussion of the winding number matrix $\mathcal{W}$ \cref{eq:results/quantum_circuit_hardware_implementation/characterizing_braids_and_knotted_structures/winding_number_matrix} as a characteristic quantity for multi-band non-Hermitian twister models. 
    \item \cref{app:two-band}: Specialization to the two-band twister model $\hat{H}^{(2)}(k)$ \cref{eq:quantum_hardware_simulation/two-band_twister_model/H2_twister_model}. 
    \item \cref{app:four-band}: Specialization to the four-band twister model $\hat{H}^{(4)}(k)$ \cref{eq:quantum_hardware_simulation/four-band_twister_model/H4_twister_model}. 
    \item \cref{app:knot_polynomials}: Definition and computation of the Alexander and Jones polynomials, with Solomon's knot as an example. 
    \item \cref{app:GBBP}: Definition of the global biorthogonal Berry phase and its inadequacy in uniquely characterizing the knots and links in the family of twister models unlike the winding number matrices. 
    \item \cref{app:ibm_hardware_details}: IBM hardware details. 
\end{itemize}

\tableofcontents


\section{Generalized Twister Models}
\label{app:generalized_twister_models}

In this section, we develop a general mathematical framework for a family of non-Hermitian tight-binding models which we refer to as generalized twister models. This formalism allows us to analytically characterize the structure of the complex energy bands and reveal how their trajectories form torus knots and links in three-dimensional space. We first define the models and their spectral properties before establishing the correspondence between the band structure and torus knot embeddings. Concrete realizations of the two- and four-band twister models are discussed in more detail in \cref{app:two-band,app:four-band}. 

\subsection{Definitions}

We begin by defining the twister matrices which encode a cyclic hopping structure with a parameter-dependent phase twist in one of the components. For convenience, we reiterate the definitions from the main text here. 

\begin{definition}
We consider the family of $N \times N$ twister matrices

\begin{equation}
\hat{T}_V^{(N)}(k) \equiv 
\mqty[
0 & 0  & \cdots & 0 & e^{i V k} \\
1 & 0  & \cdots & 0 & 0 \\
0 & 1  & \cdots & 0 & 0 \\
\vdots & \vdots & \ddots & 0 & 0 \\
0 & 0  & \cdots & 1 & 0
], 
\end{equation}
where $k\in[0, 2\pi]$. The generalized twister model is defined as
\begin{equation}
\hat{H}^{(N)}(k) 
= m_0 \hat{\Sigma}^{(N)} + \sum_{v=1}^{V} m_v \hat{T}_V^{(N)}(k),
\end{equation}
where $m_0, m_1, \ldots, m_V\in\mathbb{C}$ and the constant shift diagonal matrix $\hat{\Sigma}^{(N)}$ is defined as $\hat{\Sigma}^{(N)}_{pq} = \left(1 - \frac{2(p-1)}{N - 1}\right) \delta_{pq}$ with $1 \leq p, q \leq N$.
\end{definition}

\begin{lemma}
The energies of $\hat{T}_V^{(N)}(k)$ are given by
\begin{equation}
E_j(k) = e^{\frac{i}{N}(V k + 2\pi j)}, \quad j = 0, 1, \dots, N-1,
\end{equation}
which are the $N$-th roots of $e^{i V k}$.
\label{lemma:energies_of_twister_models}
\end{lemma}

\begin{proof}
Let $\hat{T} = \hat{T}_V^{(N)}(k)$. Observe that $\hat{T}$ is a companion matrix of the monic polynomial

\begin{equation}
P(E) = E^N - e^{i V k}.
\end{equation}
This follows because $\hat{T}$ acts as a cyclic shift operator with a phase at the upper-right corner. Specifically, for a vector $w = [w_1, w_2, \dots, w_N]^T$, the matrix $\hat{T}$ acts as

\begin{equation}
\hat{T}w = 
\mqty[
e^{i V k} w_N \\
w_1 \\
w_2 \\
\vdots \\
w_{N-1}
],
\end{equation}
which shifts each component down by one and sends $w_N$ to the top with a phase factor $e^{i V k}$. Applying $\hat{T}$ successively $N$ times gives
\begin{equation}
\hat{T}^N w = e^{i V k} w,
\end{equation}
so $w$ is an eigenstate of $\hat{T}^N$ with energy $e^{i V k}$. Therefore, the energies $E$ of $\hat{T}$ must satisfy
\begin{equation}
E^N = e^{i V k}.
\end{equation}
Thus, the eigenvalues of $\hat{T}$ are the $N$-th roots of $e^{i V k}$, explicitly given by
\begin{equation}
E_j(k) = e^{\frac{i}{N}(V k + 2\pi j)}, \quad j = 0, 1, \dots, N-1.
\end{equation}
The argument $\theta_j(k) = \frac{V k + 2\pi j}{N}$ varies continuously with $k$ and completes $V$ full turns around the unit circle as $k$ ranges from $0$ to $2\pi$.
\end{proof}

\subsection{Torus Knots/Links in the Twister Models}

We now show that the energy trajectories of the twister model, when regarded as curves in three-dimensional space, form torus knots or links. To make this connection precise, we first recall the standard parametrization of torus knots and their extension to torus links, and that the bands of $\hat{T}_V^{(N)}(k)$ given by \cref{lemma:energies_of_twister_models} can be mapped to these parametrized curves.

\begin{lemma}[Standard torus–knot parametrisation]
\label{lem:Tpq-knot-param}
For coprime integers $p,q\in\mathbb{Z}_{>0}$ the map
\begin{equation}
\label{eq:Tpq-param}
\Gamma(\phi)=
\left( 
(2+\cos (q\phi))\cos (p\phi),\;
(2+\cos (q\phi))\sin (p\phi),\;
-\sin (q\phi) 
\right), 
\quad \phi \in [0, 2\pi],
\end{equation}
is an embedding of the $(p,q)$-torus knot.  It lies on the torus
$\left(\sqrt{x^{2}+y^{2}}-2\right)^{2}+z^{2}=1$
and winds $p$ times longitudinally and $q$ times poloidally.
\end{lemma}

\begin{lemma}[Extension to torus links]
\label{thm:Tpq-link}
Let $d=\gcd(p,q)$ and write $p=d\,p',\;q=d\,q'$ with $\gcd(p',q')=1$.
Then the $(p,q)$-torus link has $d$ disjoint components, each of
type $(p',q')$, and can be parameterized by
\begin{equation}
\label{eq:Tpq-link-param}
\Gamma_m(\phi)=
\left( 
[2+\cos (q\phi+ \tfrac{2\pi m}{d})]\cos(p\phi),\;
[2+\cos (q\phi+\tfrac{2\pi m}{d})]\sin(p\phi),\;
-\sin (q\phi+\tfrac{2\pi m}{d}) 
\right),
\quad m=0,\dots,d-1.
\end{equation}
\end{lemma}

\begin{theorem}
For fixed integers $V, N$, the collection of energy trajectories $\{E_j(k)\}_{j=0}^{N-1}$ of the Hamiltonian $\hat{T}_V^{(N)}(k)$, are embedded in $\mathbb{R}^3$ as
\begin{equation}
\gamma_j(k) \equiv 
\left( \left[2 + \cos \left( \tfrac{V k + 2\pi j}{N} \right)\right] \cos k , 
\left[2 + \cos \left( \tfrac{V k + 2\pi j}{N} \right)\right] \sin k , 
-\sin\left( \tfrac{V k + 2\pi j}{N} \right) \right),
\end{equation}
trace out a torus knot or link of type $(N, V)$ as $k$ runs from $0$ to $2\pi$.
\end{theorem}

\begin{proof}
Let $d = \gcd(V, N)$, and write
\begin{equation}
V = d V', 
\qquad
N = d N',
\end{equation}
with $\gcd(V', N') = 1$.  Denote
\begin{equation}
\Gamma_m(k)
= \left( 
\left[2+\cos(Vk+\tfrac{2\pi m}{d}) \right]\cos(Nk),
\left[2+\cos(Vk+\tfrac{2\pi m}{d}) \right]\sin(Nk),
-\sin \left(Vk+\tfrac{2\pi m}{d} \right) 
\right),
\ m=0,\ldots,d-1,
\end{equation}
and
\begin{equation}
\gamma_j(k) = 
\left(\left[2+\cos\!\left(\tfrac{Vk+2\pi j}{N}\right) \right]\cos k,
\left[2+\cos\!\left(\tfrac{Vk+2\pi j}{N}\right) \right]\sin k,
-\sin\!\left(\tfrac{Vk+2\pi j}{N}\right)\right),
\ j=0,\ldots,N-1. 
\end{equation}
Since $N=dN'$, observe that
\begin{equation}
\gamma_{j+N'}(k)
= \left(
\left[2+\cos\left(\tfrac{V k + 2\pi j}{N} + \tfrac{2\pi}{d}\right)\right]\cos k, 
\left[2+\cos\left(\tfrac{V k + 2\pi j}{N} + \tfrac{2\pi}{d}\right) \right]\sin k, 
-\sin\left(\tfrac{V k + 2\pi j}{N} + \tfrac{2\pi}{d}\right)\right).
\end{equation}
This implies that curves differing in index by $N^{\prime}$ are related by a constant shift $2 \pi / d$ in the internal parameterization. So the set $\{\gamma_j\}_{j=0}^{N-1}$ collapses to exactly $d$ distinct loops, namely $\gamma_{mN'}$, $m=0,\ldots,d-1$.

Similarly, we can re-parameterize $\Gamma_m$. For each $m$, set $j = m N'$. Then
\begin{equation}
\gamma_{mN'}(Nk)
=\left(
\left [2+\cos\!\left(\tfrac{V(Nk)+2\pi mN'}{N}\right) \right]\cos(Nk),\ \dots
\right) = \Gamma_m(k),
\end{equation}
since

\begin{equation}
\frac{V(Nk)+2\pi mN'}{N} = V k + \frac{2\pi m}{d}.
\end{equation}
Thus $\Gamma_m(k)=\gamma_{mN'}(Nk)$, i.e.\ each $\Gamma_m$ is exactly the loop $\gamma_{mN'}$ composed with the degree-$N$ covering map $k\mapsto Nk\pmod{2\pi}$. We define
\begin{equation}
H_m(s,k)=\gamma_{mN'} \left((1-s) \, N k + s \, k\right),
\quad s\in[0,1], \, k\in[0,2\pi].
\end{equation}
Then
\begin{equation}
H_m(0,k)=\gamma_{mN'}(Nk)=\Gamma_m(k),
\quad
H_m(1,k)=\gamma_{mN'}(k),
\end{equation}
so $H_m$ is a homotopy between $\Gamma_m$ and $\gamma_{mN'}$. Combining these steps shows that
\begin{equation}
\left\{\Gamma_m : m=0,\ldots,d-1\right\}
=
\left\{\gamma_j : j=0,\ldots,N-1\right\}
\end{equation}
are sets of loops on the torus, so the set $\{\gamma_j\}_{j=0}^{N-1}$ is isotopic to a torus knot/link.  This completes the proof.
\end{proof}

The above theorem establishes that each dominant shift term $\hat{T}_V^{(N)}(k)$ produces a torus knot or link structure of type $(m, N)$ in the bands. This observation directly leads to the following consequence for the full generalized twister model. 

\begin{corollary}
The generalized twister model $\hat{H}^{(N)}(k)$ can host phases whose band structures realize torus knots or links of the form $(1, N)$, $(2, N)$, $\ldots$, up to $(V, N)$.
\end{corollary}
\begin{proof}
For $i, j \in \mathbb{N}_{+}$ and $m, n \in [1, V]$, consider the limit
\begin{equation}
\lim_{\forall m \neq n, \abs{c_m/c_n} 
\rightarrow \infty} 
\hat{H}^{(N)}(k) = \hat{T}_V^{(N)}(k)
\end{equation}
which yields a band structure corresponding to the $(m, N)$ torus knot or link.
\end{proof}


\section{Characterizing non-Hermitian Braids By Winding Numbers}
\label{app:characterizing_non-Hermitian_braids}

In this section, we motivate three inequivalent definitions of the winding number in various contexts as they play different roles. These are the $k$-dependent winding number $W_{ij}(k)$, and the $k$-independent winding numbers $\bar{W}_{ij}$ and $\mathcal{W}_{ij}$ that are computed over the first period and averaged over $n$ periods respectively. As written in the main text, $W_{ij}(k)$ quantifies the partial winding information. However, characterizing a multi-strand braid requires an averaged winding number matrix $\mathcal{W}$ which is computed in terms of the first-period winding number matrix $\bar{W}$ since the band indices need to be carefully tracked. Although $\mathcal{W}_{ij}$ can distinguish between different knots and links in the generalized twister models studied in this work, we emphasize that this is \textit{not} true for generic models as it is not a knot invariant like the Jones polynomial. 

In the following, we remind the reader that $\Lambda_i(k) \propto E_i(k)$ hence the spectral braiding information is captured by $\Lambda_i(k) - \Lambda_j(k)$. We define the $k$-dependent winding number between two bands as 

\begin{definition}
Given two bands $i$ and $j$ and the trajectories $\Lambda_i(k) \neq \Lambda_j(k) \, \forall k$, the $k$-dependent winding number which quantifies the partial winding between them is given by 

\begin{equation}
W_{ij}(k) := \frac{1}{2 \pi i} \int_0^k dk' \, \partial_{k'} \ln [\Lambda_i(k') - \Lambda_j(k')]. 
\label{eq:characterizing_non-Hermitian_braids/winding_number_element_definition}
\end{equation}
\end{definition}
Generalizing to multiple bands, we have the winding number matrix whose definition is 
\begin{definition}
The winding number matrix $\mathcal{W}$ of a $N$-band Hamiltonian is a $N \times N$ symmetric matrix 

\begin{equation}
\mathcal{W} := \frac{1}{n}\sum_{a=0}^{n-1}(P^{-1})^a \bar{W} P^a,
\end{equation}
where $n$ is determined by the minimal power of the permutation matrix $P$ such that $P^n = \mathbb{I}$ and each element is the $k$-independent pairwise winding number of two bands up to $k = 2\pi$

\begin{equation}
\bar{W}_{ij} := \frac{1}{2\pi i}  \int_0^{2\pi}dk' \, \partial_{k'} \ln [\Lambda_i(k') - \Lambda_j(k')].
\label{eq:characterizing_non-Hermitian_braids/winding_number_element}
\end{equation}
\end{definition}
From the definition, it is obvious that $\mathcal{W}$ has vanishing diagonal terms. Note that for multi-band systems, there is a subtlety in defining the winding number between an arbitrary pair of bands. Consider a system with period $2\pi$. For $\mathcal{W}_{ij}$ between two bands $E_i(k)$ and $E_j(k)$ to be topological and thus well-defined, the sets of energies must satisfy $\{E_i(0), E_j(0)\} = \{E_i(2n\pi), E_j(2n\pi)\}$. This condition is automatically satisfied in two-band systems but may fail in multi-band systems since the permutation matrix of eigenstates $P$ does not have to be diagonal or off-diagonal. For example, in the unknot of $\hat{H}^{(4)}(k)$, we have $\{E_i(0), E_j(0)\} = \{E_i(8\pi), E_j(8\pi)\}$ only after $k$ evolves for four periods. To address this, we identify the smallest $n \in \mathbb{Z}_+$ such that $\{E_i(0), E_j(0)\} = \{E_i(2n\pi), E_j(2n\pi)\}$. $n$ always exists because the permutation group is finite-dimensional. Determining $n$ typically involves tracking sets of bands $\{i_1, i_2, \dots, i_n\}$ and $\{j_1, j_2, \dots, j_n\}$ that form closed cycles over $k \in [0, 2n\pi]$. This ensures that the winding numbers remain topological, although they may now take fractional values like $\mathbb{Z}/n$, which we observe in the geometric significance of $\mathcal{W}_{ij}$ below. 

\begin{figure}[htbp!]
\centering
\includegraphics[width= 0.8\linewidth]{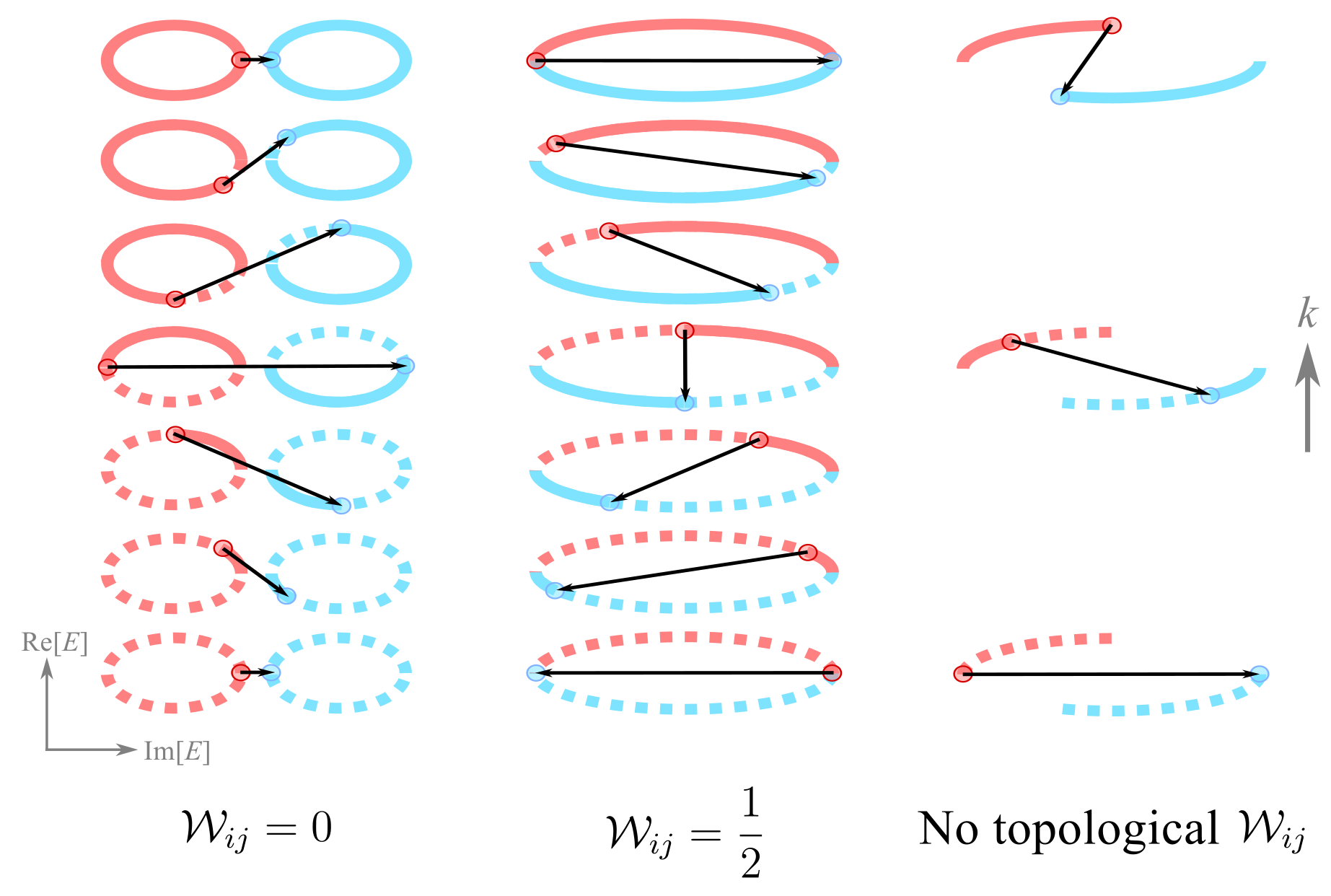}
\caption{\textbf{Geometric significance of the winding number between two bands.} Given two bands $i$ (red) and $j$ (blue), the winding number $\mathcal{W}_{ij}$ is topologically well defined only when the trajectories form closed paths in the complex plane, namely when the initial and final points coincide as the same set of points. This condition is satisfied in the leftmost and middle columns, whereas the rightmost column provides a counterexample in which $\mathcal{W}_{ij}$ is not well defined, and becomes sensitive to perturbations. For multiple bands, it may require multiple periods for all the winding numbers to be topological.}
\label{fig:w_geometric_meaning}
\end{figure}

If the two bands do not wind around each other, $\mathcal{W}_{ij} = 0$; if they only interchange positions once after evolving for $n$ periods, $\mathcal{W}_{ij} = 1/2$; if they wind around each other once or equivalently, interchange twice, $\mathcal{W}_{ij}(k) = 1$. So $\mathcal{W}_{ij}$ can only be $\mathbb{Z}/n$, and is insensitive to the exact values of $\Lambda_i(k)$ and $\Lambda_j(k)$. In multi-band models, the winding may be rational numbers as a band may not evolve to itself within one period. Geometrically, $\mathcal{W}_{ij}$ is the net phase accumulated by the vector $\Lambda_i(k) - \Lambda_j(k)$ over $2\pi$ if one writes the real and imaginary parts of $\Lambda_i(k)$ and $\Lambda_j(k)$ explicitly. 

\cref{fig:w_geometric_meaning} illustrates how the $k$-independent winding number $\mathcal{W}_{ij}$ can be determined without performing an explicit integration by leveraging the geometric interpretation introduced above. Consider the leftmost column as an example: the phase traced by the vector $E_i(k) - E_j(k)$ initially increases, subsequently decreases, and ultimately returns to zero. Here, we treat complex numbers in $\mathbb{C}$ and their corresponding two-dimensional real vectors in $\mathbb{R}^2$ interchangeably for illustrative purposes. This behavior indicates that $\mathcal{W}_{ij} = 0$. A similar analysis applies to the middle panel, which corresponds to the unknot. However, such an approach is valid only when the sets of eigenvalues at $k = 0$ and $k = 2\pi$ (boundaries of the Brillouin zone) coincide. If this condition is not satisfied, as demonstrated in the rightmost column of \cref{fig:w_geometric_meaning}, the resulting net phase becomes sensitive to the precise values of the energies, rendering it non-topological by definition. 

It is worth noting that other equivalent definitions of the $k$-independent winding number $\bar{W}_{ij}$ have been proposed in the context of multi-band systems to render it a topological invariant, as demonstrated in earlier literature \cite{long2024unsupervised}. However, the formulation introduced in this work is particularly well-suited for the generalized twister models and implementation on quantum hardware, as it avoids the need to directly access energies. 

\begin{table}[h]
\centering
\renewcommand{\arraystretch}{1.5}
\begin{tabular}{c|c||c|c||c|c}
\hline
\textbf{Knot/Link} & $\mathcal{W}$ & \textbf{Knot/Link} & $\mathcal{W}$ & \textbf{Knot/Link} & $\mathcal{W}$ \\ 
\hline
\hline
Unknot & $\mqty[
0 & \tfrac{1}{4} & \tfrac{1}{4} & \tfrac{1}{4}\\
\tfrac{1}{4} & 0 & \tfrac{1}{4} & \tfrac{1}{4}\\
\tfrac{1}{4} & \tfrac{1}{4} & 0 & \tfrac{1}{4}\\
\tfrac{1}{4} & \tfrac{1}{4} & \tfrac{1}{4} & 0
]$ 
& Hopf chain & $\mqty[
0 & \tfrac{1}{2} & \tfrac{1}{2} & 0 \\
\tfrac{1}{2} & 0 & \tfrac{1}{2} & \tfrac{1}{2} \\
\tfrac{1}{2} & \tfrac{1}{2} & 0 & \tfrac{1}{2} \\
0 & \tfrac{1}{2} & \tfrac{1}{2} & 0
]$ 
& Solomon's knot &$\mqty[
0 & \tfrac{1}{2} & \tfrac{1}{2} & \tfrac{1}{2}\\
\tfrac{1}{2} & 0 & \tfrac{1}{2} & \tfrac{1}{2}\\
\tfrac{1}{2} & \tfrac{1}{2} & 0 & \tfrac{1}{2}\\
\tfrac{1}{2} & \tfrac{1}{2} & \tfrac{1}{2} & 0
]$  \\ \hline
Hopf link + Unlink & $\mqty[
0 & 0 & 0 & 0 \\
0 & 0 & 1 & 0 \\
0 & 1 & 0 & 0 \\
0 & 0 & 0 & 0
]$ 
& Unknot + Unlink & $\mqty[
0 & 0 & 0 & 0 \\
0 & 0 & \tfrac{1}{2} & 0 \\
0 & \tfrac{1}{2} & 0 & 0 \\
0 & 0 & 0 & 0
]$ 
& Double Unlink & $\mqty[
0 & 0 & 0 & 0\\
0 & 0 & 0 & 0\\
0 & 0 & 0 & 0\\
0 & 0 & 0 & 0
]$\\ \hline
Hopf link &  $\mqty[
0 & 0 & \tfrac{1}{2} & \tfrac{1}{2} \\
0 & 0 & \tfrac{1}{2} & \tfrac{1}{2} \\
\tfrac{1}{2} & \tfrac{1}{2} & 0 & 0 \\
\tfrac{1}{2} & \tfrac{1}{2} & 0 & 0
]$  
& Unlink & $\mqty[
0 & \tfrac{1}{2} & 0 & 0\\
\tfrac{1}{2} & 0 & 0 & 0 \\
0 & 0 & 0 & \tfrac{1}{2}\\
0 & 0 & \tfrac{1}{2} & 0
]$  & &  \\
\hline
\end{tabular}
\caption{\textbf{$\mathcal{W}$ of various knots and links in the spectrum of $\hat{H}^{(4)}(k)$}.}
\label{tab:four-band_winding_number_matrix}
\end{table}

By combining \cref{eq:characterizing_non-Hermitian_braids/winding_number_element} with the geometric perspective discussed above, and utilizing the band structures extracted from numerical simulations (\cref{fig:H4_phase_diagram}), one can derive the winding number matrix $\mathcal{W}$ for all knots and links of $\hat{H}^{(4)}(k)$ as summarized in \cref{tab:four-band_winding_number_matrix}. Notably, in certain cases, the elements of $\mathcal{W}$ can be determined directly by identifying the number of bands that form closed structures such as unknots, unlinks, or Hopf links without requiring detailed phase tracking. For instance, for the Hopf chain, three closed loops are clearly visible: one formed each by the purple and blue bands, and another by the red and green bands together. The winding number between the red and green bands corresponds to an unlink yielding $\mathcal{W}_{ij} = 0$. The purple and blue bands form an unknot yielding $\mathcal{W}_{ij} = 1/2$. The remaining band pairs form Hopf links where two full periods are required to complete the closed loop in the middle, leading to $\mathcal{W}_{ij} = 1/2$ for these cases as well.

\section{Two-band Twister Model}
\label{app:two-band}

The two-band twister model is 

\begin{equation}
\hat{H}^{(2)}(k) = \mqty[
im_0 & e^{2ik} + e^{ik}m_1 \\
1 + m_1 & -im_0
],
\label{eq:two-band/H2_twister_model}
\end{equation}
and its eigenvalues are

\begin{equation}
E_\pm(k) = \pm\sqrt{e^{2ik}(m_1+1) + m_1e^{ik}(m_1+1)-m_0^2}, 
\label{eq:two-band/energies}
\end{equation}
where $\abs{e^{ik}} = 1$.

\subsection{Reconstruction of Eigenstates}
\label{app:two-band/reconstruction_of_eigenstates}

For two-band systems, the eigenstates can be parameterized on the Bloch sphere, which allows us to reconstruct them from a set of Pauli operator measurements. To compute quantities of interest such as the winding number, it is necessary to obtain the right eigenstates of the model since they are the key ingredients as given in \cref{eq:quantum_hardware_simulation/two-band/p}. Here, we show how the eigenstates of two-band models can be reconstructed from the measured observables on quantum hardware by computing the corresponding parameters on the Bloch sphere. Although the models we study are one-dimensional, we remark that the expressions to reconstruct the eigenstate \cref{eq:two-band/reconstruction_of_eigenstates/reconstructed_eigenstate} are also valid for higher-dimensional models. 

An arbitrary eigenstate corresponding to a two-band model can be parameterized on the Bloch sphere by two angles $\theta(k)$ and $\phi(k)$ as \cite{nielsen2010quantum} 
\begin{equation}
\ket{\psi(k)} = \mqty[
\psi_{+}(k)\\ \psi_{-}(k)
]
= \mqty[
c_{+}(k) e^{i \phi_{+}(k)} \\
c_{-}(k)]
= \mqty[
\cos\frac{\theta(k)}{2} e^{i \phi(k)} \\ 
\sin \frac{\theta(k)}{2}
],
\label{eq:two-band/reconstruction_of_eigenstates/reconstructed_eigenstate}
\end{equation}
where we implicitly adopt a gauge in which the second component of the spinor is taken to be real ($\phi_{-}(k) = 0$). Moreover, all the observables are rational functions of the spinor components $\psi_{\pm}$ and their conjugates $\psi^*_{\pm}$, and the components of $\ket{\psi(k)}$ are related to the Pauli operators as \cite{PhysRevB.97.195422} 

\begin{equation}
\begin{aligned}
& \expval{\hat{\sigma}^z(k)} = \frac{\abs{\psi_{+}(k)}^2-\abs{\psi_{-}(k)}^2}{\abs{\psi_{+}(k)}^2+\abs{\psi_{-}(k)}^2}, 
\quad 
\expval{\hat{\sigma}^x(k)} = \frac{\psi_{+} (k)\psi^*_{-}(k)+\psi^*_{+}(k) \psi_{-}(k)}{2\abs{\psi_{+} (k)\psi^*_{-}(k)}^2},
\quad 
\expval{\hat{\sigma}^y(k)} = \frac{i[ \psi_{+} (k)\psi^*_{-}(k)-\psi^*_{+}(k) \psi_{-} (k)]}{2\abs{\psi_{+}(k) \psi^*_{-}(k)}^2}. 
\label{eq:two-band/reconstruction_of_eigenstates/observables_expectation}
\end{aligned}
\end{equation}
Since the Bloch sphere lives in the space spanned by $(\expval{\hat{\sigma}_x(k)}, \expval{\hat{\sigma}_y(k)}, \expval{\hat{\sigma}_z(k)})$, the polar angle $\theta(k)$ and the azimuthal angle $\phi(k)$ are obtainable on quantum hardware by measuring the Pauli operators. 

\begin{theorem}
The angles to reconstruct the two-band eigenstate are 
\begin{equation}
\theta(k) = \cos^{-1}\expval{ \hat{\sigma}^z(k)}, \qquad \phi(k) = \tan^{-1} \frac{\expval{ \hat{\sigma}^y(k)}}{\expval{\hat{\sigma}^x(k)}}.
\end{equation}
With $\theta(k)$ and $\phi(k)$, $\ket{\psi(k)}$ can be reconstructed. 
\end{theorem}

\subsection{Measurement Protocol}
\label{app:two-band/measurement_protocol}

In this subsection, we show that for generic two-band non-Hermitian models, the winding between the energy bands can be reconstructed from the measured Bloch vectors of the right eigenstates on quantum hardware even if the exact representation of the Hamiltonian is unknown.

We consider a family of non-Hermitian $2 \times 2$ Hamiltonians parameterized by a parameter $k \in [0, 2\pi]$ 

\begin{equation}
\hat{H}(k) = d_x(k)\hat{\sigma}^x + d_y(k)\hat{\sigma}^y + d_z(k)\hat{\sigma}^z,
\label{eq:two-band/measurement_protocol/generic_two-band}
\end{equation}
where $d_\mu(k) \in \mathbb{C}$ for $\mu\in\{x, y, z\}$, and $\hat{\sigma}^{x,y,z}$ are the Pauli operators. Since the identity operator does not change the eigenstate or the winding number between two bands (as shown below), we can ignore terms like $d_0\hat{\mathbb{I}}_2$ without any loss of generality. The energies of $\hat{H}(k)$ are $E_{\pm}(k) = \pm \sqrt{d_x(k)^2 + d_y(k)^2 + d_z(k)^2}$. We assume there is no exceptional point to ensure that $\hat{H}(k)$ is diagonalizable for all $k$, so
\begin{equation}
E_\pm(k) \neq 0, \quad \forall k \in [0,2\pi].
\end{equation}
Let $|\psi_\pm(k)\rangle$ denote the normalized right eigenstates of $\hat{H}(k)$ corresponding to $E_{\pm}(k)$: $\hat{H}(k) |\psi_\pm(k)\rangle = E_{\pm}(k) |\psi_\pm(k)\rangle$. We define the complex Bloch vector associated with the eigenstate $\ket{\psi_\pm(k)}$ as

\begin{equation}
\vb{L}_\pm(k) = [\mel{\psi_\pm(k)}{\hat{\sigma}^x}{\psi_\pm(k)},
\mel{\psi_\pm(k)}{\hat{\sigma}^y}{\psi_\pm(k)},
\mel{\psi_\pm(k)}{\hat{\sigma}^z}{\psi_\pm(k)}]^T.
\end{equation}
This can be explicitly written as 
\begin{equation}
\vb{L}_\pm(k) = \frac{1}{\mathcal{N}^2(k)}\mqty[
\pm 2 \text{Re} \left( \frac{\pm d_z(k) + E(k)}{  d_{+}(k)} \right) \\
\mp 2 \text{Im} \left( \frac{\pm d_z(k) + E(k)}{  d_{+}(k)} \right)\\
\mathcal{N}^2(k) - 2
].
\end{equation}
where we have defined $d_{+}(k) = d^*_{-}(k) = d_x (k) + i d_y (k) \in \mathbb{C}$ and the normalization factor $\mathcal{N}(k) = \sqrt{1+ \left|\frac{E(k) - d_z(k)}{d_+(k)} \right|^2} \in \mathbb{R}$. Note that $\vb{L}_\pm(k)$ does not necessarily have unit norm since it is constructed from right eigenstates in a non-Hermitian system. What matters for our analysis is the evolution of its complex phase rather than its magnitude.

For the energy bands of $\hat{H}(k)$ in \cref{eq:two-band/measurement_protocol/generic_two-band}, the winding number can be written as
\begin{equation}
W_{+-}(k) = \frac{1}{2 \pi i} \int_0^k dk' \partial_{k'} \ln E(k'),
\end{equation}
since $E_+(k')$ and $E_-(k')$ are symmetric about the origin so we set $E(k') = E_+(k')$. It is important to note, however, that this winding number is defined using the logarithmic function, which is inherently multi-valued, thus differing by integer multiples of $2\pi i$ depending on the branch cut choice. Thus, to obtain a meaningful and continuous phase, the logarithm must be properly unwrapped. This contrasts with the conventional definition of a winding number, which is strictly integer-valued. The distinction arises because in the definition of $W_{+-}(k)$, we treat the unordered set $\{E_+(k), E_-(k)\}$ as periodic, effectively counting the number of Riemann sheets of the logarithmic function traversed by $E_+(k)$ as it evolves. In contrast, the conventional definition of the winding number assumes that $E_+(k)$ always satisfies $E_+(0) = E_+(2\pi)$ however, that might not be true here as $E_\pm(0) = E_\mp(2\pi)$ instead. 

The question thus becomes: how much of $W_{+-}(k)$ can be recovered from the Bloch vectors $\vb{L}_\pm(k)$? In the following, we assume that $d_z(k)$ is constant so that it does not contribute to $W_{+-}(k)$, and demonstrate that under this condition, the winding number $W_{+-}(k)$ can be extracted. Physically, this assumption is well motivated in many models, as $d_z(k)$ typically represents a constant term and is often taken to be constant.

To extract the winding of $\vb{L}_\pm(k)$, we consider a homeomorphic mapping preserving topological properties, given by $\mathcal{P}: S^2 \times S^1 \rightarrow \mathbb{C} \times S^1$. A commonly used example is the stereographic projection, which requires carefully removing one point (typically a pole mapped to infinity) from the sphere. Specifically, this implies $L^z_\pm(k) \neq 1$. In practice, this condition is satisfied unless the eigenstate coincides with the projection pole for all $k$. If necessary, the pole can be shifted by a global unitary rotation of the measurement basis. In the following, we adopt the explicit form of the mapping as

\begin{equation}
\mathcal{P}: \vb{L}_\pm(k) \rightarrow \mathcal{P}\vb{L}_\pm(k) = \frac{\expval{\hat{\sigma}^x_\pm(k)} + i\expval{\hat{\sigma}^y_\pm(k)}}{1 - \expval{\hat{\sigma}^z_\pm(k)}} = \frac{\pm E^*(k) + d^*_z(k)}{d_{-}(k)}.
\end{equation}
Then we define
\begin{equation}
\begin{aligned}
\mathcal{L}_{+}(k) = \frac{\mathcal{P}\vb{L}_+(k)  + \mathcal{P}\vb{L}_-(k)}{2} = \frac{d^*_z(k)}{d_{-}(k)}, \\
\mathcal{L}_{-}(k) = \frac{\mathcal{P}\vb{L}_+(k)  - \mathcal{P}\vb{L}_-(k)}{2} = \frac{E^*(k)}{d_{-}(k)}, 
\end{aligned}
\end{equation}
such that
\begin{equation}
\Lambda(k) = \mathcal{L}_+(k)\mathcal{L}^*_-(k) = \frac{d^*_z(k)E(k)}{d_-(k)d_+(k)} = \frac{1}{4}p_+(k)p_-(k), 
\end{equation}
where 

\begin{equation}
p_\pm(k) = \frac{\expval{\hat{\sigma}^x_+(k)} \pm i\expval{\hat{\sigma}^y_+(k)}}{1 - \expval{\hat{\sigma}^z_+(k)}} \pm \frac{\expval{\hat{\sigma}^x_-(k)} \pm i\expval{\hat{\sigma}^y_-(k)}}{1 - \expval{\hat{\sigma}^z_-(k)}}.
\end{equation}
Since $d_z(k)$ is assumed to be a constant, and $d_-(k)d_+(k) \in \mathbb{R}$ (a periodic path on the real axis cannot contribute to the winding because the phase is a periodic step function of $k$), the winding number of $\mathcal{L}_{+}(k)\mathcal{L}^*_{-}(k)$ can be written as 
\begin{equation}
W_{+-}(k) = \frac{1}{2 \pi i} \int_0^k dk' \, \partial_{k'} \ln \Lambda(k').
\label{eq:two-band/measurement_protocol/winding_number}
\end{equation}
Specializing to the two-band twister model yields

\begin{equation}
\Lambda(k) = \frac{-im_0}{(1+m_1)^2}E(k). 
\end{equation}
\cref{eq:two-band/measurement_protocol/winding_number} is a specialization of \cref{eq:characterizing_non-Hermitian_braids/winding_number_element} as the energies are inversion-symmetric about the origin. This proves that one can \textit{fully restore} the winding information of the energy bands from the Bloch sphere trajectories of the right eigenstates without requiring knowledge of the explicit form of the Hamiltonian. Therefore, $W_{+-}(k)$ provides a complete topological classification of the band structure for any two-band non-Hermitian model with constant $d_z(k)$. We note that the condition can be relaxed to the zero winding of $d_z(k)$, although this is less practical as it still requires access to the Hamiltonian. Nevertheless, it remains a \textit{sufficient} condition for extracting $W_{+-}(k)$. Alternative mappings may be explored to uncover additional sufficient conditions.

\subsection{Topologically Inequivalent Regions in Parameter Space}
\label{app:two-band/phase_diagram}

The degeneracy in \cref{eq:two-band/energies} occurs at $E_+(k) = E_-(k)$ or equivalently, $E_+(k) = 0$ since $E_-(k) = -E_+(k)$. Let $e^{ik} = x + iy$ where $x, y\in[-1, 1]$ so that we can deal with the real and imaginary parts separately. As such, we have 

\begin{equation}
\begin{aligned}
\Re(E^2_+(k)) &= (m_1x+x^2-y^2)(m_1+1)-m_0^2, \\
\Im(E^2_+(k)) &= y(2x+m_1)(m_1+1).
\label{eq:two-band/phase_diagram/Re_Im_E}
\end{aligned}
\end{equation}
There are several cases to consider for $E^2_\pm(k) = 0$:

\begin{enumerate}
    \item If $m_0 = 0$ and $m_1 = -1$, then $\Re(E_+(k)) = \Im(E_+(k)) = 0$. 
    \item $(x, y) = (1, 0) \Rightarrow \Re(E^2_+(k)) = (m_1+1)^2 - m_0^2 = 0$ and $\Im(E^2_+(k)) = 0$. 
    \item $(x, y) = (-1, 0) \Rightarrow \Re(E^2_+(k)) = m_1^2-1+m_0^2 = 0$ and $\Im(E^2_+(k)) = 0$. 
    \item $(x, y) = (-m_1/2, \sqrt{1-m_1^2/4}) \Rightarrow \Re(E^2_+(k)) = 1+m_0^2+m_1 = 0$ and $\Im(E^2_+(k)) = 0$ with $-2 \leq m_1 \leq 2$ since $\abs{x} = \abs{\cos k} \leq 1$. 
\end{enumerate}
To summarize, the boundaries in \cref{fig:H2_phase_diagram} are defined by a special point at $(m_0, m_1) = (0, -1)$ and  

\begin{equation}
\begin{aligned}
F_1(m_0, m_1) &= (m_1+1)^2 - m_0^2 = 0, \\
F_2(m_0, m_1) &= m_1^2-1+m_0^2 = 0, \\
F_3(m_0, m_1) &= 1 + m_0^2 + m_1 = 0.
\label{eq:two-band/phase_diagram/phase_boundaries}
\end{aligned}
\end{equation}

\section{Four-band Twister Model}
\label{app:four-band}

The four-band twister model is defined as

\begin{equation}
\hat{H}^{(4)}(k) = \mqty[
im_0 & 0 & 0 & e^{2ik} + e^{ik}m_1 \\
1 + m_1 & \frac{im_0}{3} & 0 & 0 \\
0 & 1 + m_1 & -\frac{im_0}{3} & 0 \\
0 & 0 & 1 + m_1 & -im_0
], 
\end{equation}
and its eigenvalues are 

\begin{equation}
E_{\pm,\pm}(k) = \pm\frac{1}{3}\sqrt{-5m_1^2 \pm S(k)}, 
\label{eq:four-band/energies}
\end{equation}
where $S(k) = \sqrt{16m_1^4 + 81e^{ik}(m_2+1)^3(m_2+e^{ik})}$.

\subsection{Reconstruction of Eigenstates}
\label{app:four-band/reconstruction_of_eigenstates}

In a four-band model where the eigenstates are encoded in a two-qubit Hilbert space, our goal is to reconstruct the eigenstate up to a gauge from conditional measurements on individual qubits, in direct analogy with the Bloch sphere reconstruction used for the two-band model.

For brevity, we omit the explicit $k$ dependence in the following. An arbitrary eigenstate for a four-band model encoded in two qubits can be written as 

\begin{equation}
\ket{\psi(k)} = 
\mqty[
\psi_A^+(k) \\
\psi_A^-(k) \\
\psi_B^+(k) \\
\psi_B^-(k)
] = 
\mqty[
c_A^+(k) e^{i\phi_A^+(k)} \\
c_A^-(k) e^{i\phi_A^-(k)} \\
c_B^+(k) e^{i\phi_B^+(k)} \\
c_B^-(k) 
] = 
\mqty[
\cos\frac{\theta_A(k)}{2}\cos\frac{\theta_{AB}(k)}{2}e^{i\phi_A(k)}e^{i\phi_{AB}(k)} \\
\sin\frac{\theta_A(k)}{2}\cos\frac{\theta_{AB}(k)}{2}e^{i\phi_{AB}(k)} \\
\cos\frac{\theta_B(k)}{2}\sin\frac{\theta_{AB}(k)}{2}e^{i\phi_B(k)} \\
\sin\frac{\theta_B(k)}{2}\sin\frac{\theta_{AB}(k)}{2}
], 
\label{eq:four-band/measurement_protocol/reconstructed_eigenstate}
\end{equation}
where $\pm$, $A$ and $B$ denote the spin degrees of freedom belonging to the two qubits which can be treated equally. We remind the reader that we adopt a gauge such that the final component is real-valued. We seek analogous expressions to \cref{eq:two-band/reconstruction_of_eigenstates/observables_expectation}, and the goal is to choose the correct set of observables to reproduce 

\begin{equation}
\begin{aligned}
\expval{\hat{\sigma}^z(k)}_{+-, a} &= \frac{\abs{\psi_a^+(k)}^2 - \abs{\psi_a^-(k)}^2}{\abs{\psi_a^+(k)}^2 + \abs{\psi_a^-(k)}^2}, \\
\expval{\hat{\sigma}^x(k)}_{+-, a} &= \frac{\psi_a^{+}(k) (\psi^*)_a^{-}(k)+(\psi^*)_a^{+}(k) \psi_a^{-}(k)}{2\abs{\psi_a^{+}(k) (\psi^*)_a^{-}(k)}^2}, \\
\expval{\hat{\sigma}^y}_{+-, a} &= \frac{i[\psi_a^{+}(k) (\psi^*)_a^{-}(k) -(\psi^*)_a^{+}(k) \psi_a^{-}(k)]}{2\abs{\psi_a^{+}(k) (\psi^*)_a^{-}(k)}^2}. 
\end{aligned}
\end{equation}
It turns out that the observables are 
\begin{equation}
\begin{aligned}
\hat{I}_{+-, a} = \hat{P}_{a} \otimes \hat{\mathbb{I}}_2, 
\qquad 
\hat{\sigma}_{+-, a}^\alpha = \hat{P}_{a} \otimes \hat{\sigma}^\alpha, 
\label{eq:four-band/reconstruction_of_eigenstates/first_set_of_observables}
\end{aligned}
\end{equation}
where $a \in \{A, B\}$ and $\hat{\sigma}^\alpha$ with $\alpha \in\{x, y, z\}$ are the Pauli operators, and the projection operators are given by
\begin{equation}
\hat{P}_{A} = 
\mqty[
1 & 0\\
0 & 0 
], \qquad
\hat{P}_{B} = \mqty[
0 & 0\\
0 & 1 
].
\end{equation}
Similarly, we also have the expectation values of the other qubit as 
\begin{equation}
\begin{aligned}
\expval{\hat{\sigma}^z(k)}_{-, AB} &= \frac{\abs{\psi_A^{-}(k)}^2-\abs{\psi_B^{-}(k)}^2}{\abs{\psi_A^{+}(k)}^2+\abs{\psi_B^{-}(k)}^2}, \\
\expval{\hat{\sigma}^x(k)}_{-, AB} &= \frac{\psi_A^{-}(k) (\psi^*)_B^{-}(k) +(\psi^*)_A^{-}(k) \psi_B^{-}(k)}{2\abs{\psi_A^{+}(k) (\psi^*)_B^{-}(k)}^2}, \\
\expval{\hat{\sigma}^y(k)}_{-, AB} &= \frac{i[\psi_A^{-}(k) (\psi^*)_B^{-}(k) -(\psi^*)^{-}_{A}(k) \psi_B^{-}(k)]}{2\abs{\psi_A^{-}(k) (\psi^*)_B^{-}(k)}^2}, 
\end{aligned}
\end{equation}
corresponding to the operators 

\begin{equation}
\begin{aligned}
\hat{I}_{\beta, AB} = \hat{\mathbb{I}}_2 \otimes
\hat{P}^{\beta} , 
\qquad 
\hat{\sigma}_{\beta, AB}^\alpha = \hat{\sigma}^\alpha \otimes
\hat{P}^{\beta}, 
\label{eq:four-band/reconstruction_of_eigenstates/second_set_of_observables}
\end{aligned}
\end{equation}
where $\beta \in \{+, -\}$, and the projection operators are defined as 

\begin{equation}
\hat{P}^{+} = \mqty[
1 & 0\\
0 & 0 
], \qquad
\hat{P}^{-} = \mqty[
0 & 0\\
0 & 1 
].
\end{equation}
In essence, the protocol performs conditional tomography: we reconstruct the Bloch vector of one qubit while postselecting on the spin state of the other qubit. Finally, we define an additional operator

\begin{equation}
\hat{\sigma}_{+-, AB}^z = \hat{\sigma}^z \otimes \hat{\mathbb{I}}_2. 
\label{eq:four-band/reconstruction_of_eigenstates/third_set_of_observables}
\end{equation}
The expectation value is the total intensity difference in one of the spin degrees of freedom

\begin{equation}
\expval{\hat{\sigma}^z(k)}_{+-, AB} =  \abs{\psi_A^{+}(k)}^2 + \abs{\psi_A^{-}(k)}^2 - \abs{\psi_B^{+}(k)}^2 - \abs{\psi_B^{-}(k)}^2
\end{equation}
Then we can define the angles analogous to the two-band model as 

\begin{theorem}
The angles to reconstruct the four-band eigenstate include the three polar angles: 
\begin{equation}
\theta_A(k) = \cos^{-1} \frac{\expval{\hat{\sigma}^z(k)}_{+-, A}}{\expval*{\hat{I}(k)}_{+-, A}}, \quad
\theta_B(k) = \cos^{-1} \frac{\expval{\hat{\sigma}^z(k)}_{+-, B}}{\expval*{\hat{I}(k)}_{+-, B}}, \quad
\theta_{AB}(k) = \cos^{-1} \expval{\hat{\sigma}^z(k)}_{+-, AB}. 
\label{eq:four-band/reconstruction_of_eigenstates/theta}
\end{equation}
and the three azimuthal angles:
\begin{equation}
\phi_A(k) = -\tan^{-1}\frac{\expval{\hat{\sigma}^y(k)}_{+-, A}}{\expval{\hat{\sigma}^x(k)}_{+-, A}}, 
\quad 
\phi_B(k) = -\tan^{-1}\frac{\expval{\hat{\sigma}^y(k)}_{+-, B}}{\expval{\hat{\sigma}^x(k)}_{+-, B}}, 
\quad 
\phi_{AB}(k) = -\tan^{-1}\frac{\expval{\hat{\sigma}^y(k)}_{-, AB}}{\expval{\hat{\sigma}^x(k)}_{-, AB}}, 
\label{eq:four-band/reconstruction_of_eigenstates/phi}
\end{equation}
\end{theorem}
\noindent Note that upon postselection, $\expval*{\hat{I}(k)}_{+-, A}$ and $\expval*{\hat{I}(k)}_{+-, B}$ cancel out. With these six angles, \cref{eq:four-band/measurement_protocol/reconstructed_eigenstate} can be reconstructed by measuring \cref{eq:four-band/reconstruction_of_eigenstates/first_set_of_observables,eq:four-band/reconstruction_of_eigenstates/second_set_of_observables,eq:four-band/reconstruction_of_eigenstates/third_set_of_observables}.  

The commonly used one-argument $\tan^{-1}$ function can lead to a phase shift, which matters if we consider wave function components because the shift may not be consistent for each pair. Physically, these quantities are given by tracing out the degrees of freedom from the other spin as shown above. The six parameters are sufficient because the overall normalization and a global phase are fixed (we choose the last component to be real), thus removing two redundant degrees of freedom. This mirrors the two-band model where a normalized spinor in $\mathbb{C}^2$ is described by two independent angles. A straightforward extension suggests that an $N$-qubit system requires $(2^{N+1}-2)$ real parameters to specify one normalized eigenstate in a fixed gauge. 

\subsection{Measurement Protocol}
\label{app:four-band/measurement_protocol}

In this subsection, rather than presenting a comprehensive method, we offer a general recipe for reconstructing topological invariants from the measured observables within a family of models with more details than the main text, including the twister models constructed above. While we illustrate the approach using a simple example of $\hat{H}^{(4)}(k)$, our primary aim is to demonstrate the feasibility of measuring topological invariants for distinguishing non-trivial band structures on quantum hardware, even for multi-band systems. This methodology can be readily generalized to systems with more qubits, more intricate phase diagrams, or other types of topological invariants. However, generalizations to more qubits require different mathematical tools since generically there will be no analytical expressions for the energies anymore.

We start from a generic four-band model
Hamiltonian where we use the Pauli operators instead of the Gell-Mann operators as the basis. This choice is not essential at the level of topology but it makes the mapping between the Hamiltonian and measurable operators on quantum hardware transparent. Thus, we have

\begin{equation}
\begin{aligned}
\hat{H}(k) &= c_0 (k) \ \hat{\mathbb{I}}_2 \otimes (\hat{\sigma}_{x} - i \hat{\sigma}_{y}) 
+ c_{1a} (k) \ \hat{\sigma}_{x} \otimes \hat{\sigma}_{x}
+ c_{1b} (k)\ \hat{\sigma}_{x} \otimes \hat{\sigma}_{y} \\
&\qquad + c_{2a}(k) \ \hat{\sigma}_{y} \otimes \hat{\sigma}_{x}
 + c_{2b}(k) \ \hat{\sigma}_{y} \otimes \hat{\sigma}_{y} 
 + c_{3} (k) \ (\hat{\mathbb{I}}_2 \otimes \hat{\sigma}_{z} 
 + 2 \hat{\sigma}_{z} \otimes \hat{\mathbb{I}}_2),
\label{eq:H_4b}
\end{aligned}
\end{equation}
where we assume $c_0 (k) \in \mathbb{R}$ and all the other coefficients are complex. To simplify the expressions in the following discussions, we define some auxiliary functions as follows: 
\begin{equation}
\begin{aligned}
A (k) &= \ c_{1a}(k)^2 + c_{1b}(k)^2 + c_{2a}(k)^2 + c_{2b}(k)^2 + 5 c_3(k)^2, \\
B (k) &= \ \sqrt{
   c_0^2 [(c_{1a}(k) - i \ c_{1b}(k))^2 + (c_{2a}(k) - i \ c_{2b}(k))^2] + (c_{1b}(k) c_{2a}(k) - 
       c_{1a}(k) c_{2b}(k) + 2 c_3(k)^2)^2}, \\
c (k) &= \ -c_{1a}(k) + i \ (c_{1b}(k) + c_{2a}(k)) + c_{2b}(k), \\
e (k) &= \ c_{1a}(k) + i \ (c_{1b}(k) + c_{2a}(k) + i \ c_{2b}(k)), \\
f (k) &= \ c_{1a}(k)^2 + c_{1b}(k)^2 + c_{2a}(k)^2 + c_{2b}(k)^2, \\
g (k) &= \ c_{1a}(k) - i \ c_{1b}(k) + i \ c_{2a}(k) + c_{2b}(k), \\
h (k) &= \ c_{1b}(k) c_{2a}(k) - c_{1a}(k) c_{2b}(k) + 2 c_3(k)^2.
\end{aligned}
\end{equation}
The energies of $\hat{H}(k)$ can thus be written as

\begin{equation}
E_i(k) = \pm \sqrt{A(k) \pm 2B(k)},
\end{equation}
with the corresponding eigenstate

\begin{equation}
\ket{\psi_i(k)} = \frac{1}{\mathcal{N}_i(k)}
\mqty[\frac{c(k)}{3 c_3(k)-E_i(k)} \\ 
\frac{c_0(k)\left(\pm 2 B(k)+f(k)+4 c_3(k)\left(2 c_3(k)+E_i(k)\right)\right)}{2 c_0^2(k) g(k)+e(k)(h(k) \pm B(k))} \\
\frac{c_0(k) g(k) \left(3 c_3 (k) +E_i (k) \right)}{2 c_0^2(k) g(k)+e(k) (h(k) \pm B(k))} \\
1
],
\end{equation}
where the normalization factor $\mathcal{N}_i(k) \in \mathbb{C}$, and the sign in front of $B(k)$ is consistent with the corresponding energy $E_i(k)$ inside the square root. We adopted the gauge that sets the last component as a real number, and used the relation

\begin{equation}
\sqrt{-c(k) \ c_0^2(k) \ g(k)+h^2(k)}=B(k).
\end{equation}
It is important to note that the following derivation does not rely on the exact values of the eigenvalues. Instead, they only require a one-to-one correspondence between the energies $E_i(k)$ and their associated eigenstates $\ket{\psi_i(k)}$, which remains unchanged even in the presence of branch cuts. Thus, even though individual eigenvalue branches may exchange across Riemann sheets, the one-to-one correspondence between $E_i(k)$ and $\ket{\psi_i(k)}$ remains well defined under continuous tracking along $k$, which is precisely the situation in quantum simulations.

Since the eigenstate is still too complicated to extract the winding numbers, we further adopt the assumption

\begin{equation}
c_{1 b}(k)=i \ c_{1 a}(k), \quad c_{2 b}(k)=i \ c_{2 a}(k), \quad c_{1 a}(k), c_{2 a}(k) \in \mathbb{C}; \quad c_3(k) \in i \mathbb{R},
\label{eq:four-band/measurement_protocol/assumption}
\end{equation}
which also applies to the four-band generalized twister model $\hat{H}^{(4)}(k)$ as a special case. The eigenstates are simplified to 

\begin{equation}
\ket{\psi_i(k)} = \frac{1}{\mathcal{N}_i(k)} \mqty[
\frac{-2 c_{1 a} (k) +2 i c_{2 a} (k)}{3 c_3(k)-E_i(k)} \\
\frac{c_3(k) E_i(k)+2 c_3^2(k)\pm \sqrt{c_0^2(k)\left(c_{1 a}^2(k)+c_{2 a}^2(k)\right)+c_3^4(k)}}{c_0(k) \left(c_{1 a} (k) +i c_{2 a} (k) \right)} \\
\frac{3 c_3 (k) + E_i (k)}{2 c_0 (k)} \\
1
].
\end{equation}
Although these expressions appear algebraically complicated, what matters for our purposes is not their explicit form, but the fact that the energies and eigenstates share common functional dependencies which can be accessed through the measured observables. Notice that the third component is a linear function of $E_i(k)$, the difference of which, belonging to two eigenstates without the normalization factor, can lead to their winding number. This suggests a definition of a generalized stereographic projection built from observables that isolate this dependence 

\begin{equation}
\begin{aligned}
X_i(k) &= \mel*{\psi_i(k)}{(\hat{\mathbb{I}} - \hat{\sigma}^z)\otimes\hat{\sigma}^x}{\psi_i(k)}, \\
Y_i(k) &= \mel*{\psi_i(k)}{(\hat{\sigma}^z - \hat{\mathbb{I}})\otimes\hat{\sigma}^y}{\psi_i(k)}, \\
Z_i(k) &= \mel*{\psi_i(k)}{(\hat{\mathbb{I}} - \hat{\sigma}^z)\otimes(\hat{\mathbb{I}} - \hat{\sigma}^z)}{\psi_i(k)}, 
\label{eq:gen_stereo_projection}
\end{aligned}
\end{equation}
yielding 

\begin{equation}
\Lambda_i(k) = \frac{X_i(k) + iY_i(k)}{Z_i(k)} = \frac{3c_3(k) + E_i(k)}{2c_0(k)}.
\end{equation}
Then we can extract the information of the winding between two bands by noticing that

\begin{equation}
\Lambda_i(k) - \Lambda_j(k) = \frac{E_i(k) - E_j(k)}{2 c_0(k)}. 
\label{eq:four-band/measurement_protocol/integrand}
\end{equation}
Hence, the difference $\Lambda_i(k)-\Lambda_j(k)$ reproduces the same $k$-dependent phase evolution as the corresponding eigenvalue difference, up to a real prefactor that does not contribute to the winding.

For multi-band systems, as discussed in the previous section, two bands can evolve to distinct eigenstates after traversing the Brillouin zone. To address this issue in quantum simulations, it is unnecessary to extend the scan beyond a single Brillouin zone since different zones correspond to the same physical points and the required permutation information can already be inferred from continuity at $k=0$ and $k=2\pi$. Once the closed loops formed by the bands are identified, the winding number can be computed as

\begin{equation}
W_{ij}(k) = \frac{1}{2\pi i} \int_0^k dk' \, \partial_{k'} \ln [\Lambda_i(k') - \Lambda_j(k')],
\label{eq:four-band/measurement_protocol/winding_number_element}
\end{equation}
where each term corresponds to a distinct pair of bands, and all such pairs form a closed loop. \cref{eq:four-band/measurement_protocol/winding_number_element} is the analogue of \cref{eq:characterizing_non-Hermitian_braids/winding_number_element} since $\Lambda_i(k) \propto E_i(k)$. This method applies to all models that satisfy \cref{eq:H_4b} and \cref{eq:four-band/measurement_protocol/assumption}, which includes \textit{all} four-band twister models.

\subsection{Topologically Inequivalent Regions in Parameter Space}
\label{app:four-band/phase_diagram}

The first degeneracy in \cref{eq:four-band/energies} occurs when $S^2(k) = 0$. Once again, we deal with the real and imaginary parts separately. As such, we have 

\begin{equation}
\begin{aligned}
\Re(S^2(k)) &= 16m_0^4 + 81(m_1+1)^3(m_1x + x^2 - y^2) = 0, \\
\Im(S^2(k)) &= 81(m_1+1)^3(m_1 + 2x)y = 0.
\end{aligned}
\end{equation}
There are several cases to consider for $S^2(k) = 0$:
\begin{enumerate}
    \item If $m_0 = 0$ and $m_1 = -1$, then $\Re(S^2(k)) = \Im(S^2(k)) = 0$.
    \item $(x, y) = (1, 0) \Rightarrow \Re(S^2(k)) = 16m_0^4 + 81(m_1+1)^4 = 0$ and $\Im(S^2(k)) = 0$. \\ 
    \item $(x, y) = (-1, 0) \Rightarrow \Re(S^2(k)) = 16m_0^4 + 81(m_1+1)^3(1-m_1) = 0$ and $\Im(S^2(k)) = 0$. 
    \item $(x, y) = (-m_1/2, \sqrt{1-m_1^2/4}) \Rightarrow \Re(S^2(k)) = 16m_0^4 - 81(m_1+1)^3 = 0$ and $\Im(S^2(k)) = 0$ with $-2 \leq m_1 \leq 2$ since $\abs{x} = \abs{\cos k} \leq 1$. 
\end{enumerate}
The second degeneracy occurs at $S^2(k) - 5m_0^2 = 0$. Note that this is equivalent to $\tilde{S}^2(k) = S(k)^2 - 25m_0^4 = 0$. As such, we have  

\begin{equation}
\begin{aligned}
\Re(\tilde{S}^2(k)) &= m_0^4 - 9(m_1+1)^3(m_1x + x^2 - y^2) = 0, \\
\Im(\tilde{S}^2(k)) &= 81(m_1+1)^3(m_1 + 2x)y = 0. 
\end{aligned}
\end{equation}
There are several cases to consider:

\begin{enumerate}
    \item $(x, y) = (1, 0) \Rightarrow \Re(\tilde{S}^2(k)) = m_0^4 - 9(m_1+1)^4 = 0$ and $\Im(\tilde{S}^2(k)) = 0$.
    \item $(x, y) = (-1, 0) \Rightarrow \Re(\tilde{S}^2(k)) = m_0^4 - 9(m_1+1)^3(1-m_1) = 0$ and $\Im(\tilde{S}^2(k)) = 0$. 
    \item $(x, y) = (-m_1/2, \sqrt{1-m_1^2/4} \Rightarrow \Re(\tilde{S}^2(k)) = m_0^4 + 9(m_1+1)^3 = 0$ and $\Im(\tilde{S}^2(k)) = 0$. 
\end{enumerate}
Six equations define the phase boundaries shown in \cref{fig:H4_phase_diagram}, and a degenerate point at $(m_0, m_1) = (0, -1)$ where all six equations intersect. To summarize, these six equations are 

\begin{equation}
\begin{aligned}
F_1(m_0, m_1) &= 16m_0^4 + 81(m_1+1)^4 = 0, \\
F_2(m_0, m_1) &= 16m_0^4 + 81(m_1+1)^3(1-m_1) = 0, \\
F_3(m_0, m_1) &= 16m_0^4 - 81(m_1+1)^3 = 0 \qquad \text{with} \qquad \abs{m_1} \leq 2, \\
G_1(m_0, m_1) &= m_0^4 - 9(m_1+1)^4 = 0, \\
G_2(m_0, m_1) &= m_0^4 - 9(m_1+1)^3(1-m_1) = 0, \\
G_3(m_0, m_1) &= m_0^4 + 9(m_1+1)^3 = 0, \qquad \text{with} \qquad -2 \leq m_1 \leq -1. 
\end{aligned}
\end{equation}

\section{Knot Polynomials}
\label{app:knot_polynomials}

To further characterize the topological link structure, one may also compute knot polynomials using the braid words extracted from the band structure. 
In particular, the Alexander and Jones polynomials (as shown in \cref{tab:knot_link_braid_words_polynomials}) provide useful information about the link type associated with a given braid word. 
However, for multi-component links, these polynomials are not complete invariants and may depend on the projection used to define the braid diagram in contrast to the winding number matrix which is insensitive to the choice of a projection plane. 
To see this, we also give the Kauffman brackets for the closed braid diagrams. 
Here, $A = s^{-1/4}$, where $s$ is a nonzero monomial that characterize the Alexander and Jones polynomials. 
For compactness we write $\langle \beta\rangle_{\mathrm{cl}}$ for the Kauffman bracket of the closed braid diagram obtained from the braid word \(\beta\):
\begin{itemize}
    \item Solomon's knot:
    \begin{equation}
        \left\langle
        \sigma_1 \sigma_3 \sigma_2 \sigma_1 \sigma_3 \sigma_2
        \right\rangle_{\mathrm{cl}}
        =
        -\frac{A^{16}+A^8-A^4+1}{A^4}
        =
        -A^{12}-A^4+1-A^{-4},
    \end{equation}
    with writhe \(w=6\).

    \item Hopf chain:
    \begin{equation}
        \left\langle
        \sigma_1 \sigma_3 \sigma_1 \sigma_3 \sigma_2
        \right\rangle_{\mathrm{cl}}
        =
        -\frac{(A^8+1)^2}{A^5}
        =
        -A^{11}-2A^3-A^{-5},
    \end{equation}
    with writhe \(w=5\).

    \item Hopf link:
    \begin{equation}
        \left\langle
        \sigma_2\sigma_1\sigma_3\sigma_2
        \right\rangle_{\mathrm{cl}}
        =
        -A^2(A^8+1)
        =
        -A^{10}-A^2,
    \end{equation}
    with writhe \(w=4\).

    \item Unknot:
    \begin{equation}
        \left\langle
        \sigma_2 \sigma_1 \sigma_3
        \right\rangle_{\mathrm{cl}}
        =
        -A^9,
    \end{equation}
    with writhe \(w=3\).

    \item Unlink:
    \begin{equation}
        \left\langle
        \sigma_1 \sigma_3
        \right\rangle_{\mathrm{cl}}
        =
        -A^4(A^4+1)
        =
        -A^8-A^4,
    \end{equation}
    with writhe \(w=2\).

    \item Hopf link + unlink:
    \begin{equation}
        \left\langle
        \sigma_2^2
        \right\rangle_{\mathrm{cl}}
        =
        -\frac{(A^4+1)^2(A^8+1)}{A^8}.
    \end{equation}
    Equivalently,
    \begin{equation}
        \left\langle
        \sigma_2^2
        \right\rangle_{\mathrm{cl}}
        =
        -A^8-2A^4-2-2A^{-4}-A^{-8},
    \end{equation}
    with writhe \(w=2\).

    \item Unknot + unlink:
    \begin{equation}
        \left\langle
        \sigma_2
        \right\rangle_{\mathrm{cl}}
        =
        -\frac{(A^4+1)^2}{A}
        =
        -A^7-2A^3-A^{-1},
    \end{equation}
    with writhe \(w=1\).

    \item Double unlinks:
    \begin{equation}
        \left\langle
        \mathrm{Empty}_{B_4}
        \right\rangle_{\mathrm{cl}}
        =
        -\frac{(A^4+1)^3}{A^6}
        =
        -A^6-3A^2-3A^{-2}-A^{-6},
    \end{equation}
    with writhe \(w=0\).
\end{itemize}
Here, the Kauffman bracket of the closed braid is $\langle \cdot \rangle_{\mathrm{cl}}$.

A summary of the Alexander and Jones polynomials for the knots/links is shown in \cref{tab:knot_link_braid_words_polynomials}. 

\begin{table*}[htbp!]
    \centering
\begin{tabular}{c|c|c|c|c}
    \toprule 
    \textbf{Twister Model} & \textbf{Knot/Link} & \textbf{Braid Word} 
    & \textbf{Alexander Polynomial} $\Delta_L(s)$ 
    & \textbf{Jones Polynomial} $V_L(s)$ \\
    \midrule 
    \multirow{3}{*}{Two-band}
    & Hopf link 
    & $\sigma_1^2$ 
    & $1-s$ 
    & $-s^{1/2}-s^{5/2}$ \\
    
    & Unknot 
    & $\sigma_1$ 
    & $1$ 
    & $1$ \\
    
    & Unlink 
    & Empty 
    & $0$ 
    & $-s^{-1/2}-s^{1/2}$ \\
    
    \midrule
    
    \multirow{8}{*}{Four-band}
    & Solomon's knot 
    & $\sigma_1\sigma_3\sigma_2\sigma_1\sigma_3\sigma_2$ 
    & $(1-s)(1+s^2)$ 
    & $-s^{3/2}-s^{7/2}+s^{9/2}-s^{11/2}$ \\
    
    & Hopf chain 
    & $\sigma_1\sigma_3\sigma_1\sigma_3\sigma_2$ 
    & $(1-s)^2$ 
    & $s(1+s^2)^2$ \\
    
    & Hopf link 
    & $\sigma_2\sigma_1\sigma_3\sigma_2$ 
    & $1-s$ 
    & $-s^{1/2}-s^{5/2}$ \\
    
    & Unknot 
    & $\sigma_2\sigma_1\sigma_3$ 
    & $1$ 
    & $1$ \\
    
    & Unlink 
    & $\sigma_1\sigma_3$ 
    & $0$ 
    & $-s^{-1/2}-s^{1/2}$ \\
    
    & Hopf link + unlink 
    & $\sigma_2^2$ 
    & $0$ 
    & $-s^{-1/2}(1+s)^2(1+s^2)$ \\
    
    & Unknot + unlink 
    & $\sigma_2$ 
    & $0$ 
    & $s^{-1}(1+s)^2$ \\
    
    & Double unlinks 
    & Empty 
    & $0$ 
    & $-s^{-3/2}(1+s)^3$ \\ 
    \bottomrule
\end{tabular}
    \caption{\textbf{Summary of braid words and knot polynomials for various knots/links in the two- and four-band twister models \cref{eq:quantum_hardware_simulation/two-band_twister_model/H2_twister_model,eq:quantum_hardware_simulation/four-band_twister_model/H4_twister_model} respectively.} The braid words are obtained in the discussion leading to \cref{eq:results/quantum_circuit_implementation/characterizing_braids_and_knotted_structures/W_crossing_condition} from which the winding number matrices are computed from the measured observables on quantum hardware as depicted in \cref{fig:methodology_illustration}.}
    \label{tab:knot_link_braid_words_polynomials}
\end{table*}

\subsection{Computation of the Alexander Polynomial}
\label{app:knot_polynomials/computation_of_the_Alexander_polynomials}

The Alexander polynomial $\Delta_L(s)$ of an oriented link $L$ is defined via its Seifert surface ~\cite{Alexander1928,Seifert1935,Lickorish1997,Li2019c,VanWijk2006}. One chooses an oriented basis of the first homology group of the surface with integer coefficients, which is the abelian group generated by one-dimensional cycles on the surface modulo boundaries. Using this basis, one constructs a square Seifert matrix $M$ whose elements are integers~\cite{Seifert1935,Lickorish1997}. 
The Alexander polynomial $\Delta_L(s)$ is defined up to an overall rescaling by a nonzero monomial $\pm s^a$ with integer $a$ as 
\begin{equation}
\Delta_L(s) \doteq \det(M - sM^T).
\label{eq:knot_polynomials/computation_of_knot_polynomials/Alexander_polynomial}
\end{equation}
Equivalently, it can be characterized as a unique (up to $\pm s^a$) link invariant satisfying the skein relation
\begin{equation}
\Delta_{L_+}(s) - \Delta_{L_-}(s) = \bigl(s^{1/2}-s^{-1/2}\bigr)\,\Delta_{L_0}(s),
\end{equation}
for three oriented link diagrams $L_+$, $L_-$, and $L_0$ (planar projections with local crossing information) that agree everywhere except in a small region, where $L_+$ has a positive crossing, $L_-$ the corresponding negative crossing, and $L_0$ the oriented smoothing obtained by reconnecting the strands without a crossing~\cite{Conway1970,Lickorish1997}. These conditions are supplemented by the normalization $\Delta_{\mathrm{unknot}}(s)=1$~\cite{Conway1970,Lickorish1997}. 

For links presented as closures of $N$-strand braids, we compute $\Delta_L(s)$ directly from the braid word using the reduced Burau representation~\cite{Lickorish1997}. The reduced Burau representation is a homomorphism

\begin{equation}
\bar\rho: B_N \longrightarrow \mathrm{GL}_{N-1}\bigl(\mathbb{Z}[s^{\pm 1}]\bigr),
\end{equation}
which assigns an explicitly known $(N-1) \times (N-1)$ matrix over the Laurent polynomial ring $\mathbb{Z}[s^{\pm 1}]$ to each Artin generator $\sigma_i$. Given a braid word

\begin{equation}
\beta = \sigma_{i_1}^{\epsilon_1}\sigma_{i_2}^{\epsilon_2}\cdots\sigma_{i_m}^{\epsilon_m}, \qquad \epsilon_k\in\{\pm 1\},
\end{equation}
we obtain a matrix representative
\begin{equation}
B_\beta(s) = \bar\rho(\sigma_{i_1})^{\epsilon_1}\,\bar\rho(\sigma_{i_2})^{\epsilon_2}\cdots \bar\rho(\sigma_{i_m})^{\epsilon_m}
\in \mathrm{GL}_{n-1}\bigl(\mathbb{Z}[s^{\pm 1}]\bigr)
\end{equation}
by multiplying the Burau matrices in the order prescribed by the word. The Alexander polynomial of the oriented link $L$ given by the closure of $\beta$ is then extracted from the determinant
\begin{equation}
\det\bigl(\mathbb{I}_{N-1} - B_\beta(s)\bigr),
\end{equation}
where $\mathbb{I}_{N-1}$ is the $(N-1) \times (N-1)$ identity matrix. More precisely, one has

\begin{equation}
\Delta_L(s) \doteq \frac{\det\bigl(\mathbb{I}_{N-1} - B_\beta(s)\bigr)}{1 - s},
\label{appxeqn:Alexpoly_final}
\end{equation}
where $\doteq$ denotes equality up to multiplication by a unit $\pm s^a$ in $\mathbb{Z}[s^{\pm 1}]$~\cite{Alexander1928,Lickorish1997}. The factor $1 - s$ removes the trivial eigenvalue associated with the overall longitudinal direction around the braid axis and ensures that the resulting Laurent polynomial is invariant under Markov moves on the braid, and hence depends only on the isotopy class of the closed link.

\subsubsection{Computing the Alexander Polynomial for Solomon's knot}

Consider the braid word
\begin{equation}
    \beta=\sigma_1\sigma_3\sigma_2\sigma_1\sigma_3\sigma_2
    =(\sigma_1\sigma_3\sigma_2)^2
    \in B_4 .
\end{equation}
Since this is a four-strand braid, the reduced Burau representation is
three-dimensional. We use the convention

\begin{equation}
\bar{\rho}(\sigma_1) =
\mqty[-s & 1 & 0 \\
0 & 1 & 0 \\
0 & 0 & 1
],
\qquad
\bar{\rho}(\sigma_2) =
\mqty[
1 & 0 & 0 \\
s & -s & 1 \\
0 & 0 & 1
],
\qquad
\bar{\rho}(\sigma_3)=
\mqty[
1 & 0 & 0 \\
0 & 1 & 0 \\
0 & s & -s
]. 
\end{equation}
First, we define
\begin{equation}
    C(s) := \bar{\rho}(\sigma_1)\bar{\rho}(\sigma_3)\bar{\rho}(\sigma_2).
\end{equation}
A direct multiplication gives

\begin{equation}
C(s) =
\mqty[
0 & -s & 1 \\
s & -s & 1 \\
s^2 & -s^2 & 0
].
\end{equation}
Therefore
\begin{equation}
    B_\beta(s)=\bar{\rho}(\beta)=C(s)^2.
\end{equation}
Since $\beta$ comprises two copies of $\sigma_1\sigma_3\sigma_2$, squaring $C(s)$ gives 

\begin{equation}
B_\beta(s) =
\mqty[
0 & 0 & -s \\
0 & -s^2 & 0 \\
-s^3 & 0 & 0
]. 
\end{equation}
Hence

\begin{equation}
\mathbb{I}_3-B_\beta(s) = 
\mqty[
1 & 0 & s \\
0 & 1+s^2 & 0 \\
s^3 & 0 & 1
].
\end{equation}
The determinant is
\begin{equation}
\begin{aligned}
\det\!\left(\mathbb{I}_3-B_\beta(s)\right)
&=
(1+s^2)
\det
\mqty[
1 & s \\
s^3 & 1
] \\
&= (1+s^2)(1-s^4) \\
&= (1+s^2)(1-s)(1+s)(1+s^2) \\
&= (1-s)(1+s)(1+s^2)^2.
\end{aligned}
\end{equation}
Equivalently,
\begin{equation}
\det\!\left(\mathbb{I}_3-B_\beta(s)\right)
=
-(s-1)(s+1)(s^2+1)^2 .
\end{equation}
For the closure of a four-strand braid, the standard reduced Burau formula is
\begin{equation}
    \Delta_{\beta_\mathrm{cl}}(s)
    \doteq
    \frac{1-s}{1-s^4}
    \det\!\left(\mathbb{I}_3-B_\beta(s)\right),
\end{equation}
where $\doteq$ denotes equality up to multiplication by a unit
\(\pm s^a\). Substituting the determinant gives

\begin{equation}
\begin{aligned}
\Delta_{\beta_\mathrm{cl}}(s)
&\doteq
\frac{1-s}{1-s^4}
(1+s^2)(1-s^4) \\
&= (1-s)(1+s^2).
\end{aligned}
\end{equation}
Thus
\begin{equation}
    \boxed{
    \Delta_{\beta_\mathrm{cl}}(s)
    \doteq
    (1-s)(1+s^2)
    =
    1-s+s^2-s^3 .
    }
\end{equation}

Conversely, if one considers the alternative normalization

\begin{equation}
    \Delta_{\beta_\mathrm{cl}}(s)
    \doteq
    \frac{\det(\mathbb{I}_3-B_\beta(s))}{1-s}, 
\end{equation}
then one obtains

\begin{equation}
    \Delta_{\beta_\mathrm{cl}}(s)
    \doteq
    (1+s)(1+s^2)^2 .
\end{equation}
This is not equivalent to the standard closed-\(B_4\) normalization above as it corresponds to a different normalization convention.

\subsection{Computation of the Jones Polynomial}
\label{app:knot_polynomials/computation_of_the_Jones_polynomials}

The Jones polynomial $V_L(s)$ is defined in terms of the Kauffman bracket of an unoriented link diagram which is a Laurent polynomial in a variable $A$ determined by the local relation
\begin{equation}
    \langle D_+ \rangle = A\,\langle D_0 \rangle + A^{-1}\,\langle D_\infty \rangle,
\end{equation}
where $D_+$ is a link diagram with a specified crossing, and $D_0$ and $D_\infty$ are the link diagrams obtained from $D_+$ by replacing the chosen crossing with the two possible smoothings in which the four incident strand endpoints are reconnected in one of the two non-intersecting pairings~\cite{Kauffman1987,Lickorish1997}. In addition, one imposes the loop relations 
\begin{equation}
    \langle D\cup\bigcirc \rangle = -(A^2 + A^{-2})\,\langle D \rangle, \qquad \langle \bigcirc \rangle = 1,
\end{equation}
where $D\cup\bigcirc$ denotes the disjoint union of the diagram $D$ with a simple closed curve ($\bigcirc$) in the plane. For an oriented diagram of a link with writhe $w(D)$, defined as the sum of the signs of all crossings, the Jones polynomial is then
\begin{equation}
    V_L(s) = (-A^3)^{-w(D)}\,\langle D \rangle,
\end{equation}
normalized such that $V_{\mathrm{unknot}}(s) = 1$~\cite{JonesFields,Kauffman1987}. 

\subsubsection{Computing the Jones Polynomial for Solomon's knot}

We now compute the Jones polynomial using the Kauffman bracket convention
\begin{equation}
    V_L(s)=(-A^3)^{-w(D)}\langle D\rangle,
\end{equation}
where $w(D)$ is the writhe of the oriented braid diagram. Since the braid
\begin{equation}
    \beta=\sigma_1\sigma_3\sigma_2\sigma_1\sigma_3\sigma_2
\end{equation}
contains six positive crossings; its writhe is
\begin{equation}
    w(\beta_{\mathrm{cl}})=6 .
\end{equation}
To evaluate the Kauffman bracket, we use the Temperley--Lieb representation
of the braid generators

\begin{equation}
    g_i=A\,\mathbb{I}+A^{-1}e_i ,
\end{equation}
where the Temperley--Lieb generators satisfy

\begin{equation}
    e_i^2=\delta e_i,
    \qquad
    e_i e_{i\pm1} e_i=e_i,
    \qquad
    e_i e_j=e_j e_i \qquad (|i-j|\geq2),
\end{equation}
with
\begin{equation}
    \delta=-(A^2+A^{-2}).
\end{equation}
The braid word in the Temperley-Lieb representation corresponds to

\begin{equation}
    g_\beta = g_1g_3g_2g_1g_3g_2 .
\end{equation}
Expanding this product in the Temperley--Lieb algebra $TL_4$, closing the
resulting diagrams, and grouping terms according to the number of closed
loops gives
\begin{equation}
\langle \beta_{\mathrm{cl}}\rangle
=
A^6\delta^3
+
(3A^4+1)\delta^2
+
(5A^2+2A^{-2})\delta
+
4 .
\end{equation}
Substituting
\begin{equation}
    \delta=-(A^2+A^{-2})
\end{equation}
yields

\begin{equation}
\begin{aligned}
\langle \beta_{\mathrm{cl}}\rangle
&=
A^6[-(A^2+A^{-2})]^3
+
(3A^4+1)[-(A^2+A^{-2})]^2 \\ 
&\qquad
+
(5A^2+2A^{-2})[-(A^2+A^{-2})]
+
4 \\
&=
-A^{12}-A^4+1-A^{-4} \\
&=
-\frac{A^{16}+A^8-A^4+1}{A^4}.
\end{aligned}
\end{equation}
Thus, the Jones polynomial is
\begin{equation}
\begin{aligned}
V_{\beta_{\mathrm{cl}}}(s)
&=
(-A^3)^{-6}\langle \beta_{\mathrm{cl}}\rangle \\
&=
A^{-18}
\left(
-\frac{A^{16}+A^8-A^4+1}{A^4}
\right) \\
&=
-\frac{A^{16}+A^8-A^4+1}{A^{22}} \\
&=
-A^{-6}-A^{-14}+A^{-18}-A^{-22}.
\end{aligned}
\end{equation}
Finally, using $A=s^{-1/4}$, we obtain

\begin{equation}
\begin{aligned}
V_{\beta_{\mathrm{cl}}}(s)
&=
-s^{3/2}-s^{7/2}+s^{9/2}-s^{11/2} \\
&=
-s^{3/2}\left(1+s^2-s^3+s^4\right).
\end{aligned}
\end{equation}
Thus, we have 

\begin{equation}
    \boxed{
    V_{\beta_{\mathrm{cl}}}(s)
    =
    -s^{3/2}-s^{7/2}+s^{9/2}-s^{11/2}
    =
    -s^{3/2}\left(1+s^2-s^3+s^4\right).
    }
    \label{appxeqn:Jonespoly_final}
\end{equation}

\subsection{Comparison Between Alexander and Jones Polynomials for Two- and Four-band Braid Closures}

In the two-band and four-band twister models, the relevant topological object is obtained by taking the closure of a braid word. We write \(L_{\mathrm{cl}}(\beta)\) for the closed braid link obtained from a braid word \(\beta\). Thus, for \(\beta\in B_2\), the closure may contain either one or two components, while for \(\beta\in B_4\), the closure may contain one, two, three, or four components. The number of closed components is determined by the cycle structure of the braid permutation and is therefore already encoded in the braid word.

The one-variable Alexander polynomial and the Jones polynomial treat the split multi-component closures differently. With the Jones polynomial convention used in the main text, we have

\begin{equation}
    V_L(s)=(-A^3)^{-w(D)}\langle D\rangle,
    \qquad
    A=s^{-1/4},
\end{equation}
the Kauffman-bracket loop relation is
\begin{equation}
    \langle D\cup \bigcirc\rangle
    =
    -(A^2+A^{-2})\langle D\rangle .
\end{equation}
Therefore, after substituting \(A=s^{-1/4}\), the Jones polynomial satisfies
\begin{equation}
    V_{L\sqcup \bigcirc}(s)
    =
    -\left(s^{1/2}+s^{-1/2}\right)V_L(s).
\end{equation}
Thus each additional split trivial component contributes an explicit multiplicative factor
\begin{equation}
    -\left(s^{1/2}+s^{-1/2}\right).
\end{equation}
If \(U_m\) denotes a \(m\)-component unlink, then

\begin{equation}
    V_{U_m}(s)
    =
    (-1)^{m-1}
    \left(s^{1/2}+s^{-1/2}\right)^{m-1}.
\end{equation}
For example,
\begin{align}
    V_{U_1}(s)&=1,\\
    V_{U_2}(s)&=-\left(s^{1/2}+s^{-1/2}\right),\\
    V_{U_3}(s)&=\left(s^{1/2}+s^{-1/2}\right)^2,\\
    V_{U_4}(s)&=-\left(s^{1/2}+s^{-1/2}\right)^3.
\end{align}
This behavior is already visible in the \(B_2\) examples. The closure of \(\sigma_1\in B_2\) is a single unknot, whereas the closure of the empty two-strand braid is the two-component unlink. Hence
\begin{align}
    V_{L_{\mathrm{cl}}(\sigma_1)}(s)&=1,\\
    V_{L_{\mathrm{cl}}(\mathrm{Empty}_{B_2})}(s)
    &=
    -\left(s^{1/2}+s^{-1/2}\right).
\end{align}
The Jones polynomial also distinguishes the two-strand Hopf link,
\begin{equation}
    V_{L_{\mathrm{cl}}(\sigma_1^2)}(s)
    =
    -s^{1/2}-s^{5/2}.
\end{equation}
Therefore, for the \(B_2\) examples considered here, the Jones polynomial distinguishes the Hopf link, the unknot, and the two-component unlink.

For the \(B_4\) examples, the same mechanism becomes more pronounced because more split components can occur. For instance,
\begin{align}
    V_{L_{\mathrm{cl}}(\sigma_1\sigma_3)}(s)
    &=
    -\left(s^{1/2}+s^{-1/2}\right),\\
    V_{L_{\mathrm{cl}}(\sigma_2)}(s)
    &=
    \left(s^{1/2}+s^{-1/2}\right)^2,\\
    V_{L_{\mathrm{cl}}(\mathrm{Empty}_{B_4})}(s)
    &=
    -\left(s^{1/2}+s^{-1/2}\right)^3.
\end{align}
These are unlink-type closures with different numbers of components, and the Jones polynomial retains this component information through the Kauffman-bracket loop factor.

By contrast, the standard one-variable Alexander polynomial vanishes for split multi-component links

\begin{equation}
    \Delta_{L_1\sqcup L_2}(s)=0,
\end{equation}
provided both \(L_1\) and \(L_2\) are non-empty links. Hence, split unlink-type closures collapse to the same Alexander polynomial. In the \(B_2\) case,

\begin{equation}
    \Delta_{L_{\mathrm{cl}}(\mathrm{Empty}_{B_2})}(s)=0,
\end{equation}
whereas
\begin{align}
    \Delta_{L_{\mathrm{cl}}(\sigma_1)}(s)&=1,\\
    \Delta_{L_{\mathrm{cl}}(\sigma_1^2)}(s)&\doteq 1-s.
\end{align}
Thus the Alexander polynomial still distinguishes the listed \(B_2\) examples.

In the \(B_4\) case, however, several different split closures have the same one-variable Alexander polynomial:
\begin{align}
    \Delta_{L_{\mathrm{cl}}(\sigma_1\sigma_3)}(s)&=0,\\
    \Delta_{L_{\mathrm{cl}}(\sigma_2)}(s)&=0,\\
    \Delta_{L_{\mathrm{cl}}(\sigma_2^2)}(s)&=0,\\
    \Delta_{L_{\mathrm{cl}}(\mathrm{Empty}_{B_4})}(s)&=0.
\end{align}
The Jones polynomial separates these cases because it retains the additional unlink-component factors whereas the one-variable Alexander polynomial collapses them to zero.

Therefore, for the two-band examples, both the Alexander and Jones polynomials distinguish the listed closures. 
For the four-band examples, the Jones polynomial provides a finer distinction among split unlink-type closures because it keeps track of the Kauffman-bracket loop factors associated with additional trivial components. 
This does not imply that the Jones polynomial is a complete invariant; distinct knots or links can still share the same Jones polynomial. 
Rather, for the braid closures considered here, the Jones polynomial retains component information that is lost by the one-variable Alexander polynomial.

\section{Global Biorthogonal Berry Phase}
\label{app:GBBP}

In this section, we discuss two issues regarding the global biorthogonal Berry phase. Firstly, we show that branch cuts are not taken properly into account in their usual definition, and we show how to remedy this. Next, we show that the remedied global biorthogonal Berry phase is still incapable of uniquely characterizing each knot/link by using the two-band twister model as an example, thus motivating the winding number \cref{eq:characterizing_non-Hermitian_braids/winding_number_element}.

\subsection{Branch Cut in the Global Biorthogonal Berry Phase}
\label{app:GBBP/branch_cut}

Since the global biorthogonal Berry phase is a $\mathbb{Z}_2$ invariant (taking only the values $0$ or $1$), it is insufficient to distinguish topologically inequivalent objects whose band structures contain more than two knots or links. Nevertheless, this quantity has been widely used as a topological invariant in the literature. In this section, we revisit its definition and computation, and highlight the branch cut that was often overlooked in previous works. We then provide the correct treatment so that the global biorthogonal Berry phase can be evaluated in a mathematically consistent manner. 

\begin{definition}
The global biorthogonal Berry phase is given by \cite{PhysRevLett.126.010401}

\begin{equation}
\gamma = \sum_m \gamma^{(m)} = -\frac{i}{\pi} \sum_m \int_0^{2\pi} dk' \, \mel{\psi_m^L(k')}{\partial_{k'}}{\psi_m^R(k')} \mod 2, 
\label{eq:GBBP/GBBP}
\end{equation}
where $m$ denotes the band index and $\mel{\psi_m^L(k)}{\partial_k}{\psi_m^R(k)}$ is the diagonal element of the non-abelian Berry connection matrix.
\end{definition}
\cref{eq:GBBP/GBBP} is a global quantity as it is summed over all bands. A two-band non-Hermitian Hamiltonian can be written as 

\begin{equation}
\hat{H}(k) = \vb{d}(k)\cdot\vb*{\sigma} =  \mqty[
d_z(k) & d_x(k) - id_y(k) \\
d_x(k) + id_y(k) & -d_z(k)
],  
\end{equation}
where $\vb{d}(k) = [d_x(k), d_y(k), d_z(k)] \in \mathbb{C}^3$ and $\vb*{\sigma} = [\sigma^x, \sigma^y, \sigma^z]$ are the Pauli operators. Thus, its energies are

\begin{equation}
E_\pm(k) = \pm\sqrt{d_x(k)^2 + d_y(k)^2 + d_z(k)^2}, 
\label{eq:Non-Hermitian_knots_and_links/twister_family_of_Hamiltonians/GBBP/energies}
\end{equation}
and its right and left eigenstates are \cite{PhysRevA.109.053311}

\begin{equation}
\begin{aligned}
\ket{\psi_\pm^R(k)} &= \frac{1}{\sqrt{2}}\mqty[
\pm R(k)\sqrt{1 \pm D(k)} \\
\sqrt{1 \mp D(k)}
] = \frac{1}{\sqrt{2}}\mqty[a_\pm(k) \\ b_\pm(k)], \\
\bra{\psi_\pm^L(k)} &= \frac{1}{\sqrt{2}}\mqty[
\pm \frac{\sqrt{1 \pm D(k)}}{R(k)} & \sqrt{1 \mp D(k)}
], 
\label{eq:Non-Hermitian_knots_and_links/twister_family_of_Hamiltonians/GBBP/eigenstates}
\end{aligned}
\end{equation}
where 

\begin{equation}
D(k) = \frac{d_z(k)}{E_+(k)}, \qquad R(k) = \sqrt{\frac{d_x(k)-id_y(k)}{d_x(k)+id_y(k)}}. 
\end{equation}
Note that this parametrization is useful for computing the usual (LL or RR) or biorthogonal (LR or RL) quantities not limited to 1D winding numbers, such as 2D Chern invariants \cite{shen2018topological, yao2018non} or even their higher-dimensional counterparts \cite{ghorashi2021non, gu2016holographic,zhu2023higher}.

Let $\Psi(k) = \mqty[\ket{\psi_+^R(k)} & \ket{\psi_-^R(k)}] = \mqty[a_+(k) & a_-(k) \\ b_+(k) & b_-(k)]$ then we can rewrite \cref{eq:GBBP/GBBP} in a slightly different way as such \cite{PhysRevLett.126.010401}

\begin{equation}
\begin{aligned}
\gamma &= -\frac{i}{\pi} \int_0^{2\pi} dk' \, \Tr[\Psi^{-1}(k')\partial_{k'}\Psi(k')] 
= -\frac{i}{\pi} \int_0^{2\pi} dk' \, \partial_{k'} \Tr[\ln \Psi(k')] 
= i\ln \frac{\det\Psi(2\pi)}{\det\Psi(0)}. 
\end{aligned}
\end{equation}
The global biorthogonal Berry phase is related to the permutation of the bands by 
\begin{equation}
e^{i\gamma} = (-1)^P.
\label{eq:GBBP/GBBP/permutation}
\end{equation}
Since there are only two bands, there are only two possible permutations, even or odd. Using Jacobi's formula 
\begin{equation}
\partial_k \det\Psi(k) = \det\Psi(k)\Tr[\Psi^{-1}(k)\partial_k\Psi(k)], 
\end{equation}
we get 
\begin{equation}
\begin{aligned}
\gamma &= -\frac{i}{\pi} \int_0^{2\pi} dk' \, \Tr[\Psi^{-1}(k')\partial_{k'}\Psi(k')] 
= -\frac{i}{\pi} \int_0^{2 \pi} dk' \, \frac{\partial_{k'} \det\Psi(k')}{\det\Psi(k')} 
= -\frac{i}{\pi} \int_0^{2 \pi} dk' \, \partial_{k'} \ln \det\Psi(k'). 
\label{eq:GBBP/GBBP_det_form}
\end{aligned}
\end{equation}
If the permutation is even, then $P = \gamma = 0$ in \cref{eq:GBBP/GBBP/permutation} which means that $a_\pm(0) = a_\pm(2\pi)$ and $b_\pm(0) = b_\pm(2\pi)$. The non-trivial case is when the permutation is odd, then $P = 1$ and $\gamma = \pi$ in \cref{eq:GBBP/GBBP/permutation}. As such, we have $a_\pm(0) = b_\pm(2\pi)$ and $b_\pm(0) = a_\pm(2\pi)$ instead which corresponds to a band swap. Since $\hat{H}(k) = \hat{H}(k + 2\pi)$, its right eigenstates $\ket*{\psi_\pm^R(k)}$ must also satisfy $\ket*{\psi_\pm^R(k)} = \ket*{\psi_\pm^R(k + 2\pi)}$ which implies that $D(k)$ and $R(k)$ are also periodic in $2\pi$. To ensure that the periodicity is satisfied, only an even permutation is allowed which is a contradiction. To see this, we evaluate $\det\Psi(0)$ and $\det\Psi(2\pi)$ explicitly and get  
\begin{equation}
\det\Psi(0) = \mqty|
a_+(0) & a_-(0) \\
b_+(0) & b_-(0)
| = 
\mqty|
b_+(2\pi) & b_-(2\pi) \\
a_+(2\pi) & a_-(2\pi)
| = -\det\Psi(2\pi). 
\end{equation}
This is possible as $\det\Psi(0) = -\det\Psi(2\pi)$ physically describes a band swap. However, \cref{eq:Non-Hermitian_knots_and_links/twister_family_of_Hamiltonians/GBBP/eigenstates} does not account for this possibility and implies that only an even permutation is allowed. The origin of this contradiction lies in the branch cut in $E_\pm(k)$ \cref{eq:Non-Hermitian_knots_and_links/twister_family_of_Hamiltonians/GBBP/energies} and $\sqrt{1 \pm D(k)}$ of the eigenstates in \cref{eq:Non-Hermitian_knots_and_links/twister_family_of_Hamiltonians/GBBP/eigenstates}. If one uses \cref{eq:Non-Hermitian_knots_and_links/twister_family_of_Hamiltonians/GBBP/eigenstates} to evaluate $\det\Psi(k)$ directly, one gets $\det\Psi(k) = R(k)$ thus implying that only $d_x(k)$ and $d_y(k)$ matter. However, this cannot be true. Due to the branch cut which involves $\vb{d}(k)$, $D(k)$ is also relevant. 
To resolve this, we need to account for the branch cut or, equivalently, the band swaps so that the global biorthogonal Berry phase reflects this correctly. 

Conventionally, the branch cut is chosen to be the negative real axis. Since there exists a branch cut in $k\in[0, 2\pi]$, there is a discontinuity in \cref{eq:GBBP/GBBP_det_form}. If $E_\pm(k)$ are smooth for $k\in[0, 2\pi]$ except at $k = k_0$, then the band swap can be formalized as 
\begin{equation}
\lim_{k \to k_0^-} \Psi(k) = \lim_{k \to k_0^+} \Psi(k)S, 
\end{equation}
where $S = \mqty[0 & 1 \\ 1 & 0]$. Thus, there is a correction factor such that 
\begin{equation}
\ln\det\Psi(k) = \ln\det\Psi_{\text{smooth}}(k) - i\pi\Theta(k - k_0), 
\end{equation}
where $\Theta(k - k_0)$ is the Heaviside step function, and the integrand of \cref{eq:GBBP/GBBP_det_form} should be 
\begin{equation}
\partial_k\ln\det\Psi(k) = \partial_k\ln\det\Psi_{\text{smooth}}(k) - i\pi\delta(k - k_0). 
\end{equation}
Integrating $i\pi\delta(k - k_0)$ yields 
\begin{equation}
\Delta \gamma = -i \int_0^{2\pi} dk' \; i\pi \, \delta(k' - k_0) = \pi.
\end{equation}
Evaluating \cref{eq:GBBP/GBBP_det_form} yields
\begin{equation}
\gamma = \left(-\frac{i}{\pi}\int_0^{2\pi} dk' \, \partial_{k'} \ln R(k') + \nu_E \right) \mod 2, 
\label{eq:GBBP/correct_GBBP}
\end{equation}
where $\nu_E$ is the number of band swaps of a band $E_+(k)$ or $E_-(k)$ for $k\in[0, 2\pi]$. Numerically, we can compute \cref{eq:GBBP/correct_GBBP} as \cite{yu2021experimental}
\begin{equation}
\gamma = \frac{1}{\pi} \sum_{i, n} \Im\bigg(\ln \frac{\braket{\psi_n^L(k_{i+1})}{\psi_n^R(k_i)}}{\braket{\psi_n^L(k_i)}{\psi_n^R(k_i)}}\bigg), 
\label{eq:GBBP/correct_GBBP_numerical}
\end{equation}
if the eigenstates are sorted correctly by accounting for the band swaps. 
\begin{figure}[t]
\centering
\includegraphics[width=0.95 \linewidth]{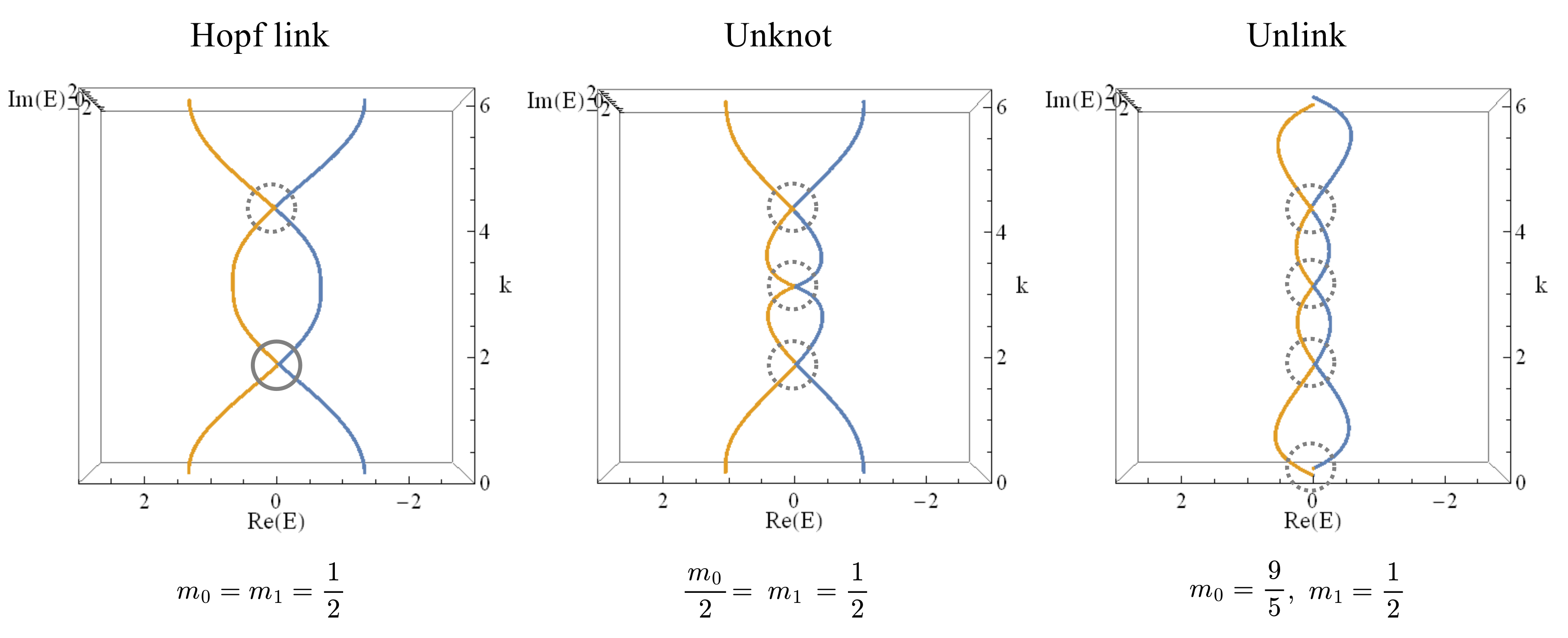}
\caption{\textbf{$\nu_E$ for the Hopf link, unknot, and unlink of \cref{eq:quantum_hardware_simulation/two-band_twister_model/H2_twister_model}.} The blue/orange plot denotes $E_{+}(k) = \sqrt{d_x^2(k)+d_y^2(k)+d_z^2(k)}$ or $E_{-}(k) = -\sqrt{d_x^2(k)+d_y^2(k)+d_z^2(k)}$. Gray circles give the position of the swapping points, and a dashed/solid circle denotes an over-/under-crossing. } 
\label{fig: NH_E}
\end{figure}

\subsection{Inadequacy of the Global Biorthogonal Berry Phase for Determining Knots/Links}
\label{app:GBBP/computing_nu_E}

With the definition above, we now analyze the global biorthogonal Berry phase \cref{eq:GBBP/correct_GBBP} for $\hat{H}^{(2)}(k)$ in \cref{eq:quantum_hardware_simulation/two-band_twister_model/H2_twister_model}. This disagrees with Ref. \cite{PhysRevLett.127.090501}, which computed the winding between two bands using the global biorthogonal Berry phase. First, we calculate the contribution from $R(k)$ to the global biorthogonal Berry phase \cref{eq:GBBP/correct_GBBP}. Here, $R(k) = \sqrt{\frac{e^{i k}\left(e^{i k}+m_1\right)}{m_1+1}}$ and $m_1 \in \mathbb{R}\setminus\{-1\}$. Since the period of the Hamiltonian is $2\pi$, we have
\begin{equation}
\gamma = \frac{1}{2 \pi i} \ln \left[\frac{e^{i 2 \pi}\left(e^{i 2 \pi}+m_1\right)}{m_1+1}\right] - \frac{1}{2 \pi i} \ln \left[\frac{e^{i 0}\left(e^{i 0}+m_1\right)}{m_1+1}\right] + \nu_E  \mod  2 = \nu_E \mod  2.
\end{equation}
So the \textit{only} contribution to the global biorthogonal Berry phase $\gamma$ is from $\nu_E$.

We now show that $\gamma$ is unable to uniquely characterize each knot/link in the family of twister models, thus motivating the winding number defined in \cref{eq:characterizing_non-Hermitian_braids/winding_number_element}. Since $\gamma$ is fully determined by $\nu_{E}$, we aim to find the $k$ points where 
\begin{equation}
E_+^2(k) = \left( e^{2ik} + m_1 e^{ik} \right)(m_1 + 1) - m_0^2
\end{equation}
hits the branch cut, $\mathbb{R_{-}}$, i.e. find real values of $k \in [0, 2\pi)$ for which $\Im (E_+^2(k)) = 0$ and $\Re (E_+^2(k)) < 0$, since $E_+^2(k)$ is a smooth function of $k \in \mathbb{R}$ with $m_0, m_1 \in \mathbb{R}$. Using $e^{ik} = \cos k + i \sin k$, we have
\begin{equation}
\begin{aligned}
\Re (E_+^2(k)) &= (m_1 + 1) \left( 2 \cos^2 k + m_1 \cos k - 1 \right) - m_0^2, \\
\Im (E_+^2(k)) &= (m_1 + 1) \sin k \left( 2 \cos k + m_1 \right).
\end{aligned}
\end{equation}
Thus $\Im (E_+^2(k)) = 0$ holds if
\begin{itemize}
  \item $\sin k = 0 \quad \Rightarrow \quad k = 0, \pi$,
  
  \begin{itemize}
  \item At $k = 0$: $\Re (E_+^2(0)) = (m_1 + 1)^2 - m_0^2 < 0 \quad \Rightarrow \quad m_0 > |m_1 + 1| \ \vee \ m_0 < -|m_1 + 1|$,
  \item At $k = \pi$: $\Re (E_+^2(\pi)) = 1 - m_1^2 - m_0^2 < 0 \quad \Rightarrow \quad m_0 > \sqrt{1 - m_1^2} \ \vee \ m_0 < - \sqrt{1 - m_1^2}$.
\end{itemize}

  \item $2\cos k + m_1 = 0 \quad \Rightarrow \quad \cos k = -\frac{m_1}{2}$, which is only valid for $|m_1| \leq 2$.

  To enforce $\Re(E_+^2(k)) < 0$, we substitute $\cos k = -\frac{m_1}{2}$ into $\Re(E_+^2(k))$ to get
  
  \begin{equation}
  \Re(E_+^2(k)) = -(m_1 + 1 + m_0^2) < 0.
  \end{equation}
  Thus, these solutions always satisfy the condition.

  If $-2 \le m_1 < -1$, then $m_1 + 1 < 0$, and we have the inequality:
  \begin{equation}
     m_1>- m_0^2 - 1.
  \end{equation}
    
  In summary, for $|m_1| \le 2$ and $m_1 \ne -1$, the solutions with $2\cos k + m_1 = 0$ contribute to $\nu_E$ for all $m_0$ if $m_1 > -1$, while for $-2 \le m_1 < -1$ they contribute only when $m_1 > -m_0^2 - 1$.

  \item $m_1 + 1 =0 \quad \Rightarrow \quad m_1 = -1$

  In this case, $E_+^2(k) = -m_0^2$ is real and independent of $k$:
\begin{equation}
\Im (E_+^2(k)) \equiv 0, \quad
\Re (E_+^2(k)) = -m_0^2 < 0. 
\end{equation}

Therefore, all $k \in [0, 2\pi)$ are valid solutions except at $m_0 = 0$. However, when $m_1 = -1$, the function $R(k)$ becomes divergent, and we can no longer directly apply \cref{eq:GBBP/correct_GBBP_numerical} to compute the global biorthogonal Berry phase. In this special case, the Hamiltonian becomes upper-triangular, and its eigenvalues are simply given by its diagonal entries, independent of the off-diagonal terms. As a result, the energy spectrum no longer depends on $k$, and the global biorthogonal Berry phase becomes trivially zero.
\end{itemize}

\begin{figure}[t]
\centering
\includegraphics[width=0.88 \linewidth]{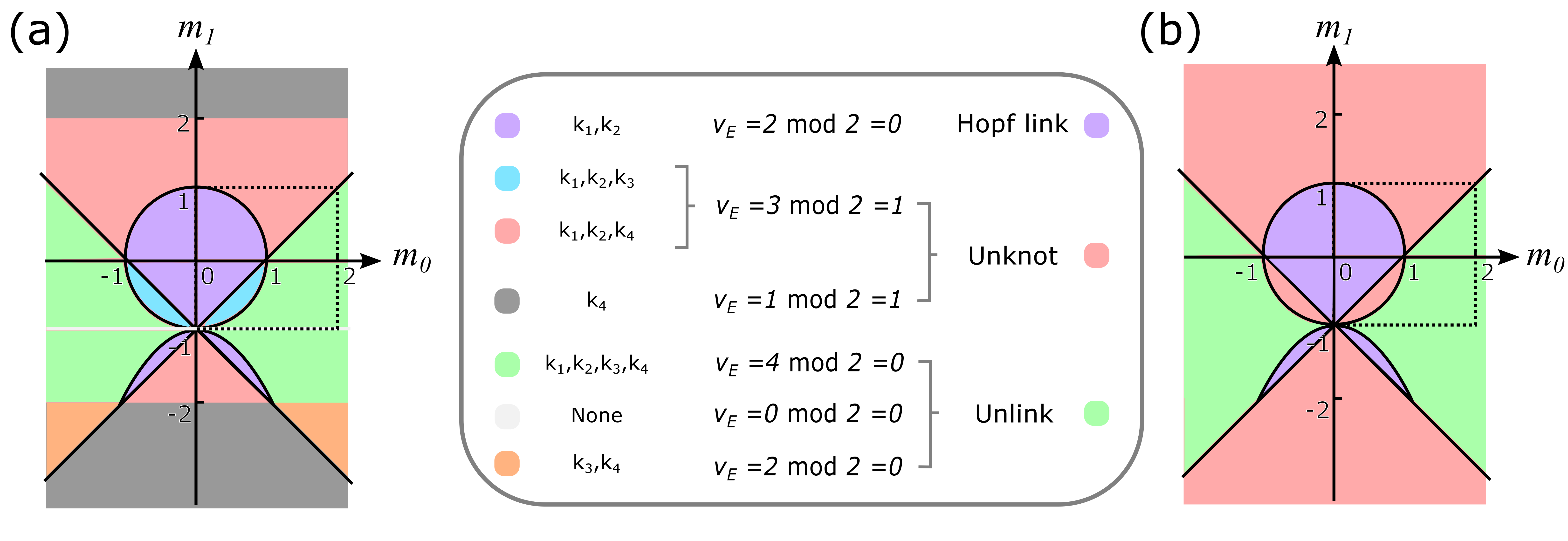}
\caption{\textbf{(a) Different regions denote the number of swapping points $\nu_E$ in $\hat{H}^{(2)}(k)$.} \textbf{(b) Complete phase diagram of $\hat{H}^{(2)}(k)$}. The dashed square denotes the phase diagram in Ref.~\cite{PhysRevLett.126.010401}. Observe that there are three topologically distinct knots/links, but the Hopf link and unlink both have $\nu_E = 0$ thus \cref{eq:GBBP/correct_GBBP} cannot distinguish them. Note that the biorthogonal Berry phase is given by $\gamma = \nu_E \mod 2$, shown as the numbers on the right side of the equations in each row.}
\label{fig:two-band_twister_model_phase_diagram_comparison}
\end{figure}

The full solution set for real $k$ satisfying $\Im(E(k)) = 0$ and $\Re(E(k)) < 0$ is:
\begin{itemize}
  \item For $|m_1| \leq 2$ and $m_1 \neq -1$:
  
  There are always two solutions, $k_1 = \cos^{-1}\left(-\frac{m_1}{2}\right)$, $k_2 = 2\pi - \cos^{-1}\left(-\frac{m_1}{2}\right)$;

  Except when $m_1 \le - m_0^2 - 1$:
  \begin{itemize}
    
  \item if $m_0 > |m_1 + 1| \ \vee \ m_0 < -|m_1 + 1| $, there is another solution $k_3 = 0$;
  
  \item if $m_0 > \sqrt{1 - m_1^2} \ \vee \ m_0 < - \sqrt{1 - m_1^2}$, there is another solution $k_4 = \pi$. 
  \end{itemize}

  \item For $m_1 = -1$: all $k \in [0, 2\pi]$ are solutions if $m_0 > 0$ (but this cannot provide the correct $\gamma$).
\end{itemize}

Furthermore, the following equations

\begin{equation}
\begin{aligned}
F'_1(m_0, m_1) &= (m_1 + 1)^2 - m_0^2 < 0, \\
F'_2(m_0, m_1) &= 1 - m_1^2 - m_0^2 < 0, \\
F'_3(m_0, m_1) &= -(m_1 + 1 + m_0^2) < 0,
\label{eq:GBBP/computing_nu_E/phase_boundaries}
\end{aligned}
\end{equation}
correspond to the regions as shown in \cref{fig:two-band_twister_model_phase_diagram_comparison}a. Note that even though this is a re-derivation of the boundaries in parameter space, it was done so to compute $\nu_E$. As such, the boundaries of the regions given by \cref{eq:GBBP/computing_nu_E/phase_boundaries} are exactly \cref{eq:two-band/phase_diagram/phase_boundaries}. 

The value of $\nu_E$ for different $(m_0, m_1)$ pairs can be directly inferred from the number of distinct $k$ points where the band swapping occurs, as shown in \cref{fig:two-band_twister_model_phase_diagram_comparison} which is in good agreement with \cref{fig: NH_E}. Here, we elaborate on how \cref{fig:two-band_twister_model_phase_diagram_comparison}(b) is constructed from \cref{fig:two-band_twister_model_phase_diagram_comparison}(a): the colored regions in (a) correspond to different sets of $k$ points at which the energy bands exchange positions. The topological invariant $\nu_E$ is simply the number of such $k$ points in each case. The solid lines in the diagram denote the boundaries, i.e., the loci where exceptional points emerge. If two regions in parameter space have different colors (indicating different sets of $k$ values) but share the same value of $\nu_E$, and one can move from one region to the other without crossing a boundary, then these regions belong to the same knot/link. An exception arises in the blue region; although it shares $\nu_E$ with other areas, it cannot be continuously connected to them without crossing a boundary. However, since it has an odd value of $\nu_E$, it must be the unknot. 

We emphasize that $\nu_E$ has information when a boundary is crossed but it is unable to characterize a knot/link since it is a $\mathbb{Z}_2$ quantity. As such, the winding number \cref{eq:characterizing_non-Hermitian_braids/winding_number_element} is necessary. 

\section{Estimating the Fidelity in Quantum Simulations}
\label{app:ibm_hardware_details}

Here, we estimate the hardware usage time and fidelity following the IBM Quantum workload usage guide. The baseline estimate is
\begin{equation}
t_{\mathrm{use}}\approx t_{0}+(\tau_{\mathrm{rep}}+\tau_{\mathrm{circ}})\,N_{\mathrm{exec}},
\end{equation}
where $t_{0}$ is the per-sub-job overhead, $\tau_{\mathrm{rep}}$ is the repetition delay, $\tau_{\mathrm{circ}}$ is the circuit duration, and
\begin{equation}
N_{\mathrm{exec}}=N_{\mathrm{circ}}\times N_{\mathrm{shots}}
\end{equation}
is the total number of circuit executions. For jobs without advanced mitigation and with the default repetition delay, IBM further provides a quick estimate
\begin{equation}
t_{\mathrm{use}}\approx 2+0.00035\,N_{\mathrm{exec}},
\end{equation}
with $t_{\mathrm{use}}$ in seconds \cite{IBM_workload_usage}. For the simulations reported in this work, each circuit is sampled with
$N_{\mathrm{shots}} = 40000$. For a single circuit per data point ($N_{\mathrm{circ}}=1$), this gives
\begin{equation}
t_{\mathrm{use}}\approx 2 + 0.00035\times 40000 \simeq 16~{\mathrm{s}}.
\end{equation}
Summing over all hardware executions reported in this work, we estimate a total device usage of approximately 8 hours.

In this work, all hardware executions are performed on the \texttt{ibm\_marrakesh} processor, and the relevant calibration metrics are summarized in \cref{tab:ibm_calibration_summary}. To obtain a depth budget for our circuits, we adopt a simple error accumulation model. Specifically, we take the median CZ process error to be $e_{\mathrm{CZ}}=2.508\times 10^{-3}$ and use a representative median single-qubit gate error $e_{1q}\approx 10^{-4}$. Let $n_{\mathrm{CZ}}$ and $n_{1q}$ denote the number of CZ and single-qubit gates applied \emph{per layer} in the transpiled circuit. In the weak-noise regime, the circuit fidelity after $L$ repeated layers can be approximated as an exponential decay

\begin{equation}
F(L)\approx \exp\!\Big[-L\big(n_{\mathrm{CZ}}e_{\mathrm{CZ}}+n_{1q}e_{1q}\big)\Big].
\end{equation}
For the $(2+1)$-qubit circuit (see \cref{fig:methodology_illustration}a in the main text), we have 3 effective layers of CZ gates in total (with 1 CZ gate in each layer) and 16 layers of single-qubit gates (with 2 single-qubit gates in each layer), so that the corresponding accumulated error exponent is
\begin{equation}
3e_{\mathrm{CZ}}+32e_{1q}
=
3\times 2.508\times 10^{-3}
+
32\times 10^{-4}
=
1.0724\times 10^{-2}.
\end{equation}
Accordingly, the estimated high fidelity is
\begin{equation}
F\approx
\exp\!\left[-\left(3e_{\mathrm{CZ}}+32e_{1q}\right)\right]
=
\exp\!\left(-1.0724\times 10^{-2}\right)
\approx 0.989.
\end{equation}

\begin{table}[t]
\centering
\begin{tabular}{lcccc}
\toprule
Metric & Median & Min & Max & Unit / Notes \\
\midrule
Readout length & 2584 & 2584 & 2584 & ns \\
Single-qubit gate length & 36 & 36 & 36 & ns \\
Two-qubit gate length (CZ connection) & 68 & 68 & 80 & ns \\
CZ error (median) & $2.508\times10^{-3}$ & -- & -- & two-qubit error \\
Readout error (median) & $1.12\times10^{-2}$ & -- & -- & assignment error \\
$T_1$ (median) & 187.71 & -- & -- & $\mu$s \\
$T_2$ (median) & 110.45 & -- & -- & $\mu$s \\
\bottomrule
\end{tabular}
\caption{\textbf{\texttt{ibm\textunderscore marrakesh} hardware details.} Calibration metrics and operation durations.}
\label{tab:ibm_calibration_summary}
\end{table}

\end{document}